\documentclass[twocolumn,prx,longbibliography,nofootinbib,superscriptaddress,aps]{revtex4-2}

\usepackage{graphicx} 
\usepackage{amsfonts}
\usepackage{amssymb}
\usepackage{amscd}
\usepackage{amsmath}    
\usepackage{amsthm}
\usepackage{appendix}
\usepackage{enumerate}
\usepackage{enumitem}
\usepackage{epsfig}
\usepackage{float}
\usepackage[colorlinks, linkcolor=red, anchorcolor=blue, citecolor=green]{hyperref}
\usepackage{makecell}
\usepackage{subfigure}
\usepackage{subfloat}
\usepackage{xcolor}		
\usepackage{physics}
\usepackage[most]{tcolorbox}
\usepackage{tabularx}
\usepackage{booktabs}
\usepackage{bm}

\usepackage[T1]{fontenc}
\usepackage{lmodern}
\usetikzlibrary{arrows.meta}
\definecolor{ink}{HTML}{263745}
\definecolor{blue}{HTML}{346C91}
\definecolor{teal}{HTML}{247970}
\definecolor{gold}{HTML}{A16B28}
\definecolor{comb}{HTML}{6861A8}
\definecolor{test}{HTML}{368476}

\usepackage{MnSymbol}
\usepackage{tikz}
\usetikzlibrary{quantikz2}
\usepackage{tikzpeople}
\usepackage{pgfplots}
\pgfplotsset{compat=1.18}
\usetikzlibrary{arrows}
\usetikzlibrary{shapes,fadings,snakes}
\usetikzlibrary{decorations.pathmorphing,patterns}
\usetikzlibrary{calc}
\usetikzlibrary{positioning}
\definecolor{amber}{rgb}{1.0, 0.75, 0.0}

\graphicspath{{./figure/}}

\newtheorem{theorem}{Theorem}

\newtheorem{lemma}{Lemma}
\newtheorem{corollary}{Corollary}

\newtheorem{definition}{Definition}
\newtheorem{proposition}{Proposition}

\newtheoremstyle{recap}
  {}                
  {}                
  {\itshape}        
  {}                
  {\bfseries}       
  {.}               
  {.5em}            
  {#3}              

\theoremstyle{recap}

\newcommand{\id}{\mathbb{I}}

\newcommand{\X}{\mathcal X}
\newcommand{\Haar}{\mu_{\mathrm{Haar}}}

\newcommand{\Img}{\operatorname{Im}}
\newcommand{\Dom}{\operatorname{Dom}}

\newcommand{\Adv}{\operatorname{Adv}}
\newcommand{\negl}{\operatorname{negl}}
\newcommand{\poly}{\operatorname{poly}}

\newcommand{\Gen}{\mathsf{Gen}}

\newcommand{\PFC}{\mathsf{PFC}}

\newcommand{\ipqTPRP}{\mathsf{ip\text{-}qTPRP}}

\newcommand{\PRComb}{\mathsf{PRComb}}
\newcommand{\td}{\Delta}
\newcolumntype{Y}{>{\raggedright\arraybackslash}X}

\newtcolorbox[auto counter]{mybox}[2][]{
	enhanced,
	breakable,
	colback=blue!5!white,
	colframe=blue!75!black,
	fonttitle=\bfseries,
	title=Box \thetcbcounter: #2,#1
}

\definecolor{daxcolor}{rgb}{0.2, 0.1, 1}
\definecolor{KBcolor}{rgb}{0.5, 0.1, 0.2}
\definecolor{atcolor}{RGB}{150,0,0}

\newcommand{\sit}{Engineering Cluster, Singapore Institute of Technology, 1 Punggol Coast Road, Singapore 828608, Republic of Singapore\looseness=-1}

\newcommand{\qinc}{Quantum Innovation Centre (Q.InC), Agency for Science, Technology and Research (A*STAR), 4 Fusionopolis Way, Kinesis \#05-01, Singapore 138635, Republic of Singapore\looseness=-1}

\newcommand{\sutd}{Science, Mathematics and Technology Cluster, Singapore University of Technology and Design, 8 Somapah Road, Singapore 487372, Republic of Singapore\looseness=-1}

\begin{document}
\title{Function-like pseudorandom unitaries generate pseudorandom quantum processes}

\author{Zitai Xu}
\affiliation{Joint Center for Quantum Information and Computer Science, University of Maryland, College Park, Maryland 20742, USA}
\affiliation{Department of Computer Science, University of Maryland, College Park, USA}

\author{Rahul Arvind}
\affiliation{Department of Physics, University of Texas at Austin, Austin, TX 78712, USA}


\author{Kishor Bharti}
\altaffiliation{Currently at IonQ Inc.}
\affiliation{Joint Center for Quantum Information and Computer Science, University of Maryland, College Park, Maryland 20742, USA}
\affiliation{UMIACS, University of Maryland, College Park, Maryland 20742, USA}

\author{Tobias Haug}
\affiliation{Quantum Research Center, Technology Innovation Institute, Abu Dhabi, United Arab Emirates}

\author{\mbox{Dax Enshan Koh}}
\affiliation{\sit}
\affiliation{\qinc}
\affiliation{\sutd}

\author{Wai-Keong Mok}
\affiliation{Institute for Quantum Information and Matter,
California Institute of Technology, Pasadena, CA 91125, USA}

\author{Andrew Tanggara}
\affiliation{Centre for Quantum Technologies, National University of Singapore, 3 Science Drive 2, Singapore 117543}
\affiliation{Nanyang Quantum Hub, School of Physical and Mathematical Sciences, Nanyang Technological University, Singapore 639673}

\author{Alexey V. Gorshkov} 
\affiliation{Joint Center for Quantum Information and Computer Science, University of Maryland, College Park, Maryland 20742, USA}
\affiliation{Joint Quantum Institute, NIST/University of Maryland, College Park, MD, 20742, USA}

\begin{abstract}
    Haar-random unitaries provide a canonical model of generic random quantum evolution, but they typically have exponential description and circuit complexity. Although pseudorandom unitaries efficiently emulate a single Haar-random unitary, many tasks require an entire reusable family of independently random-looking operations. We introduce pseudorandom function-like unitaries (PRFU), which generate such a family indexed by public labels using only a single short key. We distinguish classical- and coherent-label access and establish security under adaptive quantum queries. For classical labels, we provide a generic construction from a post-quantum pseudorandom function and a pseudorandom unitary. For coherent labels, we use an indexed path-recording framework to analyze interference across labels and prove security of a construction based on quantum-secure function and permutation primitives. We further show that PRFUs generate pseudorandom channels and quantum combs with private memory, which are secure against adaptive interventions and concurrent sessions. Combining PRFUs with unitary gluing yields a one-key family of pseudorandom unitaries whose supported register widths can be chosen after key generation. Further applications include nonce-resolved quantum authentication, coherently masked QRAM queries, and efficient emulation of random multi-time dynamics. Our results extend quantum pseudorandomness from individual operations to efficiently generated families of random-looking quantum dynamics.
\end{abstract}

\maketitle


\section{Introduction}
\label{sec:introduction}

Haar-random pure states capture the typical behavior of high-dimensional quantum systems~\cite{Page_1993}, connect to statistical mechanics through quantum typicality~\cite{vonNeumann1929,GoldsteinEtAl2006,PopescuShortWinter2006}, and model entanglement properties of highly excited eigenstates in chaotic quantum many-body systems~\cite{VidmarRigol2017}, while Haar-random unitaries model generic reversible dynamics and underlie decoupling~\cite{Dupuis2014}, information scrambling~\cite{Hayden2007}, randomized protocols, and quantum cryptography. 
A typical Haar state or unitary, however, has exponential description and circuit complexity. This tension has motivated two efficient surrogates: unitary designs, which reproduce finitely many Haar moments~\cite{DankertEtAl09,BrandaoHarrowHorodecki16,MPSY24,ChenEtAl24,MaHuang25}, and pseudorandom states and unitaries reproducing Haar behavior against efficient quantum observers~\cite{JiLiuSong18,MPSY24,BrakerskiYuen26,haug2023pseudorandom}. Though well suited to physical and information-processing tasks~\cite{Magesan2011,Huang2020shadow,liu2026scrambling,Gu2026EDP}, a fixed-order design does not, by itself, provide the security required in an adversarial setting, since an efficient observer may interact adaptively with the unknown object.

A pseudorandom unitary (PRU) supplies an efficiently implementable unitary $U_k$ that, for a uniformly sampled secret key $k$, is computationally indistinguishable from one persistent Haar-random unitary under adaptive oracle access. This single-instance guarantee is sufficient whenever an application requires one random-looking evolution. Many settings, however, naturally call for a family of such evolutions indexed by public context---for example, one reproducible unitary for each local patch, register width, nonce, coherent address, process instance, or time step. Sampling independent PRU keys can serve a predetermined finite collection, but managing a large and dynamically growing collection of secret keys is undesirable in practice. Rather than managing a separate secret key for every context, we seek to derive an entire adaptively addressable family of pseudorandom unitaries from one short master key. This one-key, many-object paradigm has already proved fruitful for quantum states: pseudorandom function-like states (PRFS) use a single key and a public label to generate a reproducible family of pseudorandom states~\cite{AnanthQianYuen22,AnanthGulatiQianYuen22,BatraEtAl26}. We develop the corresponding notion for quantum evolutions. For a finite label set $\X_\lambda$, where $\lambda$ denotes the security parameter, one short key $k$ specifies a family of efficiently implementable unitaries $\{G_k(x):x\in \mathcal{X}_{\lambda}\}$, indexed by public labels $x$. A query with label $x$ applies the corresponding $G_{k}(x)$ to a supplied quantum register, with repeated uses of the same label selecting the same unitary. To an efficient observer, access to this family should be computationally indistinguishable from a random unitary-valued function $x\mapsto H_x$, which assigns an independently sampled Haar random unitary $H_{x}$ to each label $x$ persistently across queries. We call such a keyed family $\{G_{k}(x)\}_{k}$ a pseudorandom function-like unitary (PRFU).

Having identified the target primitive, we study its realization and security under adaptive oracle access. Throughout, a distinguisher may maintain quantum memory across queries, apply intermediate quantum operations, and choose later queries using its complete previous interaction. The resulting security notion depends essentially on how the label is accessed. With classical-label access, only the label register is classical. The data, reference, and adversarial memory remain quantum. With coherent-label access, the label is supplied in a quantum register $L$, which may be in superposition itself and even entangled with other registers. A query to the family then applies the controlled unitary
\begin{equation}
W_k:=\sum_{x\in\X_\lambda}\ketbra{x}{x}_L\otimes G_k(x),
\label{eq}
\end{equation}
where each $G_{k}(x)$ acts on the supplied data register. A single coherent call can involve every label while preserving interference between them, making relative global phases and cross-label correlations observable. We provide separate constructions and security proofs for the classical- and coherent-label settings, using different proof techniques and cryptographic assumptions.

PRFUs also provide a natural route from pseudorandom unitaries to more general quantum dynamics. Initializing an ancillary environment, applying a label-indexed unitary, and discarding part of the output produces a labeled family of quantum channels, connecting PRFUs with established random-channel ensembles~\cite{KukulskiEtAl21}. In a multi-time experiment, an interactive distinguisher may intervene between successive transformations while the process retains inaccessible memory. The resulting behavior is not captured by the endpoint channel or its Choi state, which hides the intermediate intervention slots, but is naturally represented by a quantum comb or quantum strategy~\cite{ChiribellaEtAl08,ChiribellaEtAl09,GutoskiWatrous07,Gutoski12,PollockEtAl18}. These extensions allow us to describe pseudorandom quantum processes with both inaccessible environments and internal memory, while preserving the ability to select reproducible processes through public labels. An efficient observer interacting with these processes should be unable to distinguish them from their counterparts built from independent Haar-random unitaries.

Our contributions are fourfold. First, we construct a classical-label PRFU by using a post-quantum pseudorandom function (PRF) $F_k$ to derive the seed of an inner PRU family $\{V_a\}_a$:
\begin{equation}
G_k(x):=V_{F_k(x)}.
\label{eq:intro-classical-construction}
\end{equation}
Our reduction handles adaptively chosen labels in order of first occurrence, requiring at most $q$ hybrid transitions for $q$ queries, independently of the label-domain size. The key ingredient is an efficiently implementable Haar proxy, obtained unconditionally from amplified relative-error permutation-function-Clifford (PFC)
designs~\cite{MPSY24}, which allows the intermediate hybrids to be simulated efficiently. This establishes classical-label forward security from a post-quantum PRF and an $n(\lambda)$-qubit adaptive forward PRU whenever $n(\lambda)=\omega(\log\lambda)$.
The argument also extends to strong security, in which the adversary is additionally given oracle access to the inverse unitary, 
when the inner PRU is strong.

Second, we construct coherent-label PRFUs from the PFC construction, using a quantum-secure pseudorandom function (qPRF), assuming the existence of a tweakable in-place pseudorandom permutation. Our analysis uses indexed path recording~\cite{MaHuang25,FoxmanEtAl26} to account for interference between query histories involving different labels. This bounds an observer's advantage in distinguishing the ideal indexed PFC ensemble from independent Haar-random unitaries by $2q(q-1)/(2^n+1)$ for query number $q<2^n$, yielding coherent-label forward security when the PRFUs act on $n(\lambda)=\omega(\log\lambda)$ qubits.


Third, we use PRFUs to construct pseudorandom quantum combs, which describe sequences of quantum operations connected by an internal memory. Assigning distinct labels to successive operations makes the resulting process computationally indistinguishable from one built from independent Haar-random unitaries, even when an observer adaptively intervenes between operations while the internal memory remains inaccessible. Adding a public process label allows one key to specify an entire family of such combs, with security preserved under interleaved interactions with different members of the family.


Finally, we combine PRFUs with unitary gluing~\cite{SchusterHaferkampHuang2025,MaHuang25} to construct a one-key family of variable-width PRUs, allowing supported polynomially bounded register widths to be chosen after key generation. The construction uses public labels to derive local operations that jointly appear as independent Haar-random unitaries, enabling their composition across different register widths under a single key. This extends the fixed-register functionality of a single-instance PRU without requiring additional secret keys. The same flexibility supports nonce-resolved quantum authentication~\cite{HaugBansalMokKohBharti2025}, coherently masked QRAM, and the emulation of random multi-time dynamics, with labels distinguishing authentication contexts, memory addresses, or stages of evolution.



Figure~\ref{fig:overview} summarizes how one secret key generates a labeled family of pseudorandom unitaries, the constructions for the two access models, and their use in generating pseudorandom quantum processes and enabling cryptographic applications. Following this structure, Section~\ref{sec:definitions} defines the access models. Sections~\ref{sec:classical-construction} and~\ref{sec:coherent-construction} give the constructions. Sections~\ref{sec:processes}--\ref{sec:discussion} develop the process liftings, applications, and limitations. Complete proofs and further process consequences appear in the appendices.

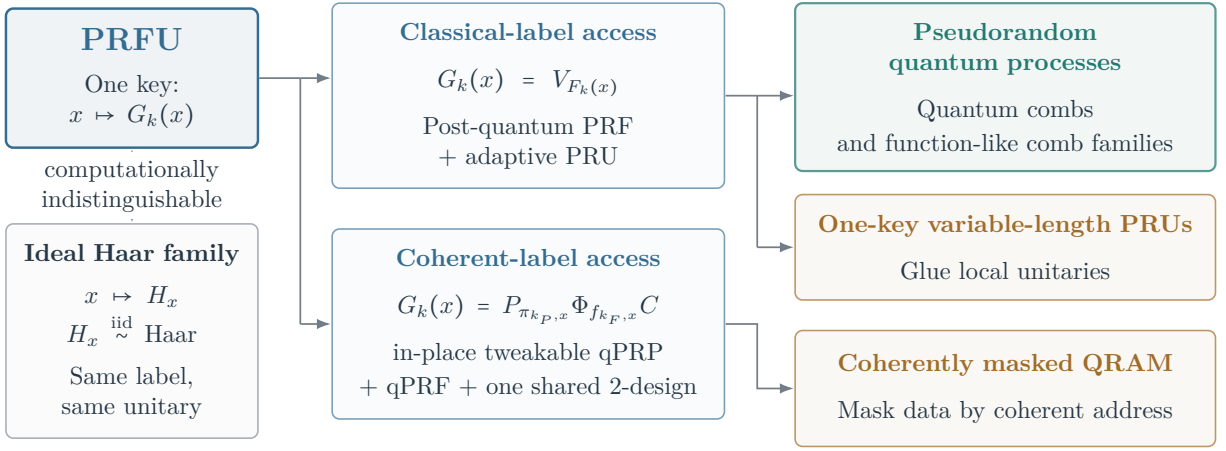
\begin{figure*}[htbp]
    \centering
    \resizebox{0.9\textwidth}{!}{
\begin{tikzpicture}[x=1cm,y=1cm,text=ink,
  font=\fontsize{10}{12}\selectfont,
  arr/.style={-{Latex[length=2mm]},draw=ink!65,line width=.7pt},
  line/.style={draw=ink!65,line width=.7pt},
  panel/.style={draw=blue!65,fill=blue!3,rounded corners=3pt,
    line width=.65pt,align=center,inner xsep=7pt,inner ysep=8pt,
    outer sep=.5pt},
  leftpanel/.style={panel,text width=3.1cm},
  middlepanel/.style={panel,text width=5.1cm,minimum height=2.15cm},
  rightpanel/.style={panel,text width=5.5cm}]

\node[leftpanel,draw=blue,line width=1pt,fill=blue!8,
  minimum height=1.4cm] (prfu) at (1.85,4.65)
  {{\color{blue}\fontsize{13}{15}\selectfont\textbf{PRFU}}\\[5pt]
   One key: $x\mapsto G_k(x)$};

\node[leftpanel,draw=ink!40,fill=ink!2] (ideal) at (1.85,1.05)
  {\textbf{Ideal Haar family}\\[5pt]
   $x\mapsto H_x$\\[2pt]
   $H_x\overset{\mathrm{iid}}{\sim}\mathrm{Haar}$\\[5pt]
   Same label, same unitary};

\draw[{Latex[length=1.5mm]}-{Latex[length=1.5mm]},
  dashed,draw=ink!55,line width=.65pt]
  (ideal.north)--(prfu.south)
  node[midway,align=center,fill=white,inner sep=3pt]
  {computationally\\indistinguishable};

\node[middlepanel] (classical) at (7.5,4.4)
  {{\color{blue}\textbf{Classical-label access}}\\[7pt]
   $G_k(x)=V_{F_k(x)}$\\[7pt]
   Post-quantum PRF + adaptive PRU};

\node[middlepanel] (coherent) at (7.5,1.15)
  {{\color{blue}\textbf{Coherent-label access}}\\[7pt]
   $G_k(x)=P_{\pi_{k_P,x}}\Phi_{f_{k_F,x}}C$\\[7pt]
   in-place tweakable qPRP\\[2pt]
   + qPRF + one shared $2$-design};

\draw[line] (prfu.east)--(4.25,4.65);
\draw[line] (4.25,1.15)--(4.25,4.65);
\draw[arr] (4.25,4.65)--(classical.west |- prfu.east);
\draw[arr] (4.25,1.15)--(coherent.west);

\node[rightpanel,draw=teal!75,fill=teal!6,line width=.9pt,
  minimum height=1.95cm,anchor=north] (processes)
  at ([xshift=6.8cm]classical.north)
  {{\color{teal}\textbf{Pseudorandom quantum processes}}\\[7pt]
   Quantum combs\\[2pt]
   and function-like comb families};

\node[rightpanel,draw=gold!70,fill=gold!5,
  minimum height=1.5cm,anchor=north] (variable)
  at ([yshift=-3mm]processes.south)
  {{\color{gold}\textbf{One-key variable-length PRUs}}\\[7pt]
   Glue local unitaries};

\node[rightpanel,draw=gold!70,fill=gold!5,
  minimum height=1.65cm,anchor=north] (qram)
  at ([yshift=-4mm]variable.south)
  {{\color{gold}\textbf{Coherently masked QRAM}}\\[7pt]
   Mask data by coherent address};

\coordinate (applicationbranch) at ([xshift=4.5mm]classical.east);
\draw[line] (classical.east)--(applicationbranch);
\draw[line] (applicationbranch)--(applicationbranch |- variable.west);
\draw[arr] (applicationbranch)--(processes.west |- classical.east);
\draw[arr] (applicationbranch |- variable.west)--(variable.west);
\draw[arr] (coherent.east)--([xshift=4.5mm]coherent.east)
  |- (qram.west);
\end{tikzpicture}}
    \caption{Overview of PRFU constructions and selected applications. A single key specifies a unitary family that emulates independent Haar unitaries, with repeated labels reusing the same unitary. The two branches give constructions for classical-label access (Sec.~\ref{sec:classical-construction}) and coherent-label access (Sec.~\ref{sec:coherent-construction}). Retaining private memory yields pseudorandom combs and function-like comb families (Sec.~\ref{sec:processes}); gluing local unitaries yields one-key variable-length PRUs (Sec.~\ref{subsec:variable-width-pru}); coherent control by memory addresses enables masked QRAM queries~(Sec.~\ref{subsec:coherent-qram}).}
    \label{fig:overview}
\end{figure*}

\section{Function-like Haar random ensemble and access models}
\label{sec:definitions}

In a Hilbert space $\mathcal{H}_{A}$, one can define the unique ensemble of random unitaries, namely the Haar random unitaries: one draws a random unitary $U\in U(2^{n})$ from the $n$-qubit unitary group according to the Haar measure $\Haar$. In this work, we will consider a labeled family of random unitaries such that each input label $x$ determines a unitary that looks independently Haar random as defined below.

\begin{definition}[Function-like Haar random family]\label{def:haar-family}
Let $\X\subseteq\{0,1\}^{d}$ be a finite label set and $A$ be an $n$-qubit data register. Sample 
\begin{equation}
\label{eq:product-haar}
\bm {H}=(H_x)_{x\in\X},\qquad
 H_{x}\overset{\mathrm{iid}}{\leftarrow}\Haar(\mathrm{U}(\mathcal{H}_{A})).
\end{equation}
The same $H_x$ is used on every occurrence of $x$, while $H_x$ and $H_y$ are independent for $x\ne y$. We call $\bm{H}$ a function-like Haar random family.
\end{definition}
The case $|\X|=1$ corresponds to the Haar random unitaries. For larger $|\X|$, the family $\mathbf H$ determines a map $x\mapsto H_x$, whose values are drawn independently from Haar measure. The independence across labels is essential: distinct labels have no correlations and, under exact Haar sampling, coincide only with probability zero. The label set $\X$ allows for a reproducible collection of random evolutions indexed by a public value, which will be useful in many scenarios. However, even sampling from a Haar random ensemble could be prohibitively expensive; it is therefore desirable to have an efficient implementation that remains computationally indistinguishable from the function-like Haar random family, which we term as a pseudorandom function-like unitary (PRFU).

We now parameterize these objects by the security parameter $\lambda$. Let $\X_\lambda\subseteq\{0,1\}^{d(\lambda)}$ be a finite label set, and let $A$ be an $n(\lambda)$-qubit register. A keyed unitary family consists of an efficient probabilistic key generator $k\leftarrow\mathsf{Gen}(1^\lambda)$, producing a key of polynomial length in $\lambda$, and an efficient procedure that, given $k$ and $x\in\mathcal X_\lambda$, implements the unitary $G_k(x)$ on $A$. Thus, each key $k$ determines a map
\begin{equation}
G_k:\mathcal X_\lambda\longrightarrow \mathrm U(\mathcal H_A), \qquad x\longmapsto G_k(x).
\end{equation}
Its value at a label \(x\) is a unitary acting on \(A\), rather than a classical string.

In the security game, there exists a challenger who samples a bit $b\in\{0,1\}$. Depending on the value of $b$, the challenger either samples the function-like Haar random family $\bm H$ or samples from a pseudorandom $G_{k}$. The security is established when any nonuniform quantum polynomial-time (QPT) adversary can distinguish these two cases with at most negligible distinguishing advantage. 

Depending on the ability of the adversary, different notions of security can be defined. In this work we will assume that the adversary has oracle access to the pseudorandom function-like unitary. The oracle $\mathcal{O}_{G_{k}}$ consists of two input registers, the first is a label $x\in\X_{\lambda}$ and the second is an $n(\lambda)$-qubit quantum register denoted as $\mathcal{H}_{A}$. 
The adversary may supply an arbitrary state $\rho_{AR}$ on the data register $A$ and its own reference register $R$, with the oracle acting only on data register $A$ as
\begin{equation}\label{eq:classical-oracle-access}
    \mathcal{O}_{G_{k}}(x,\rho_{AR})=(G_{k}(x)\otimes I_{R})\rho_{AR}(G_{k}^{\dagger}(x)\otimes I_{R}).
\end{equation}
Similarly, the oracle for the true function-like Haar family is
\begin{equation}\label{eq:true-classical-oracle-access}
    \mathcal{O}_{\bm H}(x,\rho_{AR})=(H_{x}\otimes I_{R})\rho_{AR}(H_{x}^{\dagger}\otimes I_{R}).
\end{equation}
We allow the adversary to have adaptive queries to the PRFU. More specifically, the adversary is allowed to choose the input label $x$ based on previous observations, and may even entangle the input state $\rho_{AR}$ between rounds. In the following, we establish the security notion in this context.

\begin{definition}[Classical-label adaptive PRFU]\label{def:classical-prfu}
    An efficient keyed family $\{G_{k}(x)\}_{k}$ is a classical-label PRFU if, for every non-uniform QPT adversary $\mathcal{A}$ making polynomially many oracle calls, 
    \begin{equation}\label{eq:classical-adaptive-security}
        \left|\Pr_{k}\!\left[\mathcal{A}^{\mathcal{O}_{G_{k}}}=1\right]-\Pr_{\bm H}\!\left[\mathcal{A}^{\mathcal{O}_{\bm H}}=1\right]\right|\leq\negl(\lambda).
    \end{equation}
    Here for PRFU the average is over $k\leftarrow\Gen(1^\lambda)$ while the truly random $\bm H=(H_x)_{x\in\X_{\lambda}}$ is sampled as in Definition~\ref{def:haar-family} with $\mathcal{X}=\mathcal{X}_{\lambda}$.
    We call the family strong if both experiments also provide inverse queries.
\end{definition}

Classical-label access does not mean that the adversary is classical. Only the label wire is classical. The data and reference registers remain quantum, and the adversary may retain arbitrary quantum coherence across all other registers and across successive rounds.

We will further consider a security model where the adversary can also input a quantum superposition on the label register. In this access model, each label $x$ coherently selects the unitary $G_k(x)$ in the keyed experiment, or the corresponding unitary $H_x$ from the sampled family $\mathbf H=(H_x)_{x\in\mathcal X_\lambda}$ in the ideal experiment. The resulting controlled unitaries are
\begin{equation}\label{eq:classical-controled-unitary}
\begin{split}
    W_{\mathrm{Haar}}&=\sum_{x\in\X_{\lambda}}\ketbra{x}{x}_{L}\otimes H_{x},\\
    W_{k}&=\sum_{x\in\X_{\lambda}}\ketbra{x}{x}_{L}\otimes G_{k}(x).
\end{split}
\end{equation}
The coherent-access PRFU oracle can then be defined as follows:
\begin{equation}\label{eq:coherent-oracle-access}
    \mathcal{O}_{\ket{G_{k}}}(\rho_{LAR})=(W_{k}\otimes I_{R})\rho_{LAR}(W_{k}^{\dagger}\otimes I_{R}).
\end{equation}
Similarly the coherent-access truly function-like Haar random oracle:
\begin{equation}\label{eq:coherent-haar-oracle-access}
    \mathcal{O}_{\ket{\bm H}}(\rho_{LAR})=(W_{\mathrm{Haar}}\otimes I_{R})\rho_{LAR}(W_{\mathrm{Haar}}^{\dagger}\otimes I_{R}).
\end{equation}
In the above definitions, $L$ is the quantum register for input labels, with Hilbert space $\mathcal{H}_{L}=\mathrm{span}\{\ket{x}:x\in\mathcal{X}_{\lambda}\}$, allowing the adversary to query in a superposition of labels. For example, one query of $\mathcal{O}_{\ket{G_{k}}}$ on
\begin{equation}
    \frac{1}{\sqrt{|\mathcal{X}_{\lambda}|}}\sum_{x\in\X_{\lambda}}\ket{x}_{L}\ket{0^{n}}_{A}
\end{equation}
simultaneously places every $G_{k}(x)\ket{0^{n}}$ in superposition. Under the coherent-access model, we can formulate the notion of coherent-label secure PRFU.

\begin{definition}[Coherent-label adaptive PRFU]\label{def:coherent-prfu}
    An efficient keyed family $\{G_{k}(x)\}_{k}$ is a coherent-label PRFU if, for every non-uniform QPT adversary $\mathcal{A}$ making polynomially many oracle calls, 
    \begin{equation}\label{eq:coherent-adaptive-security}
        \left|\Pr_{k}\!\left[\mathcal{A}^{\mathcal{O}_{\ket{G_{k}}}}=1\right]-\Pr_{\bm H}\!\left[\mathcal{A}^{\mathcal{O}_{\ket{\bm H}}}=1\right]\right|\leq\negl(\lambda).
    \end{equation}
\end{definition}

Classical-label access can be viewed as coherent-label access with complete dephasing of the label register in the computational basis before each query. This removes coherence between distinct labels while preserving the conditional quantum states of the data and reference registers. Since this simulation applies to both the keyed family and the ideal Haar family throughout an adaptive interaction, security under coherent-label access directly implies security under classical-label access.

The strong coherent-label variant additionally exposes $W_k^\dagger$ in the real experiment and $W_{\mathrm{Haar}}^\dagger$ in the ideal experiment, with classical direction choices. 

\begin{proposition}[Access hierarchy]\label{prop:hierarchy}
Coherent-label forward security implies classical-label forward security, and
the analogous implication holds for strong access. Strong security implies
forward security.
\end{proposition}

The first implication follows from the dephasing simulation above, while the second follows by restricting the adversary to forward queries. However, the inverse statements are not necessarily true. Security under a weaker access model does not by itself establish security under a stronger one. Our constructions and reductions therefore specify explicitly whether labels may be queried coherently and whether inverse queries are available.

\section{Construction and security with classical labels}
\label{sec:classical-construction}


We first consider classical-label access. In this setting, a query can introduce at most one new label, which allows us to organize the security proof by the order in which distinct labels first appear. We give a generic construction from a post-quantum PRF and a PRU, and then instantiate the inner PRU using the PFC construction.

\subsection{Generic construction from a PRF and a PRU}

We begin with a generic construction in which a post-quantum PRF derives, from each public label, the seed of an underlying PRU.
Let $F_k:\{0,1\}^{d(\lambda)}\rightarrow\{0,1\}^{s(\lambda)}$ be a post-quantum secure pseudorandom function, which is secure against quantum distinguishers making classical adaptive queries to its function oracle. Let $\{V_a\}_{a\in\{0,1\}^{s(\lambda)}}$ be an efficiently computable adaptive forward PRU family on $A$, with $a$ sampled uniformly. We require the map $a\mapsto V_a$ to be deterministic. Any randomness used to generate the circuit is included in $a$. The notation $V_a$ is reserved for this inner PRU, whereas $G_k(x)$ denotes the function-like family constructed below.

For $k\leftarrow\Gen_F(1^\lambda)$ and $x\in\X_\lambda\subseteq\{0,1\}^{d(\lambda)}$, define
\begin{equation}
G_k(x):=V_{F_k(x)}.
\label{eq:classical-construction}
\end{equation}

The adversary receives only the oracle $\mathcal O_{G_k}$ from Eq.~\eqref{eq:classical-oracle-access}. Neither the PRF value $F_k(x)$ nor a circuit description for the selected inner unitary is revealed. For the construction in Eq.~\eqref{eq:classical-construction}, this oracle is
\begin{equation}
\mathcal O_{G_k}(x,\rho_{AR})
=
\bigl(V_{F_k(x)}\otimes I_R\bigr)\rho_{AR}
\bigl(V_{F_k(x)}^\dagger\otimes I_R\bigr).
\label{eq:constructed-classical-oracle}
\end{equation}
Thus every repeated occurrence of a classical label invokes the same persistent inner unitary, while the PRF derives the inner seed associated with each label.

For a classical-label adversary $\mathcal A$, define the forward distinguishing advantage of the construction by 
\begin{align}
\Adv_G^{\mathsf{cl},\rightarrow}(\mathcal A;\lambda)&:=\bigl|\Pr_k\!\left[\mathcal A^{\mathcal O_{G_k}}(1^\lambda)=1\right]\nonumber\\[-1mm]
&\qquad-\Pr_{\bm H}\!\left[\mathcal A^{\mathcal O_{\bm H}}(1^\lambda)=1\right]\bigr|.
\label{eq:classical-construction-advantage}
\end{align}
We use the superscript $\to$ to denote forward-only oracle access. Here and below, the key $k$ is sampled according to key generator $\Gen_F(1^\lambda)$ and $\bm H$ sampled from the function-like Haar family as in Definition~\ref{def:haar-family}. For a classical-query PRF distinguisher $\mathcal B_F$, define the distinguishing advantage to be
\begin{align}
\Adv_F^{\mathsf{prf}}(\mathcal B_F;\lambda)&:=
\bigl|\Pr_k\!\left[\mathcal B_F^{F_k}(1^\lambda)=1\right]\nonumber\\[-1mm]
&\qquad-\Pr_R\!\left[\mathcal B_F^R(1^\lambda)=1\right]\bigr|,
\label{eq:pqprf-advantage}
\end{align}
where $R$ is uniform over $\operatorname{Func}(\{0,1\}^{d(\lambda)},\{0,1\}^{s(\lambda)})$. For the inner family, write $\mathcal O_{V_a}(\rho_{AR})=(V_a\otimes I_R)\rho_{AR}(V_a^\dagger\otimes I_R)$ and set the forward distinguishing advantage of PRU for distinguisher $\mathcal{B}$ as
\begin{align}
\Adv_V^{\mathsf{pru},\rightarrow}(\mathcal B;\lambda)&:=
\bigl|\Pr_a\!\left[\mathcal B^{\mathcal O_{V_a}}(1^\lambda)=1\right]\nonumber\\[-1mm]
&\qquad-\Pr_H\!\left[\mathcal B^{\mathcal O_H}(1^\lambda)=1\right]\bigr|,
\label{eq:inner-pru-advantage}
\end{align}
where $a\leftarrow\{0,1\}^{s(\lambda)}$ and $H\leftarrow\Haar(\mathrm{U}(\mathcal H_A))$. The superscript $\pm$ analogously denotes the advantage when forward and inverse queries are both allowed, with a classical direction choice and one total query budget.


In our reduction, we will need an efficiently implementable ensemble that is statistically close to the Haar random ensemble up to bounded query number. The following definition formulates this idea as a Haar proxy.

\begin{definition}[Finite-query Haar proxy]\label{def:proxy}
For a unitary $U$ on $A$, let $\mathcal O_U(\rho_{AR})=(U\otimes I_R)\rho_{AR}(U^\dagger\otimes I_R)$. An efficiently samplable circuit ensemble $\mathcal D_{q,\eta}$ is a forward $(q,\eta)$ Haar proxy if, for every possibly unbounded oracle distinguisher $\mathcal A$ making at most $q$ forward calls to one persistent challenge oracle, the experiments with $D\leftarrow\mathcal D_{q,\eta}$ and $H\leftarrow\Haar(\mathrm{U}(\mathcal H_A))$, respectively, satisfy
\begin{equation}
\left|
\Pr_D\!\left[\mathcal A^{\mathcal O_D}=1\right]-\Pr_H\!\left[\mathcal A^{\mathcal O_H}=1\right]\right|\leq\eta.
\label{eq:haar-proxy}
\end{equation}
\end{definition}

Between oracle calls, $\mathcal A$ may apply arbitrary quantum operations to the oracle register, an arbitrary reference system, and its private memory. It may also interact with side processes sampled independently of the challenge, although its subsequent operations may correlate all these systems.


The probability for an efficient adversary distinguishing between PRFU and the function-like Haar ensemble can be bounded by the following theorem.

\begin{theorem}[Classical-label lifting]\label{thm:classical-lifting}
For every nonuniform QPT adversary $\mathcal A$ making at most $q=q(\lambda)$ classical-label forward calls to $\mathcal O_{G_k}$, let $\mathcal D_{q,\eta}$ be a forward $(q,\eta)$ Haar proxy. There exist a QPT classical-query PRF distinguisher $\mathcal B_F$ and QPT single-instance PRU distinguishers $\mathcal B_1,\ldots,\mathcal B_q$ such that
\begin{align}
\Adv_G^{\mathsf{cl},\rightarrow}(\mathcal A;\lambda)
&\leq
\Adv_F^{\mathsf{prf}}(\mathcal B_F;\lambda)
+\sum_{j=1}^{q}\Adv_V^{\mathsf{pru},\rightarrow}(\mathcal B_j;\lambda)\nonumber\\[-1mm]
&\quad+2q\eta.
\label{eq:classical-bound}
\end{align}
The distinguisher $\mathcal B_F$ makes at most $q$ classical function queries, and each $\mathcal B_j$ makes at most $q$ forward calls to its single-unitary challenge. Consequently, assuming the post-quantum security of $F$ and the adaptive forward security of $V$, if $\eta(\lambda)$ is negligible, then $\operatorname{Adv}^{\mathrm{cl},\rightarrow}_{G}(\mathcal A;\lambda)$ is negligible for every such adversary $\mathcal A$.

\end{theorem}

Equivalently, when the proxy error $\eta$ is negligible, any QPT adversary that distinguishes the constructed PRFU with non-negligible advantage can be used to construct a QPT distinguisher that breaks the security of either the underlying PRF or PRU.

Before discussing the existence of such Haar proxy, we first explain the hybrid underlying Theorem~\ref{thm:classical-lifting}. The first step is to replace the outer PRF by a uniformly random function. Instead of genuinely sampling from a uniformly random function, the proof maintains a lazy table indexed by the distinct labels in their order of being queried. Since an efficient adversary can make at most $q=\poly(\lambda)$ queries, this table can be efficiently implemented.

Then we build intermediate hybrids by replacing the PRU by the Haar proxy for labels appearing in the lazy table one-by-one. Since all Haar proxies are efficiently implementable, any adversary $\mathcal{A}$ that distinguishes between $\mathcal{O}_{G_k}$ and $\mathcal{O}_{\bm{H}}$ can be used to build distinguishers $\mathcal{B}_{j}$ for PRU and Haar random unitary. By the security of PRU, each hybrid is negligibly close to each other. After substituting Haar proxy for all PRU, the last hybrid is to simply replace Haar proxy by the true Haar random ensemble. The statistical distance in Definition~\ref{def:proxy} then implies the closeness of this last hybrid. 



The proxy is therefore a resource of the reduction, not a component of the family $G_k$. Its circuit, randomness, query bound, and accuracy parameter do not appear in the construction in  Eq.~\eqref{eq:classical-construction}, but appear only in the final security bound in Eq.~\eqref{eq:classical-bound}. For a fixed $q$, its information-theoretic guarantee is stronger than the computational guarantee of a PRU, but the proxy is allowed to depend on $q$ and need not remain Haar-like beyond those $q$ calls. This is analogous to the distinction between $q$-wise independent functions and pseudorandom functions. The query-bounded relaxation permits unconditional efficient constructions from relative-error unitary designs and hence introduces no additional cryptographic hardness assumption.

\begin{proposition}[Efficient forward Haar proxies]\label{prop:efficient-haar-proxies}
Let $n,q\geq1$ satisfy $C_0q\leq2^{n/4}$ where $C_0\geq1$ is a universal constant. And let $0<\eta<1$ be the proxy error. Then there is an efficiently samplable and implementable ensemble of $n$-qubit circuits that is a forward $(q,\eta)$ Haar proxy. Its sampling time and circuit size are in $\poly(n,q,\log(1/\eta))$. 
\end{proposition}

Proposition~\ref{prop:efficient-haar-proxies} follows from the amplified relative-error PFC designs of Ref.~\cite{MPSY24}. More precisely, the proxy is an $m$-fold product of independent, efficiently derandomized PFC layers, with $m=O(q+\log(1/\eta)/n)$. Appendix~\ref{app:relative-design-proxy} gives the ensemble explicitly.

After establishing the unconditional Haar proxy, the security of classical-label adaptive PRFU follows from standard cryptography assumptions.

\begin{corollary}[Classical-label adaptive PRFU]\label{cor:classical-existence}
Suppose that $F$ is a post-quantum secure PRF, $V$ is an adaptive forward PRU, and $n(\lambda)=\omega(\log\lambda)$. Then the family $G$ defined in Eq.~\eqref{eq:classical-construction} is a classical-label adaptive forward PRFU according to Definition~\ref{def:classical-prfu}.
\end{corollary}

For each adversary, the proxy in the proof of Corollary~\ref{cor:classical-existence} can be selected after fixing the upper bound for the number of queries of the adversary. Proposition~\ref{prop:efficient-haar-proxies} may therefore be instantiated with this query bound and, for example, error $\eta=2^{-\lambda}$. The PRFU $G_{k}$ remains fixed and is independent of both choices.

\begin{corollary}[Conditional strong extension]\label{cor:classical-strong}
Suppose that $F$ is a post-quantum secure PRF, $V$ is a strong adaptive PRU, and $n(\lambda)=\omega(\log\lambda)$. Then the family $G$ defined in Eq.~\eqref{eq:classical-construction} is a classical-label adaptive strong PRFU.
\end{corollary}

To extend the argument to strong security, we require a bidirectional Haar proxy that remains statistically indistinguishable from Haar under adaptive queries to both the unitary and its inverse. Such proxies also follow from the amplified-design construction of~\cite{MPSY24}: sufficiently small diamond error implies statistical indistinguishability under adaptive forward and inverse queries, with the same asymptotic efficiency. Appendix~\ref{app:bidirectional-proxy} gives the conversion and completes the proof of Corollary~\ref{cor:classical-strong}.

\subsection{Instantiation from the PFC ensemble}
\label{subsec:classical-pfc}

We now instantiate the inner PRU in Eq.~\eqref{eq:classical-construction} by the PFC construction~\cite{MaHuang25}. The outer construction and its first-occurrence proof remain unchanged. Let $N=2^{n(\lambda)}$. For a permutation $\pi:\{0,1\}^{n(\lambda)}\rightarrow\{0,1\}^{n(\lambda)}$ and a Boolean function $h:\{0,1\}^{n(\lambda)}\rightarrow\{0,1\}$, define the following unitary operators, with both sums ranging over $z\in\{0,1\}^{n(\lambda)}$:
\begin{align}
P_\pi&:=\sum_z\ket{\pi(z)}\!\bra{z},\nonumber\\
\Phi_h&:=\sum_z(-1)^{h(z)}\ketbra{z}{z}.
\label{eq:pfc-operators}
\end{align}
The ideal permutation--function--Clifford ($\PFC$) ensemble is $P_\pi\Phi_h C$, where $\pi$ and $h$ are sampled uniformly from all permutations and Boolean functions, respectively, with the domains specified above and $C$ is an independent unitary $2$-design, for example a uniformly random Clifford~\cite{MPSY24}.

For the efficient version, we replace the random permutation and Boolean function by efficiently computable pseudorandom families. Let \(\{\pi_{a_P}\}\) be a strong quantum-secure pseudorandom permutation family with key length \(s_P(\lambda)\), and let \(\{h_{a_\Phi}\}\) be a quantum-secure pseudorandom function family with key length \(s_\Phi(\lambda)\). Let \(a_C\) be a string of \(s_C(\lambda)\) random bits used by an efficient sampler to generate a unitary \(C_{a_C}\) from an exact unitary \(2\)-design. These three lengths are polynomially bounded in \(\lambda\). We combine the two keys and the sampler's random bits into a single seed \(a=(a_P,a_\Phi,a_C)\) of total length \(s(\lambda)=s_P(\lambda)+s_\Phi(\lambda)+s_C(\lambda)\), and define
\begin{equation}
V_a^{\PFC}:=
P_{\pi_{a_{\mathsf P}}}\Phi_{h_{a_{\Phi}}}C_{a_{\mathsf C}}.
\label{eq:pfc-inner-pru}
\end{equation}
The permutation and Boolean-function primitives are secure against quantum superposition queries because the data register encounters them coherently. This is stronger than the outer PRF $F_k$, which is only required to be post-quantum secure.

For $t\leq N$, the distinguishing advantage is at most $4t(t-1)/(N+1)$ between the ideal PFC ensemble and Haar measure~\cite{MaHuang25}. Replacing its random permutation and function by the preceding quantum-secure families therefore makes Eq.~\eqref{eq:pfc-inner-pru} an adaptive forward PRU whenever $n(\lambda)=\omega(\log\lambda)$. Choose the outer PRF \(F_k:\{0,1\}^{d(\lambda)}\to\{0,1\}^{s(\lambda)}\), and split its output as
\begin{equation}
F_k(x)=(F_{k}^{(\mathsf{P})}(x),F_{k}^{(\mathsf F)}(x),F_{k}^{(\mathsf C)}(x)),
\end{equation}
where the three components have lengths \(s_P(\lambda)\), \(s_\Phi(\lambda)\), and \(s_C(\lambda)\), respectively. Then the resulting classical-label family is explicitly
\begin{equation}
G_k^{\PFC}(x):=P_{\pi_{F^{(\mathsf{P})}_{k}(x)}}\Phi_{h_{F^{(\mathsf{F})}_{k}(x)}}C_{F^{(\mathsf{C})}_{k}(x)}.
\label{eq:classical-pfc-prfu}
\end{equation}

\begin{corollary}[Classical-label PFC construction]\label{cor:classical-pfc}
Suppose the outer $F$ is a post-quantum secure PRF with classical oracle access, the permutation and Boolean-function families satisfy the quantum security assumptions above, and efficient forward Haar proxies satisfy the hypothesis of Corollary~\ref{cor:classical-existence}. Let $\varepsilon_{\mathsf C}$ be the total-variation distance between the implemented $2$-design sampler and its intended exact distribution, and assume $q\varepsilon_{\mathsf C}=\negl(\lambda)$ for every polynomial $q$. If $n(\lambda)=\omega(\log\lambda)$, then the family in Eq.~\eqref{eq:classical-pfc-prfu} is a classical-label adaptive forward PRFU. In particular, the three cryptographic primitives are known from quantum-secure one-way functions~\cite{Zhandry12QRF,Zhandry25PRP}.
\end{corollary}

Appendix~\ref{app:classical-pfc} gives the resulting quantitative bound. Any negligible statistical bias caused by representing the $2$-design sampler with a fixed-length random seed is absorbed into the inner-PRU advantage.

\section{Construction and security with coherent labels}
\label{sec:coherent-construction}

The first-occurrence proof of Theorem~\ref{thm:classical-lifting} uses the fact that a $q$-query adversary can present at most $q$ distinct classical labels. This premise disappears under Eq.~\eqref{eq:coherent-oracle-access}: one query may have support on every $x\in\X_\lambda$, and the amplitudes belonging to different labels remain available for interference. Consequently there is no online ordering of a polynomial number of ``new labels'' on which to run the classical hybrid.


To construct a PRFU supporting coherent label access, we build on the PFC ensemble introduced in the preceding subsection. We use the indexed path-recording technique~\cite{FoxmanEtAl26}, allowing us to analyze queries in which different labels occur in coherent superposition. This indexed path-recording argument establishes that the ideal indexed PFC ensemble is statistically indistinguishable from a family of independent Haar-random unitaries under adaptive coherent-label queries. We then use computational hybrids to establish security of the efficient construction, assuming a quantum-secure pseudorandom function and a  tweakable in-place pseudorandom permutation secure under coherent queries to both the label and data registers.

\subsection{Quantum-secure primitives and the construction}
\label{subsec:coherent-primitives}
Unlike the classical-label construction, which uses an outer PRF to generate the inner keys, our coherent-label construction incorporates the label directly into the phase and permutation primitives to form the indexed PFC ensemble. The secret keys are sampled once and fixed throughout the interaction, while queries may involve superpositions of both labels and data. 

We first introduce the ideal indexed PFC ensemble. Recall that $\X_\lambda\subseteq\{0,1\}^{d(\lambda)}$ is the label domain, and let $\mathcal Z_\lambda:=\{0,1\}^{n(\lambda)}$ be the computational-basis domain of the data register, with $N:=|\mathcal Z_\lambda|=2^{n(\lambda)}$. Independently sample a uniformly random Boolean function $f:\X_\lambda\times\mathcal Z_\lambda\to\{0,1\}$, a family $\bm\pi=(\pi_x)_{x\in\X_\lambda}$ of independent uniformly random permutations of $\mathcal Z_\lambda$, and a unitary $C$ from an efficiently samplable exact unitary $2$-design with efficient circuit implementations. All three choices remain fixed throughout the interaction. Writing $f_x(z):=f(x,z)$, define the joint label--data operators as
\begin{align}
\widehat\Phi_f\ket{x,z}&=(-1)^{f_{x}(z)}\ket{x,z},\nonumber\\
\widehat P_{\bm\pi}\ket{x,z}&=\ket{x,\pi_x(z)}.
\label{eq:coherent-ideal-primitives}
\end{align}
We use hats for operators on the joint label and data registers. The data-register operators $P_\pi$ and $\Phi_h$ retain the notation of Eq.~\eqref{eq:pfc-operators}. The ideal indexed PFC ensemble consists of the unitaries
\begin{equation}
\begin{split}
W_{\mathrm{PFC}}&:=\widehat P_{\bm\pi}\widehat\Phi_f(I_L\otimes C)\\
&=\sum_{x\in\X_\lambda}\ketbra{x}{x}\otimes P_{\pi_x}\Phi_{f_x}C.
\end{split}
\label{eq:controlled-pfc-ensemble}
\end{equation}
For each fixed label $x$, the data-register unitary $P_{\pi_x}\Phi_{f_x}C$ has the usual random PFC distribution. The phase functions $f_x$ and permutations $\pi_x$ are independent across labels, whereas a single $C$ is shared by all labels. Repeated queries at a label reuse the same unitary.

To obtain an efficient counterpart, we replace the random phase function and permutation family by keyed pseudorandom primitives. For the phase layer, let $\{g_{k_\Phi}\}_{k_\Phi}$ be a Boolean qPRF family with $g_{k_\Phi}:\X_\lambda\times\mathcal Z_\lambda\to\{0,1\}$ and secret key $k_\Phi$. Such a family is obtained by restricting a standard Boolean qPRF on $d(\lambda)+n(\lambda)$ input bits to the concatenated encodings of pairs in $\X_\lambda\times\mathcal Z_\lambda$. Write $g_{k_\Phi,x}(z):=g_{k_\Phi}(x,z)$. The corresponding phase unitary acts as
\begin{equation}
\widehat\Phi_{g_{k_\Phi}}\ket{x,z}
=(-1)^{g_{k_\Phi,x}(z)}\ket{x,z}.
\label{eq:coherent-phase-oracle}
\end{equation}
Define the distinguishing advantage $\Adv_g^{\mathsf{qprf}}(\mathcal B_\Phi;\lambda)$ by the probability difference similar to Eq.~\eqref{eq:pqprf-advantage}, with challenge oracles replaced by the phase unitary $\widehat\Phi_{g_{k_\Phi}}$ and $\widehat\Phi_{f}$. The distinguisher $\mathcal B_\Phi$ is allowed to make adaptive quantum queries, including superpositions over both labels and data. Standard qPRF security ensures that every nonuniform QPT adversary has negligible advantage in distinguishing it from the phase unitary for a uniformly random Boolean function on the same domain. 

For the permutation layer, let $\{\Pi_{k_P}\}_{k_P}$ be a keyed family with $\Pi_{k_P}:\X_\lambda\times\mathcal Z_\lambda\to\mathcal Z_\lambda$, such that for every label $x$ $\Pi_{k_P,x}:=\Pi_{k_P}(x,\cdot)$ is a permutation on $\mathcal{Z}_{\lambda}$. Its in-place unitary is
\begin{equation}
\widehat P_{\Pi_{k_P}}\ket{x,z}
=\ket{x,\Pi_{k_P,x}(z)}.
\label{eq:tweakable-inplace-oracle}
\end{equation}
Assume that this unitary can be implemented efficiently given $k_P$, with any auxiliary workspace returned to its initial state. For a distinguisher $\mathcal B_P$ with adaptive forward-only quantum access, define its distinguishing advantage as
\begin{align}
\Adv_\Pi^{\ipqTPRP}(\mathcal B_P;\lambda)
&:=\Biggl|\Pr_{k_P}\!\left[\mathcal B_P^{\widehat P_{\Pi_{k_P}}}(1^\lambda)=1\right]\nonumber\\[-1mm]
&\hspace{10mm}-\Pr_{\bm\pi}\!\left[\mathcal B_P^{\widehat P_{\bm\pi}}(1^\lambda)=1\right]\Biggr|.
\label{eq:coherent-primitive-advantages}
\end{align}
Here $k_P$ is sampled uniformly once, and $\bm\pi$ is sampled as in the ideal ensemble above. In either experiment, queries may be in superposition over both labels and data. Then the security of the in-place permutation can be defined as follows.

\begin{definition}[Forward in-place tweakable qPRP]
\label{def:forward-inplace-tprp}
The family $\{\Pi_{k_P}\}_{k_P}$ is a forward in-place tweakable qPRP if $\Adv_\Pi^{\ipqTPRP}(\mathcal B_P;\lambda)$ is negligible in $\lambda$ for every nonuniform QPT distinguisher $\mathcal B_P$ with adaptive forward in-place access to the in-place permutation oracle in Eq.~\eqref{eq:tweakable-inplace-oracle}.
\end{definition}
The in-place requirement is essential to this assumption. As emphasized by Ref.~\cite{MaHuang24PolylogNote},  security under forward quantum queries to the standard XOR permutation oracle does not by itself establish security under in-place queries. Access to both forward and inverse XOR permutation oracles, however, implements an in-place query using one call in each direction, by computing the permuted value, swapping registers, and uncomputing the original input~\cite{KashefiEtAl02,HolmanEtAl25}. This simulation also preserves superpositions over labels, so strong tweakable qPRP security against coherent queries to both labels and data implies Definition~\ref{def:forward-inplace-tprp}. We require only forward in-place security, while retaining joint security across coherently queried labels. Security for each fixed label separately does not establish this guarantee.

Finally, sample $k_\Phi$, $k_P$, and $C$ independently, with $C$ drawn from the same exact unitary $2$-design as in the ideal ensemble, and let $k=(k_\Phi,k_P,\operatorname{desc}(C))$. The efficient construction is
\begin{equation}
\begin{split}
W_k^{\mathsf{coh}}&=\widehat P_{\Pi_{k_P}}\widehat\Phi_{g_{k_\Phi}}(I_L\otimes C)\\
&=\sum_{x\in\X_\lambda}\ketbra{x}{x}\otimes G_k^{\mathsf{coh}}(x),
\end{split}
\label{eq:coherent-pfc-construction}
\end{equation}
where $G_k^{\mathsf{coh}}(x)=P_{\Pi_{k_P,x}}\Phi_{g_{k_\Phi,x}}C$ is defined similarly to Eq.~\eqref{eq:pfc-operators} and Eq.~\eqref{eq:pfc-inner-pru}. This $W_{k}^{\mathsf{coh}}$ is a label-controlled unitary, also called a block-diagonal unitary. Relative to the decomposition $\mathcal H_L\otimes\mathcal H_A=\bigoplus_{x\in\X_\lambda}(\operatorname{span}\{\ket{x}\}_L\otimes\mathcal H_A)$, the unitary $G_k^{\mathsf{coh}}(x)$ is the block associated with label $x$ and acts only on the data register $A$. 
As in the ideal ensemble in Eq.~\eqref{eq:controlled-pfc-ensemble}, the same $C$ is shared by all labels, and the key remains fixed throughout the interaction. Replacing the keyed phase function and permutation family by their random counterparts yields exactly $W_{\mathrm{PFC}}$.

\subsection{Security proof}
\label{subsec:coherent-security}

We first establish information-theoretic closeness between the ideal indexed PFC ensemble and the function-like Haar family. Let $L$ and $A$ be the label and data registers, and let $B_j$ contain the adversary's private workspace, reference systems, and purifying registers at round $j$. By purification and deferred measurement, an adaptive $q$-query strategy can be represented by isometries $\mathcal A_j:\mathcal H_{LAB_{j-1}}\to\mathcal H_{LAB_j}$. Starting from a fixed normalized state $\ket{\psi_0}_{LAB_0}$ independent of the oracle, define
\begin{equation}
\ket{\Psi_q^W}:=(W_{LA}\mathcal A_q)\cdots(W_{LA}\mathcal A_1)\ket{\psi_0}_{LAB_0},
\label{eq:coherent-adversary-state-main}
\end{equation}
where $W$ is either $W_{\mathrm{PFC}}$ from Eq.~\eqref{eq:controlled-pfc-ensemble} or $W_{\mathrm{Haar}}$ from Eq.~\eqref{eq:classical-controled-unitary}, and identities on the current workspace are implicit. The state lies on $LAB_q$; any final channel or measurement can be appended without increasing trace distance. The two averaged output states are
\begin{align}
\rho_{\mathrm{PFC}}^{\mathcal A}&:=\mathbb E_{f,\bm\pi,C}\!\left[\ketbra{\Psi_q^{W_{\mathrm{PFC}}}}{\Psi_q^{W_{\mathrm{PFC}}}}\right],\nonumber\\
\rho_{\mathrm{Haar}}^{\mathcal A}&:=\mathbb E_{\bm H}\!\left[\ketbra{\Psi_q^{W_{\mathrm{Haar}}}}{\Psi_q^{W_{\mathrm{Haar}}}}\right],
\label{eq:coherent-averaged-output-states}
\end{align}
where the first expectation includes the random function, permutation family, and common design element, and the second includes independent $H_x\leftarrow\Haar(\mathrm U(N))$. Both experiments use the same initial state and adversarial isometries. Their trace distance satisfies the following bound.

\begin{theorem}[PFC--Haar closeness under coherent-label access]
\label{thm:pfc-haar-closeness}
For every possibly unbounded adaptive adversary $\mathcal A$ making $q<N$ forward oracle calls, the averaged output states in Eq.~\eqref{eq:coherent-averaged-output-states} satisfy
\begin{equation}
\td\!\left(\rho_{\mathrm{PFC}}^{\mathcal A},\rho_{\mathrm{Haar}}^{\mathcal A}\right)\leq\frac{2q(q-1)}{N+1},
\label{eq:pfc-haar-closeness-bound}
\end{equation}
where $\td(\rho,\sigma):=\frac12\|\rho-\sigma\|_1$ and $N=|\mathcal{Z}_{\lambda}|=2^{n(\lambda)}$. 
\end{theorem}

Appendix~\ref{app:coherent} gives a self-contained path-recording proof of this bound by recording triples $(x,z,y)$ and enforcing injectivity separately for each label. A right-unitary invariance identity transfers the common $C$ to the recorded input coordinates. Its exact second moment bounds the probability of a repeated recorded input within the same label by $\beta_q=q(q-1)/(N+1)$. On records without such collisions, the purified permutation--phase experiment agrees exactly with the recorder under an isometry of the hidden registers. Comparing both the indexed PFC experiment and the function-like Haar experiment with this same recorder gives Eq.~\eqref{eq:pfc-haar-closeness-bound}, with no dependence on the number of labels.

For a coherent-label adversary $\mathcal A$, let
\begin{align}
\Adv_G^{\mathsf{coh},\rightarrow}(\mathcal A;\lambda)
&:=\Bigl|\Pr_k\!\left[\mathcal A^{\mathcal O_{\ket{G_k^{\mathsf{coh}}}}}(1^\lambda)=1\right]\nonumber\\[-1mm]
&\hspace{10mm}-\Pr_{\bm H}\!\left[\mathcal A^{\mathcal O_{\ket{\bm H}}}(1^\lambda)=1\right]\Bigr|.
\label{eq:coherent-construction-advantage}
\end{align}

We first bound the distinguishing advantage for any efficient adversary.

\begin{theorem}[Coherent-label PFC construction]
\label{thm:coherent-pfc-construction}
For every nonuniform QPT adversary $\mathcal A$ making at most $q<N$ adaptive forward calls to the coherent-label oracle, there exist nonuniform QPT distinguishers $\mathcal B_\Phi$ and $\mathcal B_P$, making at most $q$ phase-oracle queries and $q$ forward in-place tweakable-permutation queries, respectively, such that
\begin{align}
\Adv_G^{\mathsf{coh},\rightarrow}(\mathcal A;\lambda)
&\leq\Adv_g^{\mathsf{qprf}}(\mathcal B_{\Phi};\lambda)+\Adv_\Pi^{\ipqTPRP}(\mathcal B_P;\lambda)\nonumber\\[-1mm]
&\quad+\frac{2q(q-1)}{N+1}.
\label{eq:coherent-pfc-security-bound}
\end{align}
\end{theorem}
The computational proof uses two hybrids. First, $g_{k_\Phi}$ is replaced by a uniformly random Boolean function through the phase-oracle challenge of $\mathcal B_\Phi$. Second, the keyed permutation is replaced by independent uniformly random permutations through the forward in-place tweakable-permutation challenge of $\mathcal B_P$. An efficient $2q$-wise independent phase function makes the second reduction efficient while preserving its acceptance probabilities exactly. The resulting ideal experiment is Eq.~\eqref{eq:controlled-pfc-ensemble}, to which Theorem~\ref{thm:pfc-haar-closeness} applies. Appendix~\ref{app:coherent-computational} gives the reductions and explains why the phase-oracle advantage is bounded by standard qPRF security using one function-oracle query per phase query.

\begin{corollary}[Coherent-label adaptive PRFU]
\label{cor:coherent-existence}
Suppose that $g_{k_{\Phi}}:\mathcal{X}_{\lambda}\times\mathcal{Z}_{\lambda}\rightarrow\{0,1\}$ is a quantum-secure pseudorandom function, $\Pi_{k_{P}}:\mathcal{X}_{\lambda}\times\mathcal{Z}_{\lambda}\rightarrow\mathcal{Z}_{\lambda}$ is an in-place tweakable quantum secure pseudorandom permutation as defined in Definition~\ref{def:forward-inplace-tprp}, and $n(\lambda)=\omega(\log\lambda)$. Then the family in Eq.~\eqref{eq:coherent-pfc-construction} is a coherent-label adaptive forward PRFU according to Definition~\ref{def:coherent-prfu}.
\end{corollary}

\section{Pseudorandom channels and combs}\label{sec:processes}

We now turn a function-like unitary into a random-looking irreversible evolution and then into a random-looking process with memory. Because there is no concept of canonical Haar measure for channels and combs, we will define each reference ensemble in a natural and operational way.

\subsection{Channels from Stinespring dilation}
\label{subsec:stinespring-channels}

We first obtain pseudorandom channels by applying a PRFU inside a Stinespring dilation and discarding the environment.
Fix an input system $A$, an output system $B$, and environment systems $E_{\mathrm{in}}$ and $E_{\mathrm{out}}$ such that the input and output have matching dimension, $\dim(AE_{\mathrm{in}})=\dim(BE_{\mathrm{out}})$. For a unitary $U:\mathcal{H}_{A}\otimes \mathcal{H}_{E_{\mathrm{in}}}\rightarrow \mathcal{H}_{B}\otimes\mathcal{H}_{E_\mathrm{out}}$, define the channel
\begin{equation}\label{eq:Haar-Stinespring}
    \Phi_{U}(\rho_{A})=\tr_{E_{\mathrm{out}}}[U(\rho_{A}\otimes\ketbra{0}{0}_{E_{\mathrm{in}}})U^{\dagger}].
\end{equation}
Sampling $U$ from Haar measure induces a distribution on channels $\Phi_U$, which we call the Haar–-Stinespring channel ensemble for the specified system and environment dimensions.

We can also define the function-like quantum channel in a way similar to the function-like Haar random family. Let \(H=(H_x)_{x\in\mathcal X_\lambda}\) be a function-like Haar family, with independent unitaries sampled once and fixed throughout the experiment. The corresponding function-like Haar--Stinespring family is \(\{\Phi_{H_x}\}_{x\in\mathcal X_\lambda}\). Each query with label \(x\) initializes a fresh environment \(E_{\mathrm{in}}\) in \(\lvert0\rangle\), applies \(H_x\), and discards \(E_{\mathrm{out}}\).

For the efficient counterpart, insert PRFU $G_{k}(x)$ in Eq.~\eqref{eq:Haar-Stinespring} for each label $x$ to form the ensemble $\{\Phi_{G_{k}(x)}\}_{x\in\mathcal{X}_{\lambda}}$. We can then lift the PRFU to pseudorandom function-like channels (PRFC). Here the adversary $\mathcal{A}$ may adaptively query the resulting channel, keep arbitrary references and repeat labels adaptively, while not being able to control the input environment $E_{\mathrm{in}}$ nor access the discarded environment $E_\mathrm{out}$. The distinguishing advantage \(\operatorname{Adv}_{\mathrm{PRFC}}(\mathcal A;\lambda)\) is defined as the absolute difference between $\mathcal{A}$'s acceptance probabilities when interacting with $\Phi_{G_k(x)}$, for one sampled key \(k\), and with $\Phi_{H_x}$.

\begin{theorem}[PRFU-to-channel lifting]
\label{thm:channel-lifting}
If $\{G_{k}(x)\}$ is a classical-label forward PRFU, then $\{\Phi_{G_{k}(x)}\}$ is computationally indistinguishable from function-like Haar--Stinespring channels with the same dimensions.  For every QPT adaptive labeled-channel distinguisher $\mathcal A$ making $Q$ calls, there is a PRFU distinguisher $\mathcal B$ making $Q$ calls such that
\begin{equation}
   \operatorname{Adv}_{\mathrm{PRFC}}(\mathcal A;\lambda)=
   \Adv_G^{\mathsf{cl},\rightarrow}(\mathcal B;\lambda).
   \label{eq:channel-lifting-bound}
\end{equation}
\end{theorem}

The reduction initializes $E_{\mathrm{in}}$ in the $\ket{0}$ state, forwards $AE_{\mathrm{in}}$ to its unitary oracle and traces out the $E_{\mathrm{out}}$ system. The simulation is then exact in both PRFC and function-like Haar-Stinespring experiments. The result is followed by a data-processing argument: hiding part of a pseudorandom unitary cannot make it easier to distinguish. Appendix~\ref{app:process-liftings} gives the full reduction.

\subsection{Sequential processes and their testers}
\label{subsec:comb-formalism}

A channel describes a single input-output transformation. To describe several successive interactions, we allow the process to retain a private quantum memory between rounds. This memory can carry information from earlier inputs into later outputs, producing temporal correlations. Such a sequential process is called a quantum comb~\cite{ChiribellaEtAl08,ChiribellaEtAl09,GutoskiWatrous07,Gutoski12,PollockEtAl18}. As illustrated in Fig.~\ref{fig:comb}, an $R$-round comb realization consists of channels
\begin{equation}
\mathcal C_i:X_iM_{i-1}\longrightarrow\mathsf Y_iM_i,\qquad i\in[R],
\label{eq:comb-teeth}
\end{equation}
where $X,Y,M$ denote quantum registers. At round $i$, the process receives $X_i$, applies $\mathcal C_i$ jointly to this input and its private memory $M_{i-1}$, and returns $Y_i$ while retaining $M_i$. The initial memory is prepared in a fixed state $\sigma_{M_0}$, and the terminal memory $M_R$ is discarded. Each map $\mathcal C_i$ is called a \emph{tooth} of the comb.

A \emph{tester} describes the experimenter's side of this interaction, shown in green in Fig.~\ref{fig:comb}. Its initial operation $\mathcal T_0$ prepares an initial state in register $X_1Z_0$. For $i<R$, its operation $\mathcal T_i$ maps $Y_iZ_{i-1}$ to $X_{i+1}Z_i$, and $\mathcal T_R$ measures $Y_RZ_{R-1}$ to produce an outcome. The workspace $Z_i$ allows the tester to retain quantum information and adapt later inputs to earlier outputs, including through intermediate measurements and feedback. The tester never accesses the comb's private memory $M_i$. 

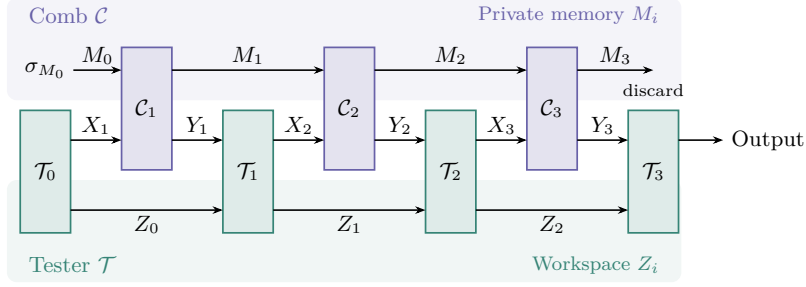
\begin{figure*}[htbp]
    \centering
    \resizebox{0.6\textwidth}{!}{\begin{tikzpicture}[x=1cm,y=1cm,font=\small,
 wire/.style={-{Stealth[length=1.5mm,width=1.1mm]},line width=.65pt},
 box/.style={draw,line width=.7pt,minimum width=.72cm,minimum height=1.75cm},
 cb/.style={box,draw=comb,fill=comb!18},
 tb/.style={box,draw=test,fill=test!16},
 lab/.style={inner sep=2pt}]
\fill[comb!7,rounded corners=4pt] (-.55,.57) rectangle (9.1,2.03);
\fill[test!7,rounded corners=4pt] (-.55,-2.03) rectangle (9.1,-.57);
\node[anchor=west,text=comb] at (-.35,1.78) {Comb $\mathcal C$};
\node[anchor=east,text=comb,font=\footnotesize] at (8.9,1.78) {Private memory $M_i$};
\node[anchor=west,text=test] at (-.35,-1.78) {Tester $\mathcal T$};
\node[anchor=east,text=test,font=\footnotesize] at (8.9,-1.78) {Workspace $Z_i$};
\foreach \i/\x in {1/1.45,2/4.35,3/7.25}{
 \node[cb] (c\i) at (\x,.45) {$\mathcal C_\i$};
}
\foreach \i/\x in {0/0,1/2.9,2/5.8,3/8.7}{
 \node[tb] (t\i) at (\x,-.45) {$\mathcal T_\i$};
}
\node at (0,1) {$\sigma_{M_0}$};
\draw[wire] (.4,1) -- node[above,lab] {$M_0$} (1.09,1);
\draw[wire] (1.81,1) -- node[above,lab] {$M_1$} (3.99,1);
\draw[wire] (4.71,1) -- node[above,lab] {$M_2$} (6.89,1);
\draw[wire] (7.61,1) -- node[above,lab] {$M_3$} (8.7,1);
\node[font=\scriptsize] at (8.7,.72) {discard};
\draw[wire] (.36,-1) -- node[below,lab] {$Z_0$} (2.54,-1);
\draw[wire] (3.26,-1) -- node[below,lab] {$Z_1$} (5.44,-1);
\draw[wire] (6.16,-1) -- node[below,lab] {$Z_2$} (8.34,-1);
\foreach \a/\b/\lbl in {.36/1.09/X_1,1.81/2.54/Y_1,3.26/3.99/X_2,4.71/5.44/Y_2,6.16/6.89/X_3,7.61/8.34/Y_3}{
 \draw[wire] (\a,0) -- node[above,lab] {$\lbl$} (\b,0);
}
\draw[wire] (9.06,0) -- (9.7,0) node[right] {Output};
\end{tikzpicture}}
    \caption{A three-round quantum comb (purple) interacting with a tester (green). The comb receives $X_i$, returns $Y_i$, and retains private memory $M_i$ between rounds. The tester interleaves its operations with the comb's teeth, carrying its own workspace $Z_i$ from one operation to the next. The comb starts with memory state $\sigma_{M_0}$ and discards its terminal memory. The tester produces the final outcome.}
    \label{fig:comb}
\end{figure*}


The input, output, and memory registers, together with the initial memory state, specify the shape of a sequential realization:
\begin{equation}\label{eq:unitary-process-shape}
\Sigma=\bigl(R,\{X_i,Y_i\}_{i=1}^{R},\{M_i\}_{i=0}^{R},\sigma_{M_0}\bigr).
\end{equation}
Here \(R\) is the number of rounds, \(X_i\) and \(Y_i\) are the communicated input and output registers, and \(M_i\) is the private memory retained after round \(i\). Each execution begins with \(M_0\) in the state \(\sigma_{M_0}\), independently of the tester, and ends by discarding \(M_R\). The shape specifies the register spaces and initialization, while the tooth maps \(\mathcal C_i\) determine the dynamics. 

To obtain a random reference ensemble from Haar measure, we consider realizations in which each tooth is a unitary interaction between the communicated input and the private memory. The following definition fixes a shape that supports unitary teeth and samples those interactions independently from Haar measure.

\begin{definition}[Sequential-Haar random comb]
\label{def:sequential-haar-comb}
Fix a public process shape
\begin{equation}\label{eq:ensemble-Haar-shape}
    \Sigma=\bigl(R,\{X_i,Y_i\}_{i=1}^{R},\{M_i\}_{i=0}^{R},\ketbra{0}{0}_{M_0}\bigr),
\end{equation}
such that for all $i\in[R]$, the register dimensions satisfy 
\begin{equation}\label{eq:process-Haar-shape-condition}
    \dim(\mathcal H_{X_i}\otimes\mathcal H_{M_{i-1}}) = \dim(\mathcal H_{Y_i}\otimes\mathcal H_{M_i}) = D. 
\end{equation}
Sample independent unitaries $\bm H=(H_1,\ldots,H_R)$ with each $H_i\leftarrow\Haar(\mathrm{U}(D))$ sampled from Haar random independently. At tooth $i$, apply $H_i:X_iM_{i-1}\rightarrow Y_iM_i$ as the channel. The resulting random comb is denoted $\mathcal H_{\bm H}^{\Sigma}$ and its distribution is called the sequential-Haar reference for the specified shape $\Sigma$.
\end{definition}

For cryptographic security, we allow an adversary to interact repeatedly with the same comb. Each execution, called a session, follows the sequence illustrated in Fig.~\ref{fig:comb} and starts with a fresh private memory in $\ketbra{0}{0}_{M_0}$. The adversary may maintain several active sessions and interleave their rounds. It may also retain a quantum workspace entangled with the inputs and outputs of different sessions.

We specify this interaction through a classical control interface. The command \(\mathsf{Open}\) starts a new session and returns a classical identifier \(s\). For a session whose next round is \(i\), the command \(\mathsf{Step}(s,X_i)\) supplies the input register \(X_i\), applies \(\mathcal C_i\) to \(X_iM_{s,i-1}\), and returns \(Y_i\) while retaining the updated private memory \(M_{s,i}\). The memory is discarded when the session completes round \(R\) or is terminated early by \(\mathsf{Abort}(s)\). The conventional $R$-round tester above is the special case that opens and completes one session.

\begin{definition}[Pseudorandom comb]
\label{def:prcomb}
An efficiently generated keyed family $\{\mathcal C_k^{\Sigma}\}_k$ is a pseudorandom comb for shape $\Sigma$ if, for every nonuniform QPT comb distinguisher $\mathcal T$, interacting through the reusable interface with polynomially many total tooth invocations and producing a final bit, the distinguishing advantage is negligible,
\begin{align}
\Adv_{\PRComb}^{\Sigma}(\mathcal T;\lambda)
&:=\Bigl|\Pr_{k}\!\left[\mathcal T[\mathcal C_k^{\Sigma}]=1\right]\nonumber\\[-1mm]
&\hspace{9mm}-\Pr_{\bm H}\!\left[\mathcal T[\mathcal H_{\bm H}^{\Sigma}]=1\right]\Bigr|
\leq\negl(\lambda).
\label{eq:prcomb-advantage}
\end{align}
\end{definition}
This is the process analogue of ordinary PRU security. Repeated access to one keyed comb must be computationally indistinguishable from repeated access to one sampled sequential-Haar comb.

\subsection{From PRFU to pseudorandom comb}
\label{subsec:single-prcomb}

We now use distinct PRFU labels for the successive teeth of a comb to obtain a pseudorandom comb.
Let $G_k:\mathcal X_\lambda\to\mathrm U(\mathcal H_A)$ be a classical-label forward PRFU, and fix a process shape $\Sigma$ satisfying the shape conditions Eq.~\eqref{eq:process-Haar-shape-condition} of the sequential-Haar ensemble. We take the common dimension in that condition to be $D=\dim\mathcal H_A$ and fix efficient identifications $X_iM_{i-1}\cong A\cong Y_iM_i$. Choose efficiently computable distinct labels $\ell_1,\ldots,\ell_R\in\mathcal X_\lambda$, and construct $\mathcal C_k^\Sigma=(G_{k}(\ell_1),\cdots,G_{k}(\ell_{R}))$ by applying  $G_k(\ell_i)$ at the $i$-th tooth, which leads to the following security guarantee.

\begin{theorem}[PRFU-to-PRComb lifting]
\label{thm:prfu-prcomb-lifting}
Let \(\{G_{k}(x)\}_{x\in\mathcal X_{\lambda}}\) be a classical-label forward PRFU, and let $\ell_1,\ldots,\ell_R$ be the distinct, efficiently computable labels in $\mathcal{X}_{\lambda}$. Then $\mathcal C_k^{\Sigma}=(G_{k}(\ell_1),\cdots,G_{k}(\ell_{R}))$ forms a pseudorandom comb. More precisely, for every QPT comb distinguisher $\mathcal T$ making at most $q$ total tooth invocations, there is a classical-label PRFU distinguisher $\mathcal B$ making at most $q$ oracle calls such that
\begin{equation}
\Adv_{\PRComb}^{\Sigma}(\mathcal T;\lambda)=\Adv_G^{\mathsf{cl},\rightarrow}(\mathcal B;\lambda).
\label{eq:single-prcomb-lifting}
\end{equation}
\end{theorem}

The reduction simulates the reusable interface, maintaining a separate private memory for each active session and forwarding every invocation of tooth \(i\) to its PRFU oracle under label \(\ell_i\). The simulations are exact and use one oracle call per tooth invocation, proving the equality.

Distinct labels ensure that different teeth emulate independent Haar unitaries. Reusing the same unitary at every tooth does not generally suffice: even for Haar-random \(U\), the two-round memoryless process \((U,U)\) is distinguishable from \((H_1,H_2)\) with independent Haar unitaries. A tester can compare the two normalized Choi states using a swap test, which accepts with probability \(1\) in the former case and \(\tfrac12(1+D^{-2})\) on average in the latter. 

The same construction extends to a labeled family of combs under one key. Let \(x\in\mathcal L_\lambda\) label a process, and choose an efficiently computable injective map \(\ell:\mathcal L_\lambda\times[R]\to\mathcal X_\lambda\). For each \(x\in\mathcal L_\lambda\), define
\begin{equation}\label{eq:PRF-comb}
    \mathcal C_{k,x}^{\Sigma} := \bigl( G_k(\ell(x,1)),\ldots,G_k(\ell(x,R)) \bigr),
\end{equation}
where the tuple specifies the successive tooth unitaries of a comb with shape $\Sigma$. The ideal family consists of independent sequential-Haar combs
\begin{equation}
    \mathcal H_x^\Sigma := (H_{x,1},\ldots,H_{x,R}),
\end{equation}
with the unitaries $H_{x,i}$ sampled independently from Haar measure for all pairs $(x,i)$ and fixed throughout the experiment.

To access this family, the adversary calls \(\mathsf{Open}(x)\) to start a new execution of the comb labeled \(x\), with fresh private memory and a returned session identifier \(s\). Subsequent calls use this identifier: \(\mathsf{Step}(s,X_i)\) supplies the next input to that execution and returns \(Y_i\), while \(\mathsf{Abort}(s)\) terminates it and discards its memory.

The adversary may choose process labels adaptively, interleave sessions while preserving their internal round order, and retain an entangled workspace across them. Process labels and session choices are classical, while inputs, outputs, and workspace remain quantum. We require the entire labeled family to be indistinguishable from independent sequential-Haar combs under this access model.

\begin{definition}[Function-like pseudorandom comb]
\label{def:function-like-prcomb}
An efficient labeled family $\{\mathcal C_{k,x}^{\Sigma}\}_{k,x}$ is a function-like pseudorandom comb if, for every nonuniform QPT concurrent process distinguisher $\mathcal T$ making polynomially many total tooth invocations and producing a final bit, the distinguishing advantage is negligible,
\begin{align}
&\Adv_{\PRComb}^{\mathsf{fl},\Sigma}(\mathcal T;\lambda)\nonumber:=\Bigl|\Pr_k\!\left[\mathcal T[\{\mathcal C_{k,x}^{\Sigma}\}_x]=1\right]\nonumber\\[-1mm]
&\qquad\qquad\qquad-\Pr_{\bm H}\!\left[\mathcal T[\{\mathcal H_x^{\Sigma}\}_x]=1\right]\Bigr|
\leq\negl(\lambda).
\label{eq:function-like-prcomb-advantage}
\end{align}
\end{definition}

The security of the efficient function-like comb family in Eq.~\eqref{eq:PRF-comb} is established in the following theorem.

\begin{theorem}[Concurrent function-like comb lifting]
\label{thm:prfu-function-like-prcomb-lifting}
Let $\{G_{k}(x)\}_{x\in\X_\lambda}$ be a classical-label forward PRFU and let $\ell:\mathcal L_\lambda\times[R]\to\mathcal X_\lambda$ be an injective and efficient map. For every nonuniform QPT concurrent process distinguisher $\mathcal T$ making at most $q$ total tooth invocations against Eq.~\eqref{eq:PRF-comb}, there is a classical-label PRFU distinguisher $\mathcal B$ making at most $q$ calls such that
\begin{equation}
\begin{split}
\Adv_{\PRComb}^{\mathsf{fl},\Sigma}(\mathcal T;\lambda)
&=\Adv_G^{\mathsf{cl},\rightarrow}(\mathcal B;\lambda).
\end{split}
\label{eq:function-like-prcomb-lifting}
\end{equation}
Consequently, the construction in Eq.~\eqref{eq:PRF-comb} forms a function-like pseudorandom comb.
\end{theorem}

The reduction keeps a separate private memory register for every active session and forwards each requested tooth to its PRFU oracle under the classical label $\ell(x,i)$. Injectivity turns the ideal PRFU table into independent Haar teeth indexed by $(x,i)$, so the simulation is exact and incurs neither a hybrid over rounds nor a loss depending on the number of possible process labels. As a special case, restricting the concurrent game to one fixed process label recovers the reusable PRComb game of Definition~\ref{def:prcomb}.

\section{Applications}
\label{sec:applications}

We now apply PRFUs to tasks that require multiple reproducible random-looking operations under a single secret key. Public labels specify the context of each operation, allowing distinct labels to select unitaries that jointly emulate independent Haar samples. We first combine this capability with unitary gluing to construct PRUs whose register widths can be chosen after key generation, and then consider nonce-resolved quantum authentication, coherently masked QRAM queries, and random multi-time dynamics. Throughout this section, $\mathrm{enc}$ denotes an efficiently computable injective encoding of the indicated public fields into $\mathcal{X}_{\lambda}$.

\subsection{One-key PRUs for variable register widths}
\label{subsec:variable-width-pru}

A standard PRU acts on a register of a fixed width. We seek a single secret key that supports pseudorandom unitaries at several register widths, allowing the width to be chosen adaptively after key generation. The desired guarantee is joint: even an adversary querying several widths should be unable to distinguish the resulting family from independent Haar-random unitaries, one fixed unitary at each supported width.

PRFUs provide this capability when combined with unitary gluing~\cite{SchusterHaferkampHuang2025,MaHuang25}. Gluing composes independent Haar unitaries on overlapping local registers into a global unitary that approximates Haar randomness under repeated forward queries. For each supported width, we use distinct PRFU labels to generate the local gates in the gluing construction. PRFU security ensures that these gates jointly emulate independent Haar unitaries. Gluing then yields a family of global pseudorandom unitaries at the supported widths.

Let \(G_k(x)\) be a classical-label forward PRFU acting on \(2w\) qubits. We construct, under the same key \(k\), a family of unitaries \(V_{k,m}\) acting on \(n_m=2mw\) qubits for \(m\in[M]\), where \(w=w(\lambda)\) and \(M=M(\lambda)\) are polynomially bounded positive integers. For each target width \(n_m\), partition the register into consecutive \(w\)-qubit blocks \(A_1,\ldots,A_{2m}\). Label each local gate by its target width, layer \(a\), and position \(j\):
\begin{equation}
\ell_{m,a,j}:=\operatorname{enc}(\mathsf{glue},n_m,a,j),
\label{eq:gluing-labels}
\end{equation}
where $a=0$ with $j\in[m]$ denotes the first layer and $a=1$ with $j\in[m-1]$ the second. Define the two-layer unitary:
\begin{align}
    V_{k,m}^{(0)}&:=\prod_{j=1}^{m}G_k(\ell_{m,0,j})^{A_{2j-1}A_{2j}},\nonumber\\
    V_{k,m}^{(1)}&:=\prod_{j=1}^{m-1}G_k(\ell_{m,1,j})^{A_{2j}A_{2j+1}}.
    \label{eq:variable-width-brickwork}
\end{align}
Here, the superscripts $(0)$ and $(1)$ on $V_{k,m}$ denote the layer indices, whereas the superscripts on $G_k$ (e.g., $A_{2j-1}A_{2j}$) denote the gate supports, acting non-trivially on the specified blocks and as the identity elsewhere. The first layer acts on disjoint pairs of blocks, and the second connects neighboring pairs, giving the overlaps required by gluing. Composing these two layers gives the circuit 
\begin{equation}\label{eq:two-layer}
    V_{k,m}:=V^{(1)}_{k,m}V^{(0)}_{k,m}.
\end{equation}
For a family $\mathbf U=\{U_m\}_{m\in[M]}$, the oracle $\mathcal O_{\mathbf U}$ accepts a classical width index $m$ and applies $U_m$ to the supplied $n_m$-qubit register. The family is fixed throughout the interaction, and only forward queries are exposed. The following theorem establishes joint security across the supported widths. 

\begin{theorem}[One-key PRUs for variable register widths]
\label{thm:variable-width-pru}
Let $G$ be a classical-label forward PRFU on $2w$ qubits, where $w=\omega(\log{\lambda})$. Let $M(\lambda)$ be a polynomially bounded integer and set $\mathbf V_k=\{V_{k,m}\}_{m\in[M]}$ with $V_{k,m}$ defined in Eq.~\eqref{eq:two-layer}. Let $\mathbf H=\{H_m\}_{m\in[M]}$ consist of independently sampled Haar unitaries on $2mw$ qubits. For every QPT oracle algorithm \(\mathcal A\),
\begin{equation}
    \left| \Pr_{k\leftarrow\mathsf{Gen}(1^\lambda)} \!\left[\mathcal A^{\mathcal O_{\mathbf V_k}}(1^\lambda)=1\right] - \Pr_{\mathbf H} \!\left[\mathcal A^{\mathcal O_{\mathbf H}}(1^\lambda)=1\right] \right| \leq \operatorname{negl}(\lambda).
\end{equation}
\end{theorem}

The guarantee follows by combining joint PRFU security with unitary gluing. PRFU security first allows all local gates in Eq.~\eqref{eq:variable-width-brickwork} to be replaced jointly by independent Haar unitaries. Since different widths use disjoint label sets, the resulting circuits are independent across widths. The adaptive forward-query gluing bound~\cite{MaHuang25} then replaces each encountered circuit by one persistent Haar unitary on its full register. For $q$ total queries, the accumulated statistical error from gluing is at most $18Mq^2(q-1)/2^w$, which is negligible assuming $w=\omega(\log{\lambda})$. Appendix~\ref{app:variable-width-gluing} gives the full reduction and quantitative bounds.

A variant secure under both forward and inverse queries follows from a strong classical-label PRFU and inverse-secure gluing~\cite{AnanthBostanciGulatiLin2025}. See Appendix~\ref{app:strong-variable-width-gluing}.

\subsection{Nonce-resolved quantum authentication}
\label{subsec:nonce-auth}

We consider quantum authentication with a public nonce: a classical label accompanying each ciphertext that selects the operation used to encode and verify it. We combine a strong classical-label PRFU with the design layers of the pseudorandom quantum authentication scheme (PQAS) of Ref.~\cite{HaugBansalMokKohBharti2025}, using one secret key for all nonces. In addition to the inherited authentication guarantee, we establish rejection of ciphertexts submitted under previously unused nonces.

Let \(M,T,B\) be message, tag, and auxiliary registers of \(m,\ell,b\) qubits, respectively, and let \(G_{k_G}(x)\) be a strong classical-label PRFU on their combined \(m+\ell+b\) qubits. Independently sample efficiently implementable unitaries $V^{(4)}_{a_4}$ and $V^{(2)}_{a_2}$ from an approximate unitary $4$-design and an exact unitary $2$-design, respectively, on the same register, with parameters chosen as in Ref.~\cite{HaugBansalMokKohBharti2025}. The seeds $a_4,a_2$ are independent of $k_G$, sampled once, and retained as part of the master key $k=(k_G,a_4,a_2)$. For a nonce $\nu$, define
\begin{equation}
    U_{k,\nu}=G_{k_G}(\mathrm{enc}(\mathbf{auth},\nu))V^{(4)}_{a_4}V^{(2)}_{a_2}.
\label{eq:nonce-auth-unitary}
\end{equation}
Here $\mathrm{enc}$ is an efficient injective encoding into $\mathcal{X}_{\lambda}$, and the same two design unitaries are reused across nonces. The $2$-design supplies the second-moment identity controlling the mean authentication error, while the approximate $4$-design controls fluctuations across keys, as required by the authentication analysis of Ref.~\cite{HaugBansalMokKohBharti2025}.

We append an $\ell$-qubit tag initialized to $\ket{0^{\ell}}_{T}$ and an auxiliary register in the maximally mixed state $\tau_B:=I_B/2^b$ to the message, then apply $U_{k,\nu}$, forming the encoded state
\begin{equation}
    \mathsf{Enc}_{k,\nu}(\rho_M):=U_{k,\nu}\bigl(\rho_M\otimes\ketbra{0^\ell}{0^\ell}_T\otimes\tau_B\bigr)U_{k,\nu}^\dagger.
\label{eq:nonce-authentication}
\end{equation}
To verify, apply \(U_{k,\nu}^{\dagger}\), accept if the tag is all zero, and return the message register upon acceptance. Honest ciphertexts are accepted exactly. The security reduction therefore uses strong PRFU security to simulate verification through inverse challenge calls.

\begin{proposition}[Fresh-nonce rejection]
\label{prop:fresh-nonce-rejection}
Suppose \(G\) is a strong classical-label PRFU and $U_{k,\nu}$ is defined by Eq.~\eqref{eq:nonce-auth-unitary}. For every QPT adversary making adaptive encryption and verification queries, there exists a QPT strong classical-label PRFU distinguisher $\mathcal B$ such that the probability that the adversary's final ciphertext is accepted under a nonce unused in all preceding queries satisfies
\begin{equation}
    p_{\mathsf{fresh}}\leq 2^{-\ell}+\Adv_G^{\mathsf{cl},\pm}(\mathcal B;\lambda).
\label{eq:fresh-nonce-rejection}
\end{equation}
\end{proposition}

To see this, replace the PRFU family by independent Haar unitaries. Conditional on the design seeds, right multiplication by $V^{(4)}_{a_4}V^{(2)}_{a_2}$ preserves their joint Haar distribution. The unitary at the final fresh nonce is therefore independent of the preceding interaction, and Haar averaging gives acceptance probability $2^{-\ell}$. Appendix~\ref{app:nonce-auth-proof} gives the reduction. For each fixed nonce, Eq.~\eqref{eq:nonce-auth-unitary} is an instance of PQAS, so its authentication security follows directly from Ref.~\cite{HaugBansalMokKohBharti2025} under that work's attack model and parameter conditions. Proposition~\ref{prop:fresh-nonce-rejection} adds the fresh-nonce guarantee under adaptive encryption and verification queries.

\subsection{Coherently masked QRAM queries}
\label{subsec:coherent-qram}

Quantum random-access memory (QRAM) allows a classical database to be queried in a coherent superposition of addresses~\cite{GiovannettiLloydMaccone08}. For a database \(D:\mathcal X_\lambda\to\{0,1\}^m\), its lookup unitary acts as
\begin{equation}\label{eq:QRAM}
    Q_D\lvert x\rangle_L\lvert b\rangle_B =\lvert x\rangle_L\lvert b\oplus D(x)\rangle_B,
\end{equation}
with the action on address superpositions determined by linearity. We use a coherent-label PRFU to mask the retrieved data under one secret key: after lookup, an address-dependent unitary acts on the bus without measuring the address. The resulting interface preserves coherent queries while hiding the database contents from an observer with forward access.

Let $\{G_k(x)\}_{x\in\mathcal{X}_{\lambda}}$ be a coherent-label forward PRFU on $m$ qubits. After each lookup, apply the unitary selected by the address, giving the masked interface
\begin{align}
    \widetilde{Q}_{k,D}:=\left(\sum_{x\in\mathcal{X}_{\lambda}}\ketbra{x}{x}_{L}\otimes G_{k}(x)\right)Q_{D}.
\label{eq:masked-qram}
\end{align}
Only forward access to \(\widetilde Q_{k,D}\) is exposed. A key holder can undo the mask to recover ordinary QRAM lookup.

\begin{proposition}[Content hiding for coherently masked QRAM]
\label{prop:qram-content-hiding}
Suppose \(G\) is a coherent-label forward PRFU. For any two efficiently implementable databases \(D_0,D_1\) of the same dimensions, chosen independently of \(k\), and every QPT oracle algorithm \(\mathcal A\),
\begin{equation}\label{eq:qram-content-hiding}
    \left| \Pr_k\!\left[\mathcal A^{\mathcal O_{\widetilde Q_{k,D_0}}}(1^\lambda)=1\right] - \Pr_k\!\left[\mathcal A^{\mathcal O_{\widetilde Q_{k,D_1}}}(1^\lambda)=1\right] \right| \leq \operatorname{negl}(\lambda),
\end{equation}
where the oracles provide forward access with arbitrary coherent address queries.
\end{proposition}

Coherent-label PRFU security allows the masks to be replaced jointly by independent Haar unitaries $H_x$. At each address, the resulting bus operation is $H_xX^{D(x)}$, where $X^{D(x)}$ denotes addition of $D(x)$ modulo two. Haar invariance implies that these operations remain independent Haar unitaries, so the ideal interface is independent of the database. Appendix~\ref{app:qram-proof} gives the reduction. 

\subsection{Efficient emulation of random multi-time dynamics}
\label{subsec:random-multitime-emulation}

Random multi-time experiments involve a sequence of transformations, with interventions between steps and possibly an inaccessible environment that retains memory. Independently sampled Haar unitaries provide a reference model for these transformations, but generic Haar unitaries require exponentially large circuits in the number of qubits. The process construction above gives an efficient computational substitute: one secret key specifies reproducible experiments whose responses to efficient adaptive interventions emulate those of the sequential-Haar ensemble.

Let $G_k(x)$ be a classical-label forward PRFU on $a(\lambda)$ qubits, and fix a public process shape $\Sigma$ with polynomially many steps $R_{\max}=R_{\max}(\lambda)$, satisfying Eq.~\eqref{eq:process-Haar-shape-condition}. Each step acts on $a(\lambda)$ qubits, with the input, output, and private-memory registers specified by $\Sigma$. For an experimental instance labeled $x\in\mathcal L_\lambda$, define
\begin{align}
\ell_{\mathsf{dyn}}(x,i)&:=\operatorname{enc}(\mathsf{dynamics},\Sigma,x,i),\nonumber\\
\mathcal C_{k,x}^{\Sigma}&:=\bigl(G_k(\ell_{\mathsf{dyn}}(x,i))\bigr)_{i=1}^{R_{\max}}.
\label{eq:random-dynamics-emulator}
\end{align}
Each execution starts with fresh private memory and uses the $\mathsf{Open}(x)$, $\mathsf{Step}(s,X_i)$, and $\mathsf{Abort}(s)$ interface defined above. Reusing $x$ reproduces the same sequence of unitaries. Since the labels do not depend on the stopping time, an experiment may be stopped after any step without changing its preceding dynamics.

\begin{corollary}[Efficient emulation of sequential-Haar dynamics]
\label{cor:random-multitime-emulation}
For every QPT concurrent process distinguisher $\mathcal T$, the family in Eq.~\eqref{eq:random-dynamics-emulator} satisfies
\begin{equation}
\Adv_{\PRComb}^{\mathsf{fl},\Sigma}(\mathcal T;\lambda)
\leq\negl(\lambda),
\label{eq:random-multitime-emulation-bound}
\end{equation}
including when experimental labels and stopping times are chosen adaptively. Each step requires one PRFU evaluation, and all instances share one polynomial-length secret key.
\end{corollary}

The corollary follows directly from Theorem~\ref{thm:prfu-function-like-prcomb-lifting}: injectivity of $\ell_{\mathsf{dyn}}$ turns the ideal PRFU family into independent Haar unitaries indexed by $(x,i)$, and early termination is already allowed by the process interface. Thus the guarantee covers intermediate measurements, feedback, and interleaved executions with quantum workspace.

\section{Discussion and outlook}
\label{sec:discussion}

In this work, we introduce PRFUs as a framework for generating jointly pseudorandom quantum operations indexed by public labels under a single secret key. We establish a generic lifting from PRUs and post-quantum PRFs for classical-label access, and a separate construction for coherent-label access, whose security is proved through an indexed path-recording analysis that accounts for interference across labels. These guarantees extend to quantum channels and combs with inaccessible memory, allowing efficient adaptive interactions with entire families of processes. PRFUs thus connect unitary pseudorandomness with the construction of reproducible quantum dynamics across multiple contexts and time steps.

The relationship between PRU and PRFU remains a fundamental question. In the state setting, the construction of PRFS generators from PRS generators with logarithmic-length inputs provides a precedent for obtaining function-like pseudorandomness from a single-instance primitive~\cite{AnanthQianYuen22,AnanthGulatiQianYuen22}. However, the relationship between PRUs and PRFUs is not that clear. In the classical-label setting, our lifting converts a PRU into a PRFU using a post-quantum PRF to derive label-dependent keys. For coherent labels, however, we provide a specific construction rather than a general transformation from PRUs. This leaves open whether a comparable lifting is possible when queries may involve superpositions of labels. Establishing such a lifting, or identifying a separation, would clarify the additional cryptographic requirements imposed by coherent access to an indexed unitary family.

Our coherent-label construction also raises a concrete question about its underlying permutation primitive. While we formulate the required tweakable in-place qPRP security as an explicit assumption, an efficient construction achieving this guarantee from standard assumptions remains to be established. The challenge is to support a family of permutations whose joint pseudorandomness survives queries in superposition over both tweaks and inputs. Constructing such a primitive would instantiate the permutation component of our PRFU construction and clarify which techniques for quantum-secure permutations extend to this function-like setting.

\par\medskip
\begin{acknowledgments}
We thank Emil Khabiboulline and Zhenning Liu for helpful discussions. Z.X.~acknowledges funding and support from QuICS Lanczos Graduate Fellowship. K.B.~was supported by a Hartree Fellowship from the Joint Center for Quantum Information and Computer Science (QuICS) at the University of Maryland, College Park. D.E.K.~is supported by the National Research Foundation, Singapore through the National Quantum Office, hosted in A*STAR, under the Advanced Quantum Algorithms and Solutions (AQAS) Funding Initiative (S25Q9DA001). W.K.M. acknowledges funding provided by the Institute for Quantum Information and Matter, an NSF Physics Frontiers Center (NSF Grant PHY-2317110). A.V.G.~ acknowledges support from the U.S.~Department of Energy, Office of Science, Accelerated Research in Quantum Computing, Fundamental Algorithmic Research toward Quantum Utility (FAR-Qu). A.V.G.~was also supported in part by NSF QLCI (award No.~OMA-2120757), NSF STAQ program, NQVL:QSTD:Design:FTL, DoE ASCR Quantum Testbed Pathfinder program (award No.~DE-SC0024220), ONR MURI, AFOSR MURI, and ARL (W911NF-24-2-0107). A.V.G.~also acknowledges support from the U.S.~Department of Energy, Office of Science, National Quantum Information Science Research Centers, Quantum Systems Accelerator (award No.~DE-SCL0000121). The original ideas and proofs were developed by the authors. Generative artificial intelligence tools were subsequently used to assist with simplification of the proofs, clarification of their presentation, language editing, manuscript organization, and figure preparation.
\end{acknowledgments}


\bibliography{references}

\clearpage

\onecolumngrid
\appendix

\section{Security proof for the classical-label construction}
\label{app:classical}

We prove the lifting bound in Theorem~\ref{thm:classical-lifting}, construct the Haar proxies used in its corollaries, and derive the quantitative guarantee for the PFC instantiation. Throughout, $G_k(x)=V_{F_k(x)}$ is the construction in Eq.~\eqref{eq:classical-construction}, and the oracle interfaces and distinguishing advantages are those of Eqs.~\eqref{eq:classical-oracle-access}--\eqref{eq:true-classical-oracle-access} and Eqs.~\eqref{eq:classical-construction-advantage}--\eqref{eq:inner-pru-advantage}. Fix a nonuniform QPT adversary $\mathcal A$ making at most $q=q(\lambda)$ calls, and let $\mathcal D_{q,\eta}$ be a forward Haar proxy as in Definition~\ref{def:proxy}. All probabilities include the sampling of the oracle, the adversary's randomness, and its measurement outcomes. At each call, the reference register $R$ may include its entire retained quantum workspace.

\subsection{Replacing the outer pseudorandom function}

Let $\mathsf{Real}$ be the interaction with $\mathcal O_{G_k}$ for $k\leftarrow\Gen_F(1^\lambda)$. In the experiment $\mathsf{RF}$, replace $F_k$ by a uniformly random function $f:\{0,1\}^{d(\lambda)}\to\{0,1\}^{s(\lambda)}$, so that each call acts as
\begin{equation}
\mathcal O_f(x,\rho_{AR})
=(V_{f(x)}\otimes I_R)\rho_{AR}(V_{f(x)}^\dagger\otimes I_R).
\label{eq:random-function-oracle}
\end{equation}
Write $p_{\mathsf X}$ for the probability that $\mathcal A$ outputs one in experiment $\mathsf X$.

\begin{lemma}[PRF replacement]\label{lem:app-prf-step}
There is a QPT distinguisher $\mathcal B_F$, making at most $q$ classical function queries, such that
\begin{equation}
|p_{\mathsf{Real}}-p_{\mathsf{RF}}|
=\Adv_F^{\mathsf{prf}}(\mathcal B_F;\lambda).
\label{eq:app-prf-step}
\end{equation}
\end{lemma}
\begin{proof}
The distinguisher $\mathcal B_F$ runs $\mathcal A$ and, on a query with label $x$, evaluates its function oracle at $x$ to obtain a seed $a$. It implements $V_a$ on the supplied data register and eventually returns $\mathcal A$'s output bit. The two function-oracle experiments reproduce $\mathsf{Real}$ and $\mathsf{RF}$ exactly. All function queries are classical, although the choice of each label may depend on the preceding quantum interaction.
\end{proof}

Because the labels are classical, $\mathsf{RF}$ admits an exact lazy-sampling implementation: assign an independent uniform seed $a_x$ when a label $x$ first occurs and reuse it on every subsequent occurrence. Seeds are sampled with replacement. In particular, the simulation does not condition on the absence of seed collisions.

\subsection{Replacing local PRUs in order of first occurrence}

Assign rank $r(x)=1,2,\ldots$ to labels in the order in which they first occur. This order is determined separately within each experiment; it need not agree across different hybrids. At most $q$ ranks are reached. For $j\in\{0,\ldots,q\}$, define $\mathsf P_j$ by assigning to each new label $x$ a persistent unitary
\begin{equation}
W_x^{(j)}=
\begin{cases}
D_x, & r(x)\leq j,\quad D_x\leftarrow\mathcal D_{q,\eta},\\
V_{a_x}, & r(x)>j,\quad a_x\leftarrow\{0,1\}^{s(\lambda)}.
\end{cases}
\label{eq:pj-oracle}
\end{equation}
The draws for different labels are independent, and a query at $x$ applies $W_x^{(j)}$ to the data register. Thus
\begin{equation}
p_{\mathsf{RF}}=p_{\mathsf P_0},
\label{eq:rf-p0}
\end{equation}
and $\mathsf P_q$ uses an independent proxy for every queried label.

To analyze the $j$th transition, introduce $\mathsf Q_j$, which assigns
\begin{equation}
\widehat W_x^{(j)}=
\begin{cases}
D_x, & r(x)<j,\quad D_x\leftarrow\mathcal D_{q,\eta},\\
H_j, & r(x)=j,\quad H_j\leftarrow\Haar(\mathrm U(\mathcal H_A)),\\
V_{a_x}, & r(x)>j,\quad a_x\leftarrow\{0,1\}^{s(\lambda)}.
\end{cases}
\label{eq:qj-blocks}
\end{equation}
All choices are again independent and persistent.

\begin{lemma}[First-occurrence hybrid]\label{lem:first-occurrence-hybrid}
There are QPT single-instance PRU distinguishers $\mathcal B_1,\ldots,\mathcal B_q$, each making at most $q$ forward challenge calls, such that
\begin{equation}
|p_{\mathsf P_0}-p_{\mathsf P_q}|
\leq\sum_{j=1}^{q}\Adv_V^{\mathsf{pru},\rightarrow}(\mathcal B_j;\lambda)+q\eta.
\label{eq:app-all-pru-transitions}
\end{equation}
\end{lemma}
\begin{proof}
Fix $j$. The distinguisher $\mathcal B_j$ runs $\mathcal A$, samples proxy circuits for ranks below $j$, and samples independent inner-PRU seeds for ranks above $j$. It directs every occurrence of the label of rank $j$ to its single persistent challenge oracle. All other calls are implemented locally using the stored circuits. The PRU branch is exactly $\mathsf P_{j-1}$, and the Haar branch is exactly $\mathsf Q_j$, giving
\begin{equation}
|p_{\mathsf P_{j-1}}-p_{\mathsf Q_j}|
=\Adv_V^{\mathsf{pru},\rightarrow}(\mathcal B_j;\lambda).
\label{eq:app-pru-transition}
\end{equation}
The reduction is efficient because both the proxy circuits and the inner PRUs are efficient. It makes no challenge calls until rank $j$ occurs and at most $q$ calls in total.

The experiments $\mathsf Q_j$ and $\mathsf P_j$ differ only in the unitary assigned to rank $j$. To compare them, consider an oracle distinguisher that simulates the adversary, the label table, and all other independently sampled unitaries, while using its challenge at that rank. Definition~\ref{def:proxy} applies to this complete adaptive interaction and yields
\begin{equation}
|p_{\mathsf Q_j}-p_{\mathsf P_j}|\leq\eta.
\label{eq:app-proxy-transition}
\end{equation}
If rank $j$ is never reached, the two interactions coincide. Combining the two comparisons gives
\begin{equation}
|p_{\mathsf P_{j-1}}-p_{\mathsf P_j}|
\leq\Adv_V^{\mathsf{pru},\rightarrow}(\mathcal B_j;\lambda)+\eta.
\label{eq:app-one-label-bound}
\end{equation}
Summing over $j$ proves the lemma.

The PRU branch also reproduces seed collisions exactly. The hidden challenge seed and all locally sampled seeds are independent uniform strings; whenever two seeds coincide, the deterministic map $a\mapsto V_a$ assigns the same unitary, without requiring the reduction to recognize that coincidence. Hence no separate collision term is needed.
\end{proof}

\subsection{Replacing the proxies by Haar unitaries}

For $\ell\in\{0,\ldots,q\}$, let $\mathsf S_\ell$ assign independent Haar unitaries to the first $\ell$ ranks and independent proxies to the remaining ranks:
\begin{equation}
\widetilde W_x^{(\ell)}=
\begin{cases}
H_x, & r(x)\leq\ell,\quad H_x\leftarrow\Haar(\mathrm U(\mathcal H_A)),\\
D_x, & r(x)>\ell,\quad D_x\leftarrow\mathcal D_{q,\eta}.
\end{cases}
\label{eq:sell-blocks}
\end{equation}
These assignments are lazy and persistent, so $\mathsf S_0=\mathsf P_q$.

\begin{lemma}[Proxy replacement]\label{lem:proxy-to-haar}
The endpoint experiments satisfy
\begin{equation}
|p_{\mathsf P_q}-p_{\mathsf S_q}|\leq q\eta.
\label{eq:app-final-proxy-bound}
\end{equation}
\end{lemma}
\begin{proof}
For each $\ell$, simulate all of $\mathcal A$'s interaction except the unitary at rank $\ell$, which is supplied by a proxy-or-Haar challenge. The other unitaries are sampled independently of that challenge. This simulator may be computationally unbounded because it uses exact Haar unitaries at earlier ranks, as permitted by Definition~\ref{def:proxy}. It makes at most $q$ challenge calls, so $|p_{\mathsf S_{\ell-1}}-p_{\mathsf S_\ell}|\leq\eta$. The triangle inequality gives the claimed bound.
\end{proof}

In $\mathsf S_q$, every distinct label receives an independent Haar unitary, reused on all later calls. Deferred sampling at first occurrence therefore produces exactly the indexed Haar oracle of Eq.~\eqref{eq:true-classical-oracle-access}:
\begin{equation}
p_{\mathsf S_q}=\Pr_{\bm H}\!\left[\mathcal A^{\mathcal O_{\bm H}}(1^\lambda)=1\right].
\end{equation}

\begin{proof}[Proof of Theorem~\ref{thm:classical-lifting}]
Combining Lemmas~\ref{lem:app-prf-step}, \ref{lem:first-occurrence-hybrid}, and~\ref{lem:proxy-to-haar} gives
\begin{align}
\Adv_G^{\mathsf{cl},\rightarrow}(\mathcal A;\lambda)
&=|p_{\mathsf{Real}}-p_{\mathsf S_q}|\nonumber\\
&\leq |p_{\mathsf{Real}}-p_{\mathsf{RF}}|+|p_{\mathsf P_0}-p_{\mathsf P_q}|+|p_{\mathsf P_q}-p_{\mathsf S_q}|\nonumber\\
&\leq\Adv_F^{\mathsf{prf}}(\mathcal B_F;\lambda)
+\sum_{j=1}^{q}\Adv_V^{\mathsf{pru},\rightarrow}(\mathcal B_j;\lambda)+2q\eta.
\label{eq:app-complete-classical-bound}
\end{align}
The query bounds follow from the preceding simulations. More generally, if $Q$ bounds the total number of calls and $\ell$ bounds the number of distinct labels on every execution, a $(Q,\eta)$ proxy suffices and the proof uses only $\ell$ transitions, with total proxy error $2\ell\eta$.
\end{proof}

\subsection{Strong access and inverse queries}
\label{app:classical-strong}

The strong interface additionally permits
\begin{equation}
\mathcal O_{G_k}^{-1}(x,\rho_{AR})
=(G_k(x)^\dagger\otimes I_R)\rho_{AR}(G_k(x)\otimes I_R),
\label{eq:inverse-classical-oracle}
\end{equation}
with the analogous interface for $\bm H$. Both the label and the direction of each call are chosen classically and adaptively. A bidirectional $(q,\eta)$ Haar proxy satisfies Definition~\ref{def:proxy} for at most $q$ total calls to the sampled unitary and its inverse.

If $V$ is a strong PRU and such a proxy is available, the same hybrids apply. The PRF reduction reverses the known circuit for $V_a$ when required, each $\mathcal B_j$ uses the corresponding direction of its challenge, and every proxy supplies the persistent pair $(D_x,D_x^\dagger)$. Thus
\begin{equation}
\Adv_G^{\mathsf{cl},\pm}(\mathcal A;\lambda)
\leq\Adv_F^{\mathsf{prf}}(\mathcal B_F;\lambda)
+\sum_{j=1}^{q}\Adv_V^{\mathsf{pru},\pm}(\mathcal B_j;\lambda)+2q\eta.
\label{eq:app-strong-classical-bound}
\end{equation}
Here $q$ counts both directions together. Below we show that the amplified-design construction also supplies the required bidirectional proxies.

\subsection{Relative-error designs and efficient Haar proxies}
\label{app:relative-design-proxy}

Let $N=2^n$, let $\ket{\Phi_N}=N^{-1/2}\sum_{z=0}^{N-1}\ket{z}\ket{z}$, and use the normalized Choi operator $J(U)=(U\otimes I)\ketbra{\Phi_N}{\Phi_N}(U^\dagger\otimes I)$. For an ensemble $\mu$, define
\begin{equation}
\Omega_\mu^{(t)}:=\mathbb E_{U\leftarrow\mu}\!\left[J(U)^{\otimes t}\right].
\label{eq:choi-moment}
\end{equation}
These moments describe repeated uses of one sampled unitary, rather than fresh sampling at each call.

\begin{lemma}[Relative designs give forward proxies]\label{lem:relative-design-proxy}
An efficiently samplable circuit ensemble $\mathcal D$ satisfying
\begin{equation}
(1-\delta)\Omega_{\Haar}^{(q)}\preceq\Omega_{\mathcal D}^{(q)}\preceq(1+\delta)\Omega_{\Haar}^{(q)}
\label{eq:relative-choi-order}
\end{equation}
is a forward $(q,\delta)$ Haar proxy.
\end{lemma}
\begin{proof}
Fix an adaptive oracle distinguisher, including any side processes sampled independently of the challenge. If it terminates before $q$ calls, store its output and pad the interaction with calls on fresh registers whose outputs are discarded. Its acceptance probability is a positive linear functional of the joint Choi operator of the $q$ oracle slots: there is a positive semidefinite operator $T_q$ such that
\begin{equation}
p_\mu=\operatorname{Tr}\!\left[T_q\Omega_\mu^{(q)}\right].
\end{equation}
Here $T_q$ is the operator obtained by contracting the distinguisher's input preparation, intermediate operations, and final accepting measurement; the normalization factor associated with normalized Choi operators is absorbed into $T_q$. This representation includes arbitrary quantum memory between calls. Applying the positive functional $X\mapsto\operatorname{Tr}(T_qX)$ to Eq.~\eqref{eq:relative-choi-order} gives $(1-\delta)p_{\Haar}\leq p_{\mathcal D}\leq(1+\delta)p_{\Haar}$. Since $p_{\Haar}\leq1$, the difference is at most $\delta$.
\end{proof}

\subsubsection{An explicit amplified-design ensemble}

We use the efficiently derandomized PFC designs of Metger, Poremba, Sinha, and Yuen~\cite{MPSY24}. For an integer $t\geq1$, let $\mathfrak F_{n,2t}$ be an efficiently samplable $2t$-wise independent family of Boolean functions on $\{0,1\}^n$, and let $\mathfrak P_{n,t}$ be an efficiently samplable distribution of approximately $t$-wise independent permutations, with approximation error
\begin{equation}
\delta_{\mathsf P}\leq c_{\mathsf P}\frac{t}{N^{4t+1/2}}
\label{eq:mpsy-permutation-error}
\end{equation}
for a sufficiently small universal constant $c_{\mathsf P}>0$. Use the efficient circuit realizations in Ref.~\cite{MPSY24}, which implement the permutation unitaries as well as the phase operators. Sample the three factors independently and set
\begin{equation}
W=P_\pi\Phi_h C,\qquad
h\leftarrow\mathfrak F_{n,2t},\quad
\pi\leftarrow\mathfrak P_{n,t},\quad
C\leftarrow\operatorname{Cliff}(n).
\label{eq:mpsy-base-design}
\end{equation}
The operators $P_\pi$ and $\Phi_h$ are defined in Eq.~\eqref{eq:pfc-operators}. Denote this ensemble by $\nu_{n,t}$. Each sample has circuit size polynomial in $n$ and $t$.

\begin{proof}[Proof of Proposition~\ref{prop:efficient-haar-proxies}]
Theorem~4.4 and Lemma~4.6 of Ref.~\cite{MPSY24} give a universal constant $C_0\geq1$ such that the $t$-fold twirling channel of $\nu_{n,t}$ differs from the Haar twirling channel in diamond norm by at most
\begin{equation}
\alpha_{n,t}:=C_0t\,2^{-n/2}.
\label{eq:mpsy-base-error}
\end{equation}
For independent $W_1,\ldots,W_m\leftarrow\nu_{n,q}$, set
\begin{equation}
D:=W_m\cdots W_1.
\label{eq:mpsy-proxy-unitary}
\end{equation}
The amplification lemma in Ref.~\cite{MPSY24}, Lemma~4.11, bounds its diamond error by $\alpha_{n,q}^m$. The diamond-to-relative-error conversion in Lemma~2.7 of that work then gives
\begin{equation}
\varepsilon_{n,q,m}\leq N^{2q}\alpha_{n,q}^m
=2^{2nq-m\left(n/2-\log_2(C_0q)\right)}.
\label{eq:mpsy-relative-error}
\end{equation}
Under the hypothesis $C_0q\leq2^{n/4}$, choose
\begin{equation}
m:=\max\!\left\{1,\left\lceil
\frac{2nq+\log_2(2/\eta)}{n/2-\log_2(C_0q)}
\right\rceil\right\}.
\label{eq:mpsy-proxy-repetitions}
\end{equation}
Then $\varepsilon_{n,q,m}\leq\eta/2$ and
\begin{equation}
m\leq1+8q+\frac{4}{n}\log_2(2/\eta).
\label{eq:mpsy-proxy-repetitions-simple}
\end{equation}

For clarity, the relative-error convention of Ref.~\cite{MPSY24} places the sampled ensemble on the outside of the comparison. Writing $\mathcal D$ for the distribution of $D$, its Choi form is
\begin{equation}
(1-\varepsilon_{n,q,m})\Omega_{\mathcal D}^{(q)}
\preceq\Omega_{\Haar}^{(q)}
\preceq(1+\varepsilon_{n,q,m})\Omega_{\mathcal D}^{(q)}.
\label{eq:mpsy-relative-convention}
\end{equation}
Equivalently,
\begin{equation}
\frac{\Omega_{\Haar}^{(q)}}{1+\varepsilon_{n,q,m}}
\preceq\Omega_{\mathcal D}^{(q)}
\preceq\frac{\Omega_{\Haar}^{(q)}}{1-\varepsilon_{n,q,m}}.
\label{eq:mpsy-relative-inverted}
\end{equation}
Thus Eq.~\eqref{eq:relative-choi-order} holds with
\begin{equation}
\delta_{n,q,m}:=\frac{\varepsilon_{n,q,m}}{1-\varepsilon_{n,q,m}}\leq\eta.
\label{eq:mpsy-convention-conversion}
\end{equation}
Lemma~\ref{lem:relative-design-proxy} proves the proxy guarantee. The circuit consists of $m$ independent samples from $\nu_{n,q}$, so Eq.~\eqref{eq:mpsy-proxy-repetitions-simple} gives sampling time and circuit size polynomial in $n,q,\log(1/\eta)$.
\end{proof}

\subsubsection{Bidirectional proxies}
\label{app:bidirectional-proxy}

To justify the strong extension, a sufficiently accurate diamond-error design also controls adaptive interactions that mix forward and inverse calls. We give the conversion explicitly because the forward relative-order argument above does not directly apply to inverse slots.

Let $\mathcal M_\mu^{(q)}(X)=\mathbb E_{U\leftarrow\mu}[U^{\otimes q}X(U^\dagger)^{\otimes q}]$, and suppose
\begin{equation}
\bigl\|\mathcal M_{\mathcal D}^{(q)}-\mathcal M_{\Haar}^{(q)}\bigr\|_\diamond\leq\varepsilon.
\label{eq:app-bidirectional-design-error}
\end{equation}
Evaluating this difference on a normalized maximally entangled state gives $\|\Omega_{\mathcal D}^{(q)}-\Omega_{\Haar}^{(q)}\|_1\leq\varepsilon$. For the swap $S$ between the two registers of one Choi operator,
\begin{equation}
J(U^\dagger)=S J(U)^{\mathsf T}S^\dagger,
\label{eq:app-inverse-choi}
\end{equation}
where $\mathsf T$ denotes full transpose in the computational basis.

Pad a bidirectional interaction to $q$ calls as above, and partition its accepting event according to the classical direction string $s\in\{+,-\}^q$. Each accepting branch, kept unnormalized, defines a positive operator $T_s$ on the $q$ Choi slots. Moreover, $\|T_s\|_\infty\leq\operatorname{Tr}T_s\leq N^{2q}$: replacing every slot by the completely depolarizing channel, whose normalized Choi operator is $I/N^2$, makes $N^{-2q}\operatorname{Tr}T_s$ a probability. The averaged Choi operator for direction string $s$ is obtained from $\Omega_\mu^{(q)}$ by transposing the entire Choi slots corresponding to inverse calls and swapping their input and output registers. Denote this linear map by $\mathcal R_s$. Transposition on any collection of slots and register swaps preserve the Hilbert--Schmidt norm. Since the full Choi space has dimension $N^{2q}$,
\begin{equation}
\|\mathcal R_s(X)\|_1
\leq N^q\|\mathcal R_s(X)\|_2
=N^q\|X\|_2
\leq N^q\|X\|_1.
\end{equation}
Applying these bounds to $X=\Omega_{\mathcal D}^{(q)}-\Omega_{\Haar}^{(q)}$ and summing over the at most $2^q$ direction strings yields
\begin{equation}
|p_{\mathcal D}-p_{\Haar}|
\leq 2^q N^{3q}\varepsilon.
\label{eq:app-bidirectional-proxy-bound}
\end{equation}
The branch decomposition allows each direction to depend on all previous observations; it imposes no restriction on the intervening quantum operations or retained workspace.

For the ensemble in Eq.~\eqref{eq:mpsy-proxy-unitary}, take
\begin{equation}
m:=\max\!\left\{1,\left\lceil
\frac{3nq+q+\log_2(1/\eta)}{n/2-\log_2(C_0q)}
\right\rceil\right\}.
\label{eq:app-bidirectional-repetitions}
\end{equation}
Its diamond error is at most $\alpha_{n,q}^m\leq\eta/(2^qN^{3q})$, so Eq.~\eqref{eq:app-bidirectional-proxy-bound} proves the bidirectional $(q,\eta)$ guarantee. The denominator remains at least $n/4$, and hence $m=O(q+\log(1/\eta)/n)$. Both the sampled circuit and its inverse are efficient. Bidirectional proxies therefore have the same polynomial dependence on $n,q,\log(1/\eta)$ as forward proxies.

\subsection{Consequences for the generic construction}

\begin{proof}[Proof of Corollary~\ref{cor:classical-existence}]
Fix a nonuniform QPT adversary and a polynomial $q(\lambda)$ bounding its calls on every execution. Since $n(\lambda)=\omega(\log\lambda)$, the condition $C_0q(\lambda)\leq2^{n(\lambda)/4}$ holds for all sufficiently large $\lambda$. Proposition~\ref{prop:efficient-haar-proxies} with $\eta=2^{-\lambda}$ provides an efficient proxy for which $2q\eta$ is negligible. The PRF term in Eq.~\eqref{eq:classical-bound} is negligible by assumption. The sum of PRU advantages is also negligible: otherwise, a nonuniform distinguisher could hardwire an index $j=j(\lambda)$ maximizing that advantage and achieve at least $1/q(\lambda)$ of the sum, contradicting single-instance PRU security. The circuits for all these reductions have a common polynomial size bound. This proves the corollary.
\end{proof}

\begin{proof}[Proof of Corollary~\ref{cor:classical-strong}]
Choose the efficient bidirectional proxy above with $\eta=2^{-\lambda}$ and $q(\lambda)$ bounding all forward and inverse calls. Equation~\eqref{eq:app-strong-classical-bound}, post-quantum PRF security, and strong PRU security give the result. The polynomial sum is handled by the same nonuniform maximizing-index argument.
\end{proof}

In both proofs, the proxy is chosen after fixing the adversary's polynomial query bound. Its parameters affect only the reduction: the construction $G_k(x)=V_{F_k(x)}$ and its key-generation algorithm remain independent of $q$ and $\eta$.

\subsection{Security of the PFC instantiation}
\label{app:classical-pfc}

We now prove Corollary~\ref{cor:classical-pfc} using the inner family in Eq.~\eqref{eq:pfc-inner-pru}. First assume that the sampler produces an exact unitary $2$-design. For an inner distinguisher $\mathcal B$ making at most $q$ forward calls, consider three experiments: $\mathsf H_0$ uses the keyed PFC unitary, $\mathsf H_1$ replaces the keyed Boolean function by a uniformly random function $h$, and $\mathsf H_2$ also replaces the keyed permutation by a uniformly random permutation $\pi$. Every sampled object is held fixed throughout the interaction.

The transition from $\mathsf H_0$ to $\mathsf H_1$ is a quantum-PRF reduction. It samples the permutation key and $2$-design element locally and implements each phase operation from one query to its Boolean-function challenge by phase kickback. Thus the reduction $\mathcal B_\Phi$ makes at most $q$ challenge queries.

For the permutation transition, the reduction must implement the random phase function efficiently. A $2q$-wise independent Boolean function suffices exactly: after fixing all independent side randomness, the acceptance probability of a computation with at most $q$ phase queries is a polynomial of degree at most $2q$ in the signs $(-1)^{h(z)}$. Its expectation therefore agrees under a uniform function and under a $2q$-wise independent function. This remains true for adaptive computations by deferring intermediate measurements. The permutation reduction may consequently sample the latter function and implement its phase oracle locally.

With the strong quantum-PRP interface used in Section~\ref{subsec:classical-pfc}, an in-place permutation can be implemented by computing $\pi(z)$ into a clean register, swapping the two registers, and uncomputing the old input using $\pi^{-1}$. Each local permutation therefore costs one forward and one inverse function-oracle query. The resulting distinguisher $\mathcal B_{\mathsf P}$ makes at most $2q$ total challenge queries. Combining the two computational transitions with the adaptive PFC bound of Ref.~\cite{MaHuang25} gives
\begin{align}
\Adv_{V^{\PFC}}^{\mathsf{pru},\rightarrow}(\mathcal B;\lambda)
&\leq\Adv_\pi^{\mathsf{sqprp}}(\mathcal B_{\mathsf P};\lambda)
+\Adv_h^{\mathsf{qprf}}(\mathcal B_\Phi;\lambda)
+\frac{4q(q-1)}{2^{n(\lambda)}+1}.
\label{eq:pfc-inner-bound}
\end{align}
Here the first advantage uses the ordinary single-instance strong quantum-PRP experiment, and the second uses the single-instance quantum-PRF experiment, both with coherent data queries. The PFC bound applies directly for $q\leq2^{n(\lambda)}$; beyond that range the displayed statistical term already exceeds one, so the inequality remains valid.

Applying Eq.~\eqref{eq:pfc-inner-bound} to the inner distinguishers in Theorem~\ref{thm:classical-lifting} yields
\begin{align}
\Adv_{G^{\PFC}}^{\mathsf{cl},\rightarrow}(\mathcal A;\lambda)
&\leq\Adv_F^{\mathsf{prf}}(\mathcal B_F;\lambda)
+\sum_{j=1}^{q}\Adv_\pi^{\mathsf{sqprp}}(\mathcal B_{\mathsf P,j};\lambda)\nonumber\\
&\quad+\sum_{j=1}^{q}\Adv_h^{\mathsf{qprf}}(\mathcal B_{\Phi,j};\lambda)
+\frac{4q^2(q-1)}{2^{n(\lambda)}+1}+2q\eta.
\label{eq:pfc-classical-bound}
\end{align}
If the implemented $2$-design sampler has total-variation error $\varepsilon_{\mathsf C}$, add at most $q\varepsilon_{\mathsf C}$ to this bound, once for each inner hybrid. Taking $\eta=2^{-\lambda}$, the proxy term and the PFC term are negligible when $n(\lambda)=\omega(\log\lambda)$. The primitive-security assumptions and sampler-accuracy assumption in Corollary~\ref{cor:classical-pfc} control the remaining terms, proving that corollary. Label indexing is provided by the outer PRF; the inner reductions require only single-instance function and permutation security.

\clearpage

\section{Security proof for the coherent-label construction}
\label{app:coherent}

We prove Theorems~\ref{thm:pfc-haar-closeness} and~\ref{thm:coherent-pfc-construction}, followed by Corollary~\ref{cor:coherent-existence}. Our proof uses a label-indexed version of the forward path-recording method of Ref.~\cite{MaHuang25}. Ref.~\cite[Sec.~6.2]{FoxmanEtAl26} studies the same ideal family under the name \emph{unitary Haar cipher} and gives a recording description in terms of label--input--output triples. Our label-indexed path-recording analysis was developed independently, before Ref.~\cite{FoxmanEtAl26} appeared on arXiv, and is presented here to give a self-contained proof of the indexed PFC comparison of Eq.~\eqref{eq:controlled-pfc-ensemble} and the independent Haar family of Eq.~\eqref{eq:classical-controled-unitary} are compared with the same recorder. Their difference is controlled by the probability of repeating an input at the same label. The analysis preserves interference across labels and avoids any dependence on the size of the label domain. We then replace the random phase and permutation layers by the keyed primitives from Section~\ref{subsec:coherent-primitives}.

Fix the security parameter $\lambda$. As in Section~\ref{subsec:coherent-primitives}, the label set is $\X_\lambda$, the data-basis domain is $\mathcal Z_\lambda=\{0,1\}^{n(\lambda)}$, and $N=|\mathcal Z_\lambda|$. The zero-query case is immediate, so assume $1\leq q<N$. By purification and deferred measurement, a $q$-query adversary can be represented by isometries
\[
\mathcal A_j:\mathcal H_{LA}\otimes\mathcal H_{B_{j-1}}
\longrightarrow \mathcal H_{LA}\otimes\mathcal H_{B_j},
\qquad j=1,\ldots,q,
\]
where $B_j$ contains all private workspace, reference systems, and purifying registers at round $j$. The initial state $\ket{\psi_0}_{LAB_0}$ is normalized and independent of the oracle. Consistently with Eq.~\eqref{eq:coherent-adversary-state-main}, the state immediately after the final call is
\begin{equation}
\ket{\Psi_q^W}
:=(W_{LA}\mathcal A_q)\cdots(W_{LA}\mathcal A_1)
\ket{\psi_0}_{LAB_0},
\label{eq:coh-adversary-state}
\end{equation}
with identities on the current workspace implicit. The same sampled oracle is used at every call, as in Eqs.~\eqref{eq:coherent-oracle-access} and~\eqref{eq:coherent-haar-oracle-access}. The adversarial isometries are independent of the hidden oracle choices; adaptation is represented through the workspace. Any final channel or measurement is covered by contractivity of trace distance. Fewer than $q$ queries can be handled by retaining the decision and using dummy registers for the remaining calls. Only forward oracle access is supplied.

\subsection{Indexed records and the path-recording isometry}
\label{app:indexed-path-recording}

An indexed record stores the labels, inputs, and outputs of $t$ queries as a multiset
\begin{align}
R&=\{\!\{a_1,\ldots,a_t\}\!\},\nonumber\\[-1mm]
a_i&=(x_i,z_i,y_i)\in\X_\lambda\times\mathcal Z_\lambda\times\mathcal Z_\lambda.
\label{eq:coh-indexed-record}
\end{align}
Let $R_x$ be the multiset of pairs $(z,y)$ with label $x$, and set $r_x:=|R_x|$. Write $\Dom(R_x)$ and $\Img(R_x)$ for its sets of input and output coordinates. A record is \emph{output-injective} if the outputs in each $R_x$ are distinct, and \emph{good} if its inputs are also distinct. Thus a good record specifies a partial bijection at each label. Reusing a basis index at a different label is allowed.

One record slot has Hilbert space $\mathcal H_{\mathrm{rec}}:=\mathbb C^{\X_\lambda}\otimes\mathbb C^N\otimes\mathbb C^N$, whose three factors store the label, input, and output. For the fixed query bound $q$, take
\[
\mathcal H_E:=\bigoplus_{t=0}^{q}\operatorname{Sym}^t(\mathcal H_{\mathrm{rec}}).
\]
The length-zero sector is spanned by $\ket{\varnothing}$. Choose any ordering $(a_1,\ldots,a_t)$ of the entries of $R$, and let $m_a$ be the multiplicity of the triple $a$. The corresponding normalized symmetric state is
\begin{equation}
\ket{R}:=\frac{1}{\sqrt{t!\prod_a m_a!}}\sum_{\sigma\in\mathsf S_t}\ket{a_{\sigma(1)},\ldots,a_{\sigma(t)}}.
\label{eq:coh-relation-state}
\end{equation}
Here $\mathsf S_t$ permutes the $t$ slots. These states form an orthonormal basis indexed by multisets, independent of the chosen ordering. For $t<q$, define the recording map on output-injective records by
\begin{align}
V\ket{x,z}_{LA}\ket{R}_{E}
&:=\frac{1}{\sqrt{N-r_x}}\sum_{y\notin\Img(R_x)}\ket{x,y}_{LA}\nonumber\\[-1mm]
&\hspace{22mm}\otimes\ket{R\uplus\{(x,z,y)\}}_{E},
\label{eq:coh-path-oracle}
\end{align}
where $\uplus$ denotes multiset union. Only outputs already used at label $x$ are excluded, leaving $N-r_x$ choices. The query label is unchanged and the record register is hidden from the adversary.

\begin{lemma}[Indexed path isometry]
\label{lem:coh-path-isometry}
The map in Eq.~\eqref{eq:coh-path-oracle} is an isometry from each output-injective length-$t$ sector to the output-injective length-$(t+1)$ sector, for $0\leq t<q<N$.
\end{lemma}

\begin{proof}
Every basis state is mapped to a unit vector. Suppose that an output term arising from $\ket{x,z,R}$ overlaps one arising from $\ket{x',z',R'}$. Equality of the query registers gives $x=x'$ and the same newly chosen output $y$. Because $y$ occurs in neither $\Img(R_x)$ nor $\Img(R'_x)$ before the call, the appended triple is the unique triple with label $x$ and output $y$ in each enlarged record. Equality of the enlarged records then implies $z=z'$ and $R=R'$. Hence images of distinct basis states are orthogonal.
\end{proof}

Starting from $\ket{\varnothing}$, repeated applications of $V$ remain in the output-injective sectors. This recorder is an information-theoretic comparison object; it need not have an efficient implementation.

\subsection{Right-unitary invariance and the collision bound}

Let $\bm D=(D_x)_{x\in\X_\lambda}$ be any family of data-register unitaries and define the corresponding controlled unitary
\[
\widehat D:=\sum_{x\in\X_\lambda}\ketbra{x}{x}_L\otimes D_x.
\]
The path-recording experiment produces the state
\begin{equation}
\ket{\Psi_q^{V\widehat D}}:=(V\widehat D\mathcal A_q)\cdots(V\widehat D\mathcal A_1)\ket{\psi_0}_{LAB_0}\ket{\varnothing}_{E}.
\label{eq:coh-path-state-D}
\end{equation}
We abbreviate the case $\widehat D=I_{LA}$ by $\ket{\Psi_q^V}$. On one record slot define
\[
\begin{aligned}
T(\bm D)&:=\sum_{x\in\X_\lambda}\ketbra{x}{x}\otimes D_x\otimes I_N,\\
\Gamma_t(\bm D)&:=T(\bm D)^{\otimes t}\big|_{\operatorname{Sym}^t(\mathcal H_{\mathrm{rec}})}.
\end{aligned}
\]
The tensor-power action preserves the symmetric subspace and is unitary there. It also preserves output injectivity, because it changes only recorded input coordinates, not labels or outputs. Set $\Gamma_0(\bm D)=I$.

\begin{lemma}[Invariance under label-controlled input unitaries]
\label{lem:coh-right-invariance}
For every family $\bm D$ of data-register unitaries,
\begin{equation}
\ket{\Psi_q^{V\widehat D}}=\Gamma_q(\bm D)_{E}\ket{\Psi_q^{V}}.
\label{eq:coh-right-invariance}
\end{equation}
Consequently, the reduced state on the adversary's registers is independent of $\bm D$:
\[
\tr_E\ketbra{\Psi_q^{V\widehat D}}{\Psi_q^{V\widehat D}}
=\tr_E\ketbra{\Psi_q^V}{\Psi_q^V}.
\]
\end{lemma}

\begin{proof}
On an output-injective length-$t$ record, direct expansion of Eq.~\eqref{eq:coh-path-oracle} gives
\begin{equation}
V\widehat D_{LA}\bigl(I_{LA}\otimes\Gamma_t(\bm D)_{E}\bigr)=\bigl(I_{LA}\otimes\Gamma_{t+1}(\bm D)_{E}\bigr)V.
\label{eq:coh-intertwining}
\end{equation}
Both sides apply $D_x$ to the newly recorded input and $D_{x_i}$ to each preceding input with label $x_i$. Labels and outputs are unchanged, so the available outputs and the factors $N-r_x$ agree. Distinct label--output pairs ensure that no additional multiplicity factors arise. Since every $\mathcal A_j$ acts trivially on $E$, induction gives Eq.~\eqref{eq:coh-right-invariance}. Tracing out $E$ proves the reduced-state identity.
\end{proof}

Let $\Pi_{\mathrm{out}}^{(q)}$ and $\Pi_{\mathrm{good}}^{(q)}$ denote the projectors onto the output-injective and good length-$q$ record sectors. Embed these sectors in $\mathcal H_{\mathrm{rec}}^{\otimes q}$, with slot $j$ decomposed as $L_jZ_jY_j$ for its label, input, and output. On this ambient tensor product, define the projector
\begin{equation}
Q_{jk}:=\sum_{x\in\X_\lambda}\ketbra{x,x}{x,x}_{L_jL_k}\otimes\sum_{z\in\mathcal Z_\lambda}\ketbra{z,z}{z,z}_{Z_jZ_k}.
\label{eq:coh-equality-projector}
\end{equation}
Identities on all unlisted factors, including the recorded outputs, are implicit. An individual $Q_{jk}$ need not preserve the symmetric subspace, but the sum over all pairs does. On the output-injective symmetric sector, every nongood record has equal label--input pairs in at least two positions. Thus, as an inequality of quadratic forms on this sector,
\begin{equation}
\Pi_{\mathrm{out}}^{(q)}-\Pi_{\mathrm{good}}^{(q)}\preceq\sum_{1\leq j<k\leq q}Q_{jk}.
\label{eq:coh-operator-union}
\end{equation}

Now sample $\bm C=(C_x)_{x\in\X_\lambda}$ independently of the adversary's initial state and operations, with each marginal $C_x$ an exact unitary $2$-design. Correlations between different $C_x$ are allowed. The second-moment identity on two data-register copies is
\begin{align}
&\mathbb E_{C_x}\!\left[(C_x^\dagger)^{\otimes2}\left(\sum_z\ketbra{z,z}{z,z}\right)C_x^{\otimes2}\right]\nonumber\\[-1mm]
&\hspace{18mm}=\frac{2}{N+1}\Pi_{\mathrm{sym}}^{N,2}\preceq\frac{2}{N+1}I,
\label{eq:coh-two-design-twirl}
\end{align}
where $\Pi_{\mathrm{sym}}^{N,2}$ projects onto $\operatorname{Sym}^2(\mathbb C^N)$. Define $\widehat C:=\sum_x\ketbra{x}{x}_L\otimes C_x$. Since $Q_{jk}$ already projects onto equal labels, only one marginal $C_x$ enters each summand, and Eq.~\eqref{eq:coh-two-design-twirl} implies
\begin{align}
&\mathbb E_{\bm C}\!\left[(\widehat C_j^\dagger\otimes\widehat C_k^\dagger)Q_{jk}(\widehat C_j\otimes\widehat C_k)\right]\nonumber\\[-1mm]
&\hspace{43mm}\preceq\frac{2}{N+1}I.
\label{eq:coh-block-twirl}
\end{align}
In Eq.~\eqref{eq:coh-block-twirl}, $\widehat C_j$ acts on the recorded label and input of slot $j$. The state $\ket{\Psi_q^V}$ is independent of $\bm C$ and is output-injective. Moving $\widehat C$ to the record using Lemma~\ref{lem:coh-right-invariance}, and applying Eqs.~\eqref{eq:coh-operator-union} and~\eqref{eq:coh-block-twirl}, therefore gives
\begin{align}
1-\mathbb E_{\bm C}\!\left[\bra{\Psi_q^{V\widehat C}}\Pi_{\mathrm{good}}^{(q)}\ket{\Psi_q^{V\widehat C}}\right]
&\leq\binom q2\frac{2}{N+1}\nonumber\\[-1mm]
&=:\beta_q=\frac{q(q-1)}{N+1}.
\label{eq:coh-bad-weight}
\end{align}
The estimate uses only the second moment at one label at a time. It therefore applies both to the shared choice $C_x=C$ in Eq.~\eqref{eq:controlled-pfc-ensemble} and to independent Haar choices, with no factor depending on $|\X_\lambda|$.

\subsection{Exact comparison on good records}

We next purify the randomness in the permutation and phase layers of Eq.~\eqref{eq:controlled-pfc-ensemble}. Let the hidden registers $P$ and $F$ contain a permutation $\pi_x\in\mathsf{Perm}(\mathcal Z_\lambda)$ for every label $x$ and a Boolean function $f:\X_\lambda\times\mathcal Z_\lambda\to\{0,1\}$, respectively. Initialize them in the uniform superposition
\begin{equation}
\ket{\Omega}_{PF}:=\frac{1}{\sqrt{(N!)^{|\X_\lambda|}2^{N|\X_\lambda|}}}\sum_{\bm\pi,f}\ket{\bm\pi,f}_{PF},
\label{eq:coh-pf-initial}
\end{equation}
and define the purified permutation--phase oracle by
\begin{align}
&\mathsf{pfO}\bigl(\ket{x,z}_{LA}\otimes\ket{\bm\pi,f}_{PF}\bigr)\nonumber\\[-1mm]
&\hspace{8mm}:=(-1)^{f(x,z)}\ket{x,\pi_x(z)}_{LA}\otimes\ket{\bm\pi,f}_{PF}.
\label{eq:coh-pf-oracle}
\end{align}
Here $\bm\pi$ ranges over all permutation families and $f$ over all Boolean functions on $\X_\lambda\times\mathcal Z_\lambda$. Tracing out the hidden registers $PF$ recovers exactly the experiment with a persistent random function and independent persistent random permutations.

For a length-$t$ record $R$, let $\delta_{\bm\pi,R}=1$ if $\pi_x(z)=y$ for every occurrence of $(x,z,y)$ in $R$, and let $\delta_{\bm\pi,R}=0$ otherwise. The unnormalized permutation--phase state associated with $R$ is
\begin{align}
\ket{\widetilde{\mathsf{pf}}_R}
&:=\frac{1}{\sqrt{(N!)^{|\X_\lambda|}}}\sum_{\bm\pi}\delta_{\bm\pi,R}\ket{\bm\pi}\nonumber\\[-1mm]
&\quad\otimes\frac{1}{\sqrt{2^{N|\X_\lambda|}}}\sum_f
(-1)^{\sum_{i=1}^{t}f(x_i,z_i)}\ket{f}.
\label{eq:coh-algebraic-pf-vector}
\end{align}
The exponent counts every occurrence in $R$. Incompatible permutation constraints give the zero vector; repeated queries to the same input contribute with their multiplicities.

For a good record $R$, define the normalized permutation--phase relation state
\begin{align}
\ket{\mathsf{pf}_R}
&:=\frac{1}{\sqrt{\prod_x(N-r_x)!}}\sum_{\bm\pi}\delta_{\bm\pi,R}\ket{\bm\pi}\nonumber\\[-1mm]
&\quad\otimes\frac{1}{\sqrt{2^{N|\X_\lambda|}}}\sum_f
(-1)^{\sum_{i=1}^{t}f(x_i,z_i)}\ket{f}.
\label{eq:coh-normalized-pf-vector}
\end{align}
For each label $x$, exactly $(N-r_x)!$ permutations extend the partial bijection $R_x$. Hence there are $\prod_x(N-r_x)!$ compatible families $\bm\pi$, and Eq.~\eqref{eq:coh-normalized-pf-vector} has unit norm. Moreover,
\begin{equation}
\ket{\widetilde{\mathsf{pf}}_R}=\alpha_R\ket{\mathsf{pf}_R},\qquad \alpha_R:=\prod_{x\in\X_\lambda}\sqrt{\frac{(N-r_x)!}{N!}}.
\label{eq:coh-pf-scaling}
\end{equation}

\begin{lemma}[Orthogonality of good relation states]
\label{lem:coh-pf-orthogonality}
The vectors $\{\ket{\mathsf{pf}_R}:R\text{ good}\}$ are orthonormal.
\end{lemma}

\begin{proof}
The function factor in Eq.~\eqref{eq:coh-normalized-pf-vector} is the Fourier-basis vector indexed by the label--input domain $\{(x,z):(x,z,y)\in R\text{ for some }y\}$. Thus records with different domains have orthogonal function factors. If two good records have the same domain but define different partial maps, their permutation factors have disjoint support. The normalization was established above.
\end{proof}

The good length-$q$ records determine a projector on $PF$ and a map from $PF$ to the record register:
\begin{align}
\widetilde\Pi_{\mathrm{good}}^{(q)}&:=\sum_{\substack{|R|=q\\R\ \mathrm{good}}}\ketbra{\mathsf{pf}_R}{\mathsf{pf}_R},\nonumber\\
\mathsf{Comp}_q&:=\sum_{\substack{|R|=q\\R\ \mathrm{good}}}\ket{R}\!\bra{\mathsf{pf}_R}.
\label{eq:coh-compression-map}
\end{align}
Lemma~\ref{lem:coh-pf-orthogonality} shows that $\widetilde\Pi_{\mathrm{good}}^{(q)}$ is a projector and $\mathsf{Comp}_q$ is an isometry on its range.

\begin{lemma}[Good-sector selection]
\label{lem:coh-good-selection}
For every length-$q$ record $R$,
\begin{equation}
\widetilde\Pi_{\mathrm{good}}^{(q)}\ket{\widetilde{\mathsf{pf}}_R}=\begin{cases}\alpha_R\ket{\mathsf{pf}_R},&R\text{ good},\\0,&R\text{ not good}.\end{cases}
\label{eq:coh-good-selection}
\end{equation}
\end{lemma}

\begin{proof}
If a record assigns different outputs to one input, or the same output to different inputs at one label, no permutation extends it and Eq.~\eqref{eq:coh-algebraic-pf-vector} vanishes. Otherwise, removing repeated copies leaves a partial bijection at every label. For each pair $(x,z)$, let $m_{x,z}$ be its multiplicity in $R$ and define $S(R):=\{(x,z):m_{x,z}\text{ is odd}\}$. The function factor in Eq.~\eqref{eq:coh-algebraic-pf-vector} is the Fourier vector indexed by $S(R)$. If such a consistent record of length $q$ is not good, it contains a repeated triple and therefore $|S(R)|<q$. Every good length-$q$ record has a label--input domain of size $q$, so its function factor is orthogonal to that of $R$. This proves the second case of Eq.~\eqref{eq:coh-good-selection}; the first is Eq.~\eqref{eq:coh-pf-scaling}.
\end{proof}

For a controlled unitary $\widehat D=\sum_x\ketbra{x}{x}\otimes D_x$, define
\begin{equation}
\ket{\Psi_q^{\mathsf{pfO}\widehat D}}:=(\mathsf{pfO}\widehat D\mathcal A_q)\cdots(\mathsf{pfO}\widehat D\mathcal A_1)\ket{\psi_0}_{LAB_0}\ket{\Omega}_{PF}.
\label{eq:coh-pf-adversary-state}
\end{equation}
In what follows, identity operators on registers complementary to the displayed projectors and isometries are suppressed.

\begin{lemma}[Compression identity]
\label{lem:coh-compression}
For every controlled unitary $\widehat D=\sum_x\ketbra{x}{x}\otimes D_x$,
\begin{equation}
(I_{LAB_q}\otimes\mathsf{Comp}_q)\widetilde\Pi_{\mathrm{good}}^{(q)}
\ket{\Psi_q^{\mathsf{pfO}\widehat D}}
=\Pi_{\mathrm{good}}^{(q)}\ket{\Psi_q^{V\widehat D}}.
\label{eq:coh-compression-identity}
\end{equation}
\end{lemma}

\begin{proof}
Expand both experiments over the computational-basis labels, inputs, and outputs at successive queries. Each ordered history determines a multiset $R$ and the same product of matrix elements of $\widehat D$ and the adversarial isometries. In the purified experiment, its hidden vector is Eq.~\eqref{eq:coh-algebraic-pf-vector}. Lemma~\ref{lem:coh-good-selection} removes precisely the nongood length-$q$ records. A history ending in a good record has good prefixes at every round, and the recorder assigns it the normalization
\[
\prod_{x\in\X_\lambda}\sqrt{\frac{(N-r_x)!}{N!}}=\alpha_R.
\]
Equation~\eqref{eq:coh-pf-scaling} gives exactly the same coefficient in the purified experiment, and $\mathsf{Comp}_q\ket{\mathsf{pf}_R}=\ket R$. The projected expansions therefore agree term by term. Summing their amplitudes proves Eq.~\eqref{eq:coh-compression-identity}, including interference between histories with the same record. Both the labels and the records remain coherent throughout.
\end{proof}

\subsection{Statistical comparison with the Haar experiment}

The projections above act only on hidden registers. The following identity converts their discarded probability directly into a trace-norm bound.

\begin{lemma}[Projection on a hidden environment]
\label{lem:coh-hidden-projection}
Let $\rho_{SE}$ be a normalized state, let $Q_E$ be a projector on the hidden register, and write $Q:=I_S\otimes Q_E$. Then
\begin{equation}
\left\|\tr_E(\rho)-\tr_E(Q\rho Q)\right\|_1=1-\operatorname{Tr}(Q\rho).
\label{eq:coh-hidden-projection}
\end{equation}
\end{lemma}

\begin{proof}
Cyclicity of the partial trace on $E$ eliminates the two cross terms, so $\tr_E(\rho)-\tr_E(Q\rho Q)=\tr_E((I-Q)\rho(I-Q))\succeq0$. The trace norm of this positive operator equals its trace, which is $1-\operatorname{Tr}(Q\rho)$.
\end{proof}

To compare both target ensembles with the same recorder, allow a possibly different input unitary at each label and define
\begin{equation}
W_{\PFC(\bm C)}:=\sum_{x\in\X_\lambda}\ketbra{x}{x}\otimes P_{\pi_x}\Phi_{f_x}C_x,
\label{eq:coh-general-indexed-pfc}
\end{equation}
Here $\bm C=(C_x)_{x\in\X_\lambda}$ has a specified joint distribution, $f$ is uniformly random on $\X_\lambda\times\mathcal Z_\lambda$, and the $\pi_x$ are independent uniformly random permutations; the function, permutation family, and $\bm C$ are sampled independently. Let $\rho_{\PFC(\bm C)}^{\mathcal A}$ be the adversary's final state averaged over these choices. The reference state on $LAB_q$ is
\[
\rho_V^{\mathcal A}:=\tr_E\ketbra{\Psi_q^V}{\Psi_q^V}.
\]

\begin{lemma}[Comparison with the indexed recorder]
\label{lem:coh-pfc-to-recorder}
If every marginal $C_x$ is an exact unitary $2$-design, with arbitrary correlations among labels, then
\begin{equation}
\left\|\rho_{\PFC(\bm C)}^{\mathcal A}-\rho_V^{\mathcal A}\right\|_1
\leq2\beta_q=\frac{2q(q-1)}{N+1}.
\label{eq:coh-pfc-to-recorder}
\end{equation}
\end{lemma}

\begin{proof}
Define the subnormalized state on $LAB_q$
\begin{equation}
\sigma_{\mathrm{good}}:=\mathbb E_{\bm C}\!\left[\tr_E\!\left(
\Pi_{\mathrm{good}}^{(q)}
\ketbra{\Psi_q^{V\widehat C}}{\Psi_q^{V\widehat C}}
\Pi_{\mathrm{good}}^{(q)}\right)\right].
\label{eq:coh-common-good-state}
\end{equation}
The compression identity, Eq.~\eqref{eq:coh-compression-identity}, implies that the same state is obtained by projecting $PF$ with $\widetilde\Pi_{\mathrm{good}}^{(q)}$ in the purified permutation--phase experiment and tracing out $PF$. An isometry acting only on the retained hidden subspace does not change this partial trace.

By right-unitary invariance, the unprojected recorder has reduced state $\rho_V^{\mathcal A}$ for every $\bm C$. Lemma~\ref{lem:coh-hidden-projection}, applied in each experiment and then averaged, therefore gives
\begin{align}
\rho_V^{\mathcal A}-\sigma_{\mathrm{good}}&\succeq0,
\qquad
\rho_{\PFC(\bm C)}^{\mathcal A}-\sigma_{\mathrm{good}}\succeq0,\nonumber\\
\left\|\rho_V^{\mathcal A}-\sigma_{\mathrm{good}}\right\|_1
&=\left\|\rho_{\PFC(\bm C)}^{\mathcal A}-\sigma_{\mathrm{good}}\right\|_1\nonumber\\
&=1-\operatorname{Tr}(\sigma_{\mathrm{good}})\leq\beta_q.
\label{eq:coh-projected-comparison}
\end{align}
The last inequality is Eq.~\eqref{eq:coh-bad-weight}. Taking the triangle inequality through $\sigma_{\mathrm{good}}$ proves Eq.~\eqref{eq:coh-pfc-to-recorder}.
\end{proof}

\begin{proof}[Proof of Theorem~\ref{thm:pfc-haar-closeness}]
First take $C_x=C$ for every label, where $C$ is the common exact unitary $2$-design element in Eq.~\eqref{eq:controlled-pfc-ensemble}. Then $W_{\PFC(\bm C)}=W_{\mathrm{PFC}}$, and Lemma~\ref{lem:coh-pfc-to-recorder} gives
\begin{equation}
\left\|\rho_{\mathrm{PFC}}^{\mathcal A}-\rho_V^{\mathcal A}\right\|_1
\leq2\beta_q.
\label{eq:coh-pfc-path-norm}
\end{equation}
Next take $C_x=H_x$, where the $H_x$ are independent Haar-random unitaries in $\mathrm U(N)$. Every marginal is an exact unitary $2$-design, so the same lemma applies. Conditional on any fixed $f$ and $\bm\pi$, left invariance of Haar measure gives
\[
P_{\pi_x}\Phi_{f_x}H_x\sim\Haar(\mathrm U(N))
\]
for every $x$, and independence across labels is preserved. Consequently, Eq.~\eqref{eq:coh-general-indexed-pfc} has exactly the function-like Haar distribution defined in Eq.~\eqref{eq:classical-controled-unitary}, and
\begin{equation}
\left\|\rho_{\mathrm{Haar}}^{\mathcal A}-\rho_V^{\mathcal A}\right\|_1
\leq2\beta_q.
\label{eq:coh-haar-path-norm}
\end{equation}
The triangle inequality therefore yields
\[
\left\|\rho_{\mathrm{PFC}}^{\mathcal A}-\rho_{\mathrm{Haar}}^{\mathcal A}\right\|_1
\leq4\beta_q.
\]
Using $\td(\rho,\sigma)=\frac12\|\rho-\sigma\|_1$ and $\beta_q=q(q-1)/(N+1)$ proves Eq.~\eqref{eq:pfc-haar-closeness-bound}.
\end{proof}

\subsection{Computational security and the phase-oracle convention}
\label{app:coherent-computational}

We first make explicit the phase-oracle convention used in Theorem~\ref{thm:coherent-pfc-construction}. For any Boolean function $h:\X_\lambda\times\mathcal Z_\lambda\to\{0,1\}$, let $\widehat\Phi_h$ be the phase oracle defined as in Eq.~\eqref{eq:coherent-phase-oracle}. The phase-oracle advantage of $\mathcal B_\Phi$ is
\begin{align}
\Adv_g^{\mathsf{qprf}}(\mathcal B_\Phi;\lambda)
&:=\Bigl|\Pr_{k_\Phi}\!\left[\mathcal B_\Phi^{\widehat\Phi_{g_{k_\Phi}}}(1^\lambda)=1\right]\nonumber\\[-1mm]
&\hspace{9mm}-\Pr_f\!\left[\mathcal B_\Phi^{\widehat\Phi_f}(1^\lambda)=1\right]\Bigr|,
\label{eq:app-coh-phase-advantage}
\end{align}
where $k_\Phi$ is sampled by the qPRF key generator and $f$ is uniformly random on $\X_\lambda\times\mathcal Z_\lambda$. Both challenge oracles are persistent and allow adaptive queries in superposition over labels and data. This is the advantage used in Theorem~\ref{thm:coherent-pfc-construction}; the phase-kickback simulation below bounds it by standard qPRF security.

Fix the $q$-query adversary $\mathcal A$ from Theorem~\ref{thm:coherent-pfc-construction}. Let $p_i$ denote its acceptance probability in the following four persistent-oracle experiments:
\begin{description}
\item[$\mathsf C_0$.] Use the real construction $W_k^{\mathsf{coh}}$ from Eq.~\eqref{eq:coherent-pfc-construction}.
\item[$\mathsf C_1$.] Replace $g_{k_\Phi}$ by a uniformly random function $f:\X_\lambda\times\mathcal Z_\lambda\to\{0,1\}$, leaving $\Pi_{k_P}$ and $C$ unchanged.
\item[$\mathsf C_2$.] In addition, replace $\Pi_{k_P}$ by a persistent family $\bm\pi=(\pi_x)_{x\in\X_\lambda}$ of independent uniformly random permutations. This is the ideal indexed PFC experiment of Eq.~\eqref{eq:controlled-pfc-ensemble}.
\item[$\mathsf C_3$.] Use the function-like Haar oracle $W_{\mathrm{Haar}}$ from Eq.~\eqref{eq:classical-controled-unitary}.
\end{description}

For the first transition, the phase-oracle distinguisher $\mathcal B_\Phi$ samples $k_P$ and $C$ independently and then runs $\mathcal A$. Whenever $\mathcal A$ makes a coherent-PRFU query, $\mathcal B_\Phi$ applies $I_L\otimes C$, invokes its phase-oracle challenge once on $LA$, and applies the locally implemented permutation $\widehat P_{\Pi_{k_P}}$. If the challenge is $\widehat\Phi_{g_{k_\Phi}}$, the simulation is exactly $\mathsf C_0$; if the challenge is $\widehat\Phi_f$ for uniformly random $f$, it is exactly $\mathsf C_1$. Consequently,
\begin{equation}
|p_0-p_1|=\Adv_g^{\mathsf{qprf}}(\mathcal B_\Phi;\lambda),
\label{eq:coh-qprf-hybrid}
\end{equation}
and $\mathcal B_\Phi$ makes one phase-oracle query per call made by $\mathcal A$, for a total of at most $q$ queries. The sampled key and design element have efficient implementations, so $\mathcal B_\Phi$ is nonuniform QPT whenever $\mathcal A$ is.

To relate this reduction to standard qPRF security, encode $(x,z)$ as the corresponding $(d(\lambda)+n(\lambda))$-bit input. For a Boolean function $h$, the XOR function oracle is
\begin{equation}
Q_h\ket{x,z,b}=\ket{x,z,b\oplus h(x,z)},\qquad b\in\{0,1\}.
\label{eq:app-coh-xor-oracle}
\end{equation}
Preparing an auxiliary answer qubit in $\ket{-}:=(\ket0-\ket1)/\sqrt2$ gives
\begin{equation}
Q_h\bigl(\ket{x,z}\ket{-}\bigr)
=(-1)^{h(x,z)}\ket{x,z}\ket{-}.
\label{eq:app-coh-phase-kickback}
\end{equation}
The answer qubit remains unchanged and unentangled. Hence a standard qPRF distinguisher simulates every phase query exactly using one XOR-oracle query. Uniform randomness on the full bit-string domain restricts to uniform randomness on $\X_\lambda\times\mathcal Z_\lambda$, so both challenge branches are preserved. The resulting distinguisher has the same advantage and makes at most $q$ queries.

In the second transition, the permutation distinguisher must implement the random phase layer efficiently. A $2q$-wise independent function suffices. Let $\mathcal K_{2q}$ be an efficiently samplable and reversibly evaluable family of such Boolean functions on $\X_\lambda\times\mathcal Z_\lambda$.

\begin{lemma}[Limited independence for the phase layer]
\label{lem:coh-limited-independence}
Consider a quantum algorithm whose only access to $f$ is through at most $q$ calls to $\widehat\Phi_f$. Its initial state and intervening operations, including any other oracle, are independent of $f$. Then its average acceptance probability is identical for a uniformly random $f$ and for $f\leftarrow\mathcal K_{2q}$. This remains true with any additional randomness sampled independently of $f$.
\end{lemma}

\begin{proof}
Fix all independently sampled randomness, including the permutation challenge, and purify intermediate measurements. Write $s_{x,z}:=(-1)^{f(x,z)}$. After $j$ phase queries, every amplitude is a polynomial of degree at most $j$ in the signs $s_{x,z}$: an intervening isometry preserves the degree, and a phase query multiplies each basis component by one sign. The final acceptance probability therefore has degree at most $2q$. Using $s_{x,z}^2=1$, multilinearize this polynomial. Every resulting monomial involves at most $2q$ distinct inputs, and its expectation is the same under uniform and $2q$-wise independent functions. Finally average over the fixed randomness. Adaptively obtained measurement outcomes are included in the purified computation, not conditioned upon as independent randomness.
\end{proof}

For completeness, set $m:=\max\{d(\lambda)+n(\lambda),\lceil\log_2(2q)\rceil,1\}$ and injectively encode $(x,z)$ in $\mathbb F_{2^m}$. Sample the $2q$ coefficients of a polynomial of degree at most $2q-1$ uniformly and independently in this field. Evaluate it at the encoded input and apply a fixed nonzero $\mathbb F_2$-linear functional to obtain a bit. Evaluations at any set of at most $2q$ distinct inputs are independent and uniform. A field representation can be fixed as part of the nonuniform circuit family; reversible field arithmetic then evaluates the function with complexity polynomial in $d(\lambda)$, $n(\lambda)$, and $q$. Computing the bit, applying a phase, and uncomputing the workspace implements its phase oracle exactly. This family is used only by the reduction; the honest construction does not depend on the adversary's query bound.

The permutation distinguisher $\mathcal B_P$ samples $\widetilde f\leftarrow\mathcal K_{2q}$ and $C$ independently and runs $\mathcal A$. It answers each oracle call by applying $I_L\otimes C$, then the efficiently implemented phase $\widehat\Phi_{\widetilde f}$, and finally its forward in-place tweakable-permutation challenge. In the real branch, the last operation is $\widehat P_{\Pi_{k_P}}$; in the ideal branch, it is $\widehat P_{\bm\pi}$ for a persistent family of independent uniformly random permutations. Applying Lemma~\ref{lem:coh-limited-independence} separately in the two branches shows that their acceptance probabilities equal $p_1$ and $p_2$, respectively. Hence
\begin{equation}
|p_1-p_2|=\Adv_\Pi^{\ipqTPRP}(\mathcal B_P;\lambda).
\label{eq:coh-qtprp-hybrid}
\end{equation}
This is the advantage in Eq.~\eqref{eq:coherent-primitive-advantages} for the interface of Definition~\ref{def:forward-inplace-tprp}. The reduction is nonuniform QPT and uses at most $q$ forward in-place queries.

Finally, Theorem~\ref{thm:pfc-haar-closeness} gives
\begin{equation}
|p_2-p_3|\leq\frac{2q(q-1)}{N+1}.
\label{eq:coh-path-hybrid}
\end{equation}
By Eq.~\eqref{eq:coherent-construction-advantage}, the original adversary's advantage is $|p_0-p_3|$. Applying the triangle inequality to Eqs.~\eqref{eq:coh-qprf-hybrid}, \eqref{eq:coh-qtprp-hybrid}, and~\eqref{eq:coh-path-hybrid} gives Eq.~\eqref{eq:coherent-pfc-security-bound}, proving Theorem~\ref{thm:coherent-pfc-construction} with the claimed query bounds.

\begin{proof}[Proof of Corollary~\ref{cor:coherent-existence}]
Fix a nonuniform QPT adversary and a polynomial $q(\lambda)$ bounding its number of calls. Since $n(\lambda)=\omega(\log\lambda)$, one has $q(\lambda)<2^{n(\lambda)}$ for all sufficiently large $\lambda$, and the statistical term in Eq.~\eqref{eq:coherent-pfc-security-bound} is negligible. The phase-oracle advantage is negligible by the reduction to standard qPRF security in Eqs.~\eqref{eq:app-coh-xor-oracle}--\eqref{eq:app-coh-phase-kickback}, and the permutation advantage is negligible by Definition~\ref{def:forward-inplace-tprp}. Therefore Eq.~\eqref{eq:coherent-pfc-security-bound} is negligible, as required by Definition~\ref{def:coherent-prfu}.
\end{proof}

\subsection{Why the construction is forward only}

Sharing one design element $C$ across all labels restricts the construction to forward security. To see this, suppose $|\X_\lambda|\geq2$ and choose distinct labels $x\neq x'$. If the adversary had both forward and inverse access, it could implement the relative unitary
\begin{equation}
G_k^{\mathsf{coh}}(x')\bigl(G_k^{\mathsf{coh}}(x)\bigr)^\dagger
=P_{\Pi_{k_P,x'}}\Phi_{g_{k_\Phi,x'}}\Phi_{g_{k_\Phi,x}}P_{\Pi_{k_P,x}}^\dagger,
\label{eq:coh-common-C-cancellation}
\end{equation}
where the common factor $C$ cancels. The right-hand side is a signed permutation. By contrast, for independent Haar blocks $H_x$ and $H_{x'}$, the relative unitary $H_{x'}H_x^\dagger$ is Haar random. Apply the relative unitary to two copies of the same computational-basis state and then measure both outputs in that basis. The outcomes agree with probability one for a signed permutation, whereas the Haar-averaged collision probability is $2/(N+1)$. Thus neither Theorem~\ref{thm:pfc-haar-closeness} nor Theorem~\ref{thm:coherent-pfc-construction} asserts strong coherent-label security.

\clearpage

\section{Proofs for the channel and comb liftings}
\label{app:process-liftings}

We prove the lifting theorems of Section~\ref{sec:processes} by simulating each channel call or comb tooth with one call to the classical-label forward oracle of Eq.~\eqref{eq:classical-oracle-access}. The reductions preserve the full joint quantum state, including correlations with the distinguisher's workspace and with other active sessions. Throughout, register sizes and process horizons are polynomially bounded, and the register identifications specified by the public process shape are efficient.

\subsection{Haar--Stinespring channels}

\begin{proof}[Proof of Theorem~\ref{thm:channel-lifting}]
Let $\mathcal A$ be an adaptive labeled-channel distinguisher. The reduction $\mathcal B$ runs $\mathcal A$ and implements the dilation in Eq.~\eqref{eq:Haar-Stinespring} using its unitary challenge. On a query with classical label $x$ and input register $A$, it prepares a fresh $E_{\mathrm{in}}$ in $\ketbra{0}{0}$, queries the challenge on $AE_{\mathrm{in}}$ under label $x$, and interprets the result as $BE_{\mathrm{out}}$. It returns $B$ after discarding $E_{\mathrm{out}}$.

For an input possibly entangled with a reference $Z$, the real challenge implements exactly
\begin{equation}
\rho_{AZ}\longmapsto
\bigl(\Phi_{G_k(x)}\otimes\id_Z\bigr)(\rho_{AZ}).
\label{eq:app-channel-real-call}
\end{equation}
With the ideal challenge, the same procedure implements $\Phi_{H_x}\otimes\id_Z$, where the Haar unitary $H_x$ is reused whenever the label repeats. The environment is freshly initialized on every call in both experiments, as required by the channel interface. Thus every simulated query agrees with the corresponding real or ideal channel query, even when later labels and inputs depend on earlier outputs. The reduction outputs $\mathcal A$'s final bit, so its two acceptance probabilities agree separately with those of $\mathcal A$. This proves Eq.~\eqref{eq:channel-lifting-bound}, with exactly one forward unitary call per channel call.
\end{proof}

\subsection{Exact simulation of concurrent sessions}
\label{app:comb-session-simulator}

Fix a shape $\Sigma$ satisfying Eq.~\eqref{eq:process-Haar-shape-condition} and the injective label map $\ell$ used in Eq.~\eqref{eq:PRF-comb}. Let $\bm V=(V_u)_{u\in\mathcal X_\lambda}$ be any fixed family of unitaries on the common $D$-dimensional register. These unitaries define the labeled combs $(V_{\ell(x,1)},\ldots,V_{\ell(x,R)})$. Their concurrent interface is the one in Definition~\ref{def:function-like-prcomb}: $\mathsf{Open}(x)$ starts a session with fresh private memory, $\mathsf{Step}(s,X_i)$ invokes its next tooth, and $\mathsf{Abort}(s)$ terminates the session. Only the classical commands and labels are restricted; all data registers and retained workspace may be entangled.

\begin{lemma}[Exact concurrent simulation]
\label{lem:exact-concurrent-simulation}
For every concurrent process distinguisher $\mathcal T$ making at most $q$ tooth invocations, there is an oracle algorithm $\mathcal B^{\mathcal O_{\bm V}}$ that simulates this interface exactly using at most $q$ classical-label forward calls. If $\mathcal T$ is QPT, then so is $\mathcal B$.
\end{lemma}

\begin{proof}
The simulator runs $\mathcal T$ and keeps a classical table containing the process label $x_s$ and next-round index $i_s$ for each active session $s$. It also retains that session's private memory register $M_{s,i_s-1}$. On $\mathsf{Open}(x)$, it allocates a new identifier $s$, sets $x_s=x$ and $i_s=1$, prepares $M_{s,0}$ in $\ketbra{0}{0}$, and returns $s$. On $\mathsf{Step}(s,X_i)$, with $i=i_s$, it submits $X_iM_{s,i-1}$ to its unitary oracle under the classical label $\ell(x_s,i)$. The returned register is partitioned as $Y_iM_{s,i}$ according to $\Sigma$. The simulator returns $Y_i$ and, if $i<R$, retains $M_{s,i}$ and sets $i_s=i+1$. After the final round, or upon $\mathsf{Abort}(s)$, it discards the memory and removes the session from the active table. Invalid commands receive the same fixed error response in both experiments.

For every fixed $\bm V$, equality of the simulated and genuine joint states follows by induction over the command sequence. Memory initialization, discarding, and classical table updates coincide by construction. For a tooth invocation, let $Z$ contain the distinguisher's remaining workspace and all other active session memories. Both implementations apply
\begin{equation}
\sigma_{X_iM_{s,i-1}Z}\longmapsto
\bigl(V_{\ell(x_s,i)}\otimes I_Z\bigr)
\sigma_{X_iM_{s,i-1}Z}
\bigl(V_{\ell(x_s,i)}^{\dagger}\otimes I_Z\bigr).
\label{eq:app-concurrent-step}
\end{equation}
This identity holds for an arbitrary joint input state, so the induction includes correlations between sessions and adaptive measurements by the distinguisher. It therefore also holds after averaging over any distribution of $\bm V$. Each tooth invocation requires one oracle call. Since an efficient distinguisher issues only polynomially many commands, the session table and private memories require only polynomial resources.
\end{proof}

\subsection{Reusable and function-like pseudorandom combs}

\begin{proof}[Proof of Theorem~\ref{thm:prfu-prcomb-lifting}]
Apply Lemma~\ref{lem:exact-concurrent-simulation} with a single fixed process label and tooth labels $\ell_1,\ldots,\ell_R$. The command $\mathsf{Open}$ creates a new memory register for each execution. With the real PRFU challenge, the simulated tooth unitaries are $G_k(\ell_i)$, so the service is precisely $\mathcal C_k^\Sigma$ from Theorem~\ref{thm:prfu-prcomb-lifting}. With the ideal challenge, distinctness of the $\ell_i$ makes $(H_{\ell_1},\ldots,H_{\ell_R})$ a tuple of independent Haar unitaries. The tuple is fixed across sessions, while each session has fresh initial memory. This is exactly the reusable sequential-Haar experiment of Definition~\ref{def:prcomb}. The simulator outputs the tester's final bit, proving equality of the two distinguishing advantages in Eq.~\eqref{eq:single-prcomb-lifting}. It uses at most $q$ oracle calls for $q$ tooth invocations.
\end{proof}

\begin{proof}[Proof of Theorem~\ref{thm:prfu-function-like-prcomb-lifting}]
Use the simulator of Lemma~\ref{lem:exact-concurrent-simulation} with the full label map $\ell(x,i)$. Its real challenge implements $G_k(\ell(x,i))$ at tooth $i$ of the comb labeled $x$, exactly as in Eq.~\eqref{eq:PRF-comb}. In the ideal experiment, define
\begin{equation}
H_{x,i}:=H_{\ell(x,i)}.
\label{eq:app-indexed-haar-teeth}
\end{equation}
Injectivity of $\ell$ ensures that these are independent Haar unitaries for distinct pairs $(x,i)$. Thus different process labels select independent sequential-Haar combs, and every session with the same process label reuses the same tooth tuple. Fresh initialization at $\mathsf{Open}(x)$ gives each session its own private memory, as required by Definition~\ref{def:function-like-prcomb}. Consequently,
\begin{align}
\Pr_k\!\left[\mathcal T[\{\mathcal C_{k,x}^{\Sigma}\}_x]=1\right]
&=\Pr_k\!\left[\mathcal B^{\mathcal O_{G_k}}=1\right],\nonumber\\
\Pr_{\bm H}\!\left[\mathcal T[\{\mathcal H_x^{\Sigma}\}_x]=1\right]
&=\Pr_{\bm H}\!\left[\mathcal B^{\mathcal O_{\bm H}}=1\right].
\label{eq:app-concurrent-probability-identities}
\end{align}
Taking the absolute difference proves Eq.~\eqref{eq:function-like-prcomb-lifting}. The reduction makes one call per tooth invocation, with no additional loss from the horizon or the number of process labels.
\end{proof}

\subsection{Why different teeth need distinct labels}

The independence of the ideal teeth is essential. The example mentioned after Theorem~\ref{thm:prfu-prcomb-lifting} already distinguishes repeated use of one Haar unitary from a sequential-Haar comb without any private memory.

\begin{proposition}[Repeated and independent unitary teeth]
\label{prop:reused-pru-separation}
For a two-round memoryless shape with $D$-dimensional inputs and outputs, a tester distinguishes $(U,U)$, with $U\leftarrow\Haar(\mathrm U(D))$, from $(H_1,H_2)$, with independent Haar $H_1,H_2$, with acceptance-probability gap
\begin{equation}
\frac12\left(1-\frac1{D^2}\right).
\label{eq:reused-pru-separation}
\end{equation}
\end{proposition}

\begin{proof}
At each tooth, the tester supplies half of the normalized maximally entangled state $\ket{\Phi_D}=D^{-1/2}\sum_{j=1}^D\ket{j,j}$ and retains the other half. It stores the first output while invoking the second tooth, then performs a swap test on the resulting normalized Choi states $\ket{J(V)}=(V\otimes I)\ket{\Phi_D}$. Repeated use of $U$ produces identical Choi states, so the test accepts with probability one. For independent Haar teeth,
\begin{equation}
\mathbb E_{H_1,H_2}\!\left[|\braket{J(H_1)}{J(H_2)}|^2\right]
=\frac1{D^2}\mathbb E_V\!\left[|\tr V|^2\right]
=\frac1{D^2},
\label{eq:haar-choi-overlap}
\end{equation}
where $V=H_1^\dagger H_2$ is Haar. The final equality follows from $\mathbb E[V_{ij}\overline{V_{kl}}]=\delta_{ik}\delta_{jl}/D$. The averaged acceptance probability is therefore $\tfrac12(1+D^{-2})$, giving the stated gap.
\end{proof}

\section{Consequences for pseudorandom processes}
\label{sec:process-consequences}

The exact simulation above also establishes security for prefixes of a process, efficient causal transformations, and a simple prediction task at a fresh process label. All three consequences use the same classical control interface and allow arbitrary quantum workspace.

\subsection{Prefix-consistent variable horizons}
\label{subsec:variable-horizon}

Fix a public polynomial bound $R_{\max}$ and a shape $\Sigma_{\max}$ satisfying Eq.~\eqref{eq:process-Haar-shape-condition} through round $R_{\max}$. Choose distinct efficiently computable labels $\ell_1,\ldots,\ell_{R_{\max}}$. For $R\leq R_{\max}$, let $\Sigma_R$ be the truncated shape, with memory discarded after round $R$, and define
\begin{equation}
\mathcal C_k^{(R)}:=\bigl(G_k(\ell_1),\ldots,G_k(\ell_R)\bigr)
\quad\text{with shape }\Sigma_R.
\end{equation}
The labels do not depend on the selected horizon. In the ideal family, all horizons likewise use prefixes of one persistent sequence $(H_i)_{i\in[R_{\max}]}$ of independent Haar teeth.

\begin{proposition}[Variable-horizon security]
\label{prop:variable-horizon}
The family $\{\mathcal C_k^{(R)}\}_{R\in[R_{\max}]}$ is prefix-consistent: terminating any longer member after round $R$ gives $\mathcal C_k^{(R)}$. If $G$ is a classical-label forward PRFU, this family is computationally indistinguishable from the ideal prefix family under adaptive horizon selection and concurrent sessions. For every QPT distinguisher making at most $q$ total tooth invocations, there is a PRFU distinguisher making at most $q$ calls with the same advantage.
\end{proposition}

\begin{proof}
For $R\leq R'$, the first $R$ rounds of $\mathcal C_k^{(R')}$ and $\mathcal C_k^{(R)}$ use identical tooth unitaries, register identifications, and initial memory. Discarding the current memory after round $R$ therefore gives the same output state for every tester, proving prefix consistency.

To simulate joint access to the family, augment the table in Lemma~\ref{lem:exact-concurrent-simulation} with each session's selected horizon $R_s$. At round $i$, query the unitary oracle under $\ell_i$ and discard the session memory upon reaching $R_s$ or receiving $\mathsf{Abort}(s)$. The simulation is exact command by command. Its real branch implements the keyed prefix family, and its ideal branch implements prefixes of the same Haar sequence $H_i=H_{\ell_i}$. This proves the security claim and query count. The tester may also decide during an execution to stop at any earlier round, using the existing abort command.
\end{proof}

Independent labeled families of prefixes follow by replacing $\ell_i$ with an injective map $\ell(x,i)$ and including $x$ in the session's opening command. Independence then holds between process labels; different horizons at the same process label retain their common prefix.

\subsection{Closure under efficient causal transformations}
\label{subsec:causal-closure}

An efficient causal transformation $\mathcal S$ is a quantum procedure that uses a source process through its declared interface and exposes a new process or an output system. It may initialize ancillary systems, interleave interventions, and discard or measure outputs. Let $r_{\mathcal S}(q)$ bound the number of source-process invocations needed when the transformed interface is invoked at most $q$ times. For ensembles $\mathsf X,\mathsf Y$ with the same interface, write
\begin{equation}
\Adv_{\mathsf{proc}}(\mathcal D;\mathsf X,\mathsf Y;\lambda)
:=\left|\Pr[\mathcal D[\mathsf X]=1]-\Pr[\mathcal D[\mathsf Y]=1]\right|.
\end{equation}

\begin{proposition}[Causal data processing]
\label{prop:causal-data-processing}
For every efficient causal transformation $\mathcal S$ and QPT distinguisher $\mathcal D$ making at most $q$ transformed-process invocations, there is a QPT distinguisher $\mathcal B$ making at most $r_{\mathcal S}(q)$ source-process invocations such that
\begin{equation}
\Adv_{\mathsf{proc}}\bigl(\mathcal D;\mathcal S(\mathsf X),\mathcal S(\mathsf Y);\lambda\bigr)
=\Adv_{\mathsf{proc}}(\mathcal B;\mathsf X,\mathsf Y;\lambda).
\label{eq:causal-data-processing}
\end{equation}
\end{proposition}

\begin{proof}
The reduction $\mathcal B$ implements $\mathcal S$ using its challenge process, runs $\mathcal D$ on the exposed interface, and outputs $\mathcal D$'s final bit. All operations of $\mathcal S$, including its randomness and intermediate measurements, are part of this simulation. For either challenge ensemble, the resulting interaction is exactly the corresponding transformed experiment. Equality of the acceptance probabilities proves Eq.~\eqref{eq:causal-data-processing}; the invocation bound follows from the definition of $r_{\mathcal S}$.
\end{proof}

Thus efficient interventions, coarse-graining, channel extraction, and measurement preserve computational indistinguishability. The ideal comparison is obtained by applying the same transformation to the sequential-Haar process; it need not itself be a sequential-Haar ensemble of another shape.

\subsection{Prediction at a fresh process label}
\label{subsec:fresh-process-labels}

Consider a learner that makes at most $q$ training tooth invocations through the function-like comb interface and retains quantum workspace $Z$. At the end of training, it chooses a classical process label $x^\star$ that has not occurred in any $\mathsf{Open}$ command and supplies an input $X_1$ for the first tooth of a new session. The challenger initializes the private memory, applies the first tooth, and measures the output $Y_1$ in a fixed rank-one orthonormal basis. The output register and measurement outcome are withheld. The learner, using only its retained workspace and with no further process queries, guesses the outcome. The input $X_1$ may be entangled with $Z$.

If $d_Y=\dim\mathcal H_{Y_1}$ and $G$ is a classical-label forward PRFU, then there is a PRFU distinguisher $\mathcal B$ for which
\begin{equation}
p_{\mathrm{guess}}^{\mathrm{real}}
\leq\frac1{d_Y}+\Adv_G^{\mathsf{cl},\rightarrow}(\mathcal B;\lambda),
\qquad q_{\mathcal B}\leq q+1.
\label{eq:fresh-process-label-prediction}
\end{equation}

\begin{proof}
In the ideal experiment, condition on the classical training transcript and the chosen value $x^\star$. Since the label is fresh and $\ell$ is injective, the first-tooth unitary $H_{x^\star,1}$ is independent of the conditional joint state $\rho_{X_1Z}$. This remains true with adaptive label selection: the training interaction accesses only entries associated with other process labels. Set $d_M=\dim\mathcal H_{M_1}$, so $D=d_Yd_M$ by Eq.~\eqref{eq:process-Haar-shape-condition}. Haar averaging on $X_1M_0\cong Y_1M_1$ gives
\begin{equation}
\begin{split}
&\mathbb E_H\!\left[(H\otimes I_Z)
\bigl(\rho_{X_1Z}\otimes\ketbra{0}{0}_{M_0}\bigr)
(H^\dagger\otimes I_Z)\right]\\
&\hspace{12mm}=\frac{I_{Y_1}}{d_Y}\otimes\frac{I_{M_1}}{d_M}\otimes\rho_Z,
\end{split}
\label{eq:app-fresh-haar-twirl}
\end{equation}
where the register ordering on the left is understood as $X_1M_0Z$. After tracing out the private memory, the measurement outcome is uniform and independent of $Z$. Every guessing measurement therefore succeeds with probability $1/d_Y$ in the ideal experiment.

The reduction $\mathcal B$ simulates training as in Lemma~\ref{lem:exact-concurrent-simulation}, makes one additional call under $\ell(x^\star,1)$ to implement the challenge tooth, and outputs one exactly when the learner guesses the measurement outcome. Its real acceptance probability is $p_{\mathrm{guess}}^{\mathrm{real}}$ and its ideal acceptance probability is $1/d_Y$. Their difference is bounded by its PRFU advantage, proving Eq.~\eqref{eq:fresh-process-label-prediction} with at most $q+1$ calls.
\end{proof}

This prediction bound uses a first tooth with fresh initial memory and a previously unused process label. It does not assert a uniform prediction baseline for later teeth, whose private memory may carry correlations from earlier rounds.

\section{Proofs for the applications}
\label{app:applications}

This appendix gives the reductions and quantitative bounds for Section~\ref{sec:applications}. Each reduction preserves the adversary's quantum workspace and simulates the stated interface throughout the interaction. Random unitaries are sampled once and reused whenever their labels recur.

\subsection{Variable-width unitaries}
\label{app:variable-width-gluing}

We use the widths $n_m=2mw$, labels in Eq.~\eqref{eq:gluing-labels}, and circuits in Eq.~\eqref{eq:two-layer}. For a family $\mathbf U=\{U_m\}_{m\in[M]}$, the oracle $\mathcal O_{\mathbf U}$ accepts a classical width index $m$ and applies $U_m$ to the supplied $n_m$-qubit register, with identity on the adversary's remaining registers. Widths may be chosen adaptively, and a repeated width selects the same unitary. Write
\begin{align}
\Adv_{\mathsf{var}}^{\rightarrow}(\mathcal A;\lambda)
&:=\Bigl|\Pr_k\!\left[\mathcal A^{\mathcal O_{\mathbf V_k}}(1^\lambda)=1\right]\nonumber\\
&\qquad-\Pr_{\mathbf H}\!\left[\mathcal A^{\mathcal O_{\mathbf H}}(1^\lambda)=1\right]\Bigr|,
\label{eq:variable-width-advantage}
\end{align}
where the $H_m$ are independent Haar unitaries on $n_m$ qubits.

The statistical ingredient is the adaptive gluing bound of Ma and Huang~\cite[Theorem~10 and Eqs.~(C.20)--(C.21)]{MaHuang25}. When every local gate in Eq.~\eqref{eq:two-layer} is an independent Haar unitary, the averaged output state of a $t$-query forward algorithm differs from its output with a single Haar unitary on the full register by at most
\begin{equation}
\delta_m(t):=\frac{18m\,t(t-1)}{2^w}
\label{eq:app-gluing-error}
\end{equation}
in trace distance. The bound applies to a persistent unitary and arbitrary adaptive interventions, including operations on a reference system. It therefore also bounds the difference of acceptance probabilities. For $m=1$, the two ensembles coincide exactly; Eq.~\eqref{eq:app-gluing-error} remains a valid upper bound.

\begin{proof}[Proof of Theorem~\ref{thm:variable-width-pru}]
Fix a QPT adversary $\mathcal A$ making at most $q$ global oracle calls. A classical-label PRFU distinguisher $\mathcal B$ simulates a query at width $n_m$ by evaluating the two layers in Eq.~\eqref{eq:variable-width-brickwork} through its challenge oracle, in the order prescribed by Eq.~\eqref{eq:two-layer}. Each such query uses $2m-1$ PRFU calls. Hence $\mathcal B$ makes at most $q(2M-1)$ calls and is efficient for polynomially bounded $q,M,w$.

Let $p_0$ be the real acceptance probability and let $p_1$ be the acceptance probability when every local gate is replaced by its ideal PRFU counterpart. The simulation is exact in both branches, so
\begin{equation}
|p_0-p_1|=\Adv_G^{\mathsf{cl},\rightarrow}(\mathcal B;\lambda).
\label{eq:app-width-prfu-step}
\end{equation}
In the ideal branch, the local gates are mutually independent Haar unitaries because their labels are distinct. Including $n_m$ in each label also makes the resulting circuits independent across widths.

Next, consider $q+1$ statistical hybrids, indexed by $j=0,\ldots,q$. In hybrid $j$, the first $j$ distinct widths encountered by $\mathcal A$ select independent global Haar unitaries; all later new widths select independently sampled circuits of Haar local gates. Once a width is encountered, its unitary is retained for all subsequent queries. Adjacent hybrids coincide up to the first query at the $j$th new width. Condition on that common prefix and the selected width $m$. No gate for this width has yet been used, so its unitary is independent of the conditional workspace and all previously sampled width instances. The remainder of the experiment is an adaptive algorithm making at most $q$ calls to this unitary, with the other width instances forming independent side processes. Equation~\eqref{eq:app-gluing-error} bounds the change in acceptance probability by $\delta_m(q)$. If fewer than $j$ widths are queried, the two hybrids coincide throughout.

At most $q$ distinct widths occur. Writing $p_2$ for the acceptance probability with independent global Haar unitaries gives
\begin{equation}
|p_1-p_2|\leq q\max_{m\leq M}\delta_m(q).
\label{eq:app-width-statistical-step}
\end{equation}
Combining the two comparisons yields the quantitative bound
\begin{align}
\Adv_{\mathsf{var}}^{\rightarrow}(\mathcal A;\lambda)
&\leq\Adv_G^{\mathsf{cl},\rightarrow}(\mathcal B;\lambda)
+q\max_{m\leq M}\delta_m(q)\nonumber\\
&\leq\Adv_G^{\mathsf{cl},\rightarrow}(\mathcal B;\lambda)
+\frac{18Mq^2(q-1)}{2^w}.
\label{eq:variable-width-bound}
\end{align}
The first term is negligible by PRFU security. The second is negligible for polynomially bounded $M,q$ and $w=\omega(\log\lambda)$, proving the theorem.
\end{proof}

The construction uses $2m-1$ local evaluations at width $n_m$, arranged in two layers of disjoint gates. The key is sampled independently of the requested width; the public parameters $w$ and $M$ specify the supported range. These statements concern the full $n_m$-qubit register, including any public padding used to accommodate a smaller payload.

\subsection{Inverse-secure recursive gluing}
\label{app:strong-variable-width-gluing}

A strong variant uses a different circuit, based on the three-factor gluing theorem of Ananth, Bostanci, Gulati, and Lin~\cite[Theorem~7.1]{AnanthBostanciGulatiLin2025}. Suppose $G_k(x)$ is a strong classical-label PRFU on $n_0$ qubits. Choose an overlap of $b$ qubits with $n_0\geq2b$, and define the supported widths recursively by
\begin{equation}
n_{s+1}:=2n_s-b,
\qquad n_s=2^s(n_0-b)+b.
\label{eq:app-strong-width-recurrence}
\end{equation}
For one gluing step, write the larger register as $ABC$, with $|A|=|C|=n_s-b$ and $|B|=b$. Three independent Haar unitaries on $n_s$ qubits, composed as $U^{(3)}_{AB}U^{(2)}_{BC}U^{(1)}_{AB}$, emulate a Haar unitary on $ABC$ under forward and inverse queries. For at most $t$ mixed-direction calls, their statistical distinguishing advantage is bounded by
\begin{equation}
\varepsilon_s(t):=C_{\mathsf{gl}}
\left(t^2 2^{-b/2}+t^3 2^{-b}+t^3 2^{-n_s/8}\right)
\label{eq:app-strong-gluing-step}
\end{equation}
for a universal constant $C_{\mathsf{gl}}$. The alternating query convention used in the cited theorem suffices here: a strategy making $t$ calls with classically chosen directions can be simulated in at most $t$ forward--inverse pairs by routing unused calls to fresh dummy registers and discarding their outputs.

Fix a maximum recursion depth $R$. For a target depth $r\in\{0,\ldots,R\}$ and each leaf $v\in\{1,2,3\}^r$, define
\begin{equation}
T_{0,v}^{(r)}:=G_k\!\left(\operatorname{enc}(\mathsf{strong\text{-}glue},r,v)\right).
\label{eq:app-strong-tree-leaves}
\end{equation}
For $s=0,\ldots,r-1$ and $v\in\{1,2,3\}^{r-s-1}$, set
\begin{equation}
T_{s+1,v}^{(r)}
:=\bigl(T_{s,v3}^{(r)}\bigr)_{AB}
\bigl(T_{s,v2}^{(r)}\bigr)_{BC}
\bigl(T_{s,v1}^{(r)}\bigr)_{AB},
\label{eq:app-strong-tree-recursion}
\end{equation}
with the register dimensions specified above. The root $T_{r,\emptyset}^{(r)}$ acts on $n_r$ qubits and uses $3^r$ local PRFU evaluations. Its inverse is obtained by reversing the circuit and inverting every local gate. Labels distinguish all leaves and target depths, while repeated evaluations reuse them.

To analyze this family, first replace all leaves jointly by independent Haar unitaries using strong PRFU security. Let $\Delta_r(t)$ bound the statistical distinguishing advantage between the resulting depth-$r$ root and a persistent Haar unitary on $n_r$ qubits, allowing at most $t$ forward or inverse calls. At depth zero the distributions agree. At the next level, replace the three children one at a time by independent Haar unitaries and then apply Eq.~\eqref{eq:app-strong-gluing-step}. Each global call invokes each child once, so
\begin{equation}
\Delta_0(t)=0,\qquad
\Delta_{s+1}(t)\leq3\Delta_s(t)+\varepsilon_s(t).
\label{eq:app-strong-gluing-recurrence}
\end{equation}
Consequently, $\Delta_r(t)\leq\sum_{s=0}^{r-1}3^{r-1-s}\varepsilon_s(t)$.

Let $\Adv_{\mathsf{var}}^{\pm}(\mathcal A;\lambda)$ denote the difference of acceptance probabilities between classical access to the pairs $(T_{r,\emptyset}^{(r)},(T_{r,\emptyset}^{(r)})^\dagger)$ and to independent Haar pairs $(H_r,H_r^\dagger)$ at the same widths. For a $q$-query adversary, the PRFU reduction makes at most $q3^R$ calls. The first-occurrence hybrid used above then gives
\begin{equation}
\Adv_{\mathsf{var}}^{\pm}(\mathcal A;\lambda)
\leq\Adv_G^{\mathsf{cl},\pm}(\mathcal B;\lambda)
+q\max_{r\leq R}\Delta_r(q).
\label{eq:app-strong-variable-width-bound}
\end{equation}
Here the advantage uses the absolute difference of acceptance probabilities, and the direction of each query is classical. Since $n_s\geq n_0$,
\begin{align}
q\max_{r\leq R}\Delta_r(q)
=O\!\left(3^R\left(q^3 2^{-b/2}
+q^4 2^{-b}+q^4 2^{-n_0/8}\right)\right).
\label{eq:app-strong-variable-width-asymptotic}
\end{align}
For polynomially bounded $n_0,q,3^R$ and $b=\omega(\log\lambda)$, this error is negligible and all supported widths and evaluation costs are polynomially bounded. Thus the recursive construction yields a jointly secure family under forward and inverse access, as asserted after Theorem~\ref{thm:variable-width-pru}.

\subsection{Fresh-nonce rejection}
\label{app:nonce-auth-proof}

We use the encoding in Eq.~\eqref{eq:nonce-authentication} with the unitary in Eq.~\eqref{eq:nonce-auth-unitary}. The ciphertext register is $C=MTB$, with dimensions $2^m$, $2^\ell$, and $2^b$ for its three factors. Write $R:=V^{(4)}_{a_4}V^{(2)}_{a_2}$, so that $U_{k,\nu}=G_{k_G}(\operatorname{enc}(\mathbf{auth},\nu))R$. The design seeds are sampled independently of $k_G$ and fixed throughout the experiment. After applying $U_{k,\nu}^{\dagger}$, verification accepts on the projector
\begin{equation}
P_0:=I_M\otimes\ketbra{0^\ell}{0^\ell}_T\otimes I_B.
\label{eq:app-auth-accept-projector}
\end{equation}
Upon acceptance it returns $M$ and discards $T$ and $B$; upon rejection it returns a failure flag.

\begin{proof}[Proof of Proposition~\ref{prop:fresh-nonce-rejection}]
Let $\mathcal A$ make at most $q$ encryption and verification queries before submitting a final nonce $\nu^\star$ and ciphertext, which may be entangled with its workspace $Z$. By hypothesis, $\nu^\star$ has appeared in none of the preceding queries.

In the ideal strong-PRFU experiment, condition on the design seeds, the classical transcript, and the selected fresh nonce. Lazy sampling leaves $H_{\operatorname{enc}(\mathbf{auth},\nu^\star)}$ independent of the conditional joint ciphertext--workspace state $\sigma_{CZ}$. Abbreviate this Haar unitary as $H$. The full decoding unitary is $R^{\dagger}H^{\dagger}$. Since $\dim C=2^{m+\ell+b}$, Haar twirling gives
\begin{align}
&\mathbb E_H\!\left[(R^{\dagger}H^{\dagger}\otimes I_Z)\sigma_{CZ}(HR\otimes I_Z)\right] \nonumber\\
&\qquad=(R^{\dagger}\otimes I_Z)\left(\frac{I_C}{2^{m+\ell+b}}\otimes\sigma_Z\right)(R\otimes I_Z)
=\frac{I_C}{2^{m+\ell+b}}\otimes\sigma_Z.
\label{eq:app-fresh-nonce-twirl}
\end{align}
The ideal probability of acceptance is therefore
\begin{equation}
\operatorname{Tr}\!\left[P_0\frac{I_C}{2^{m+\ell+b}}\right]
=\frac{2^{m+b}}{2^{m+\ell+b}}
=2^{-\ell}.
\label{eq:app-fresh-nonce-acceptance}
\end{equation}
This equality holds for every conditioned choice and hence after averaging over the design seeds and transcript. In particular, reusing the design layers across nonces does not change the fresh-nonce bound.

Construct a strong classical-label distinguisher $\mathcal B$ that samples the design seeds independently of its challenge. To simulate encryption, it appends the tag and auxiliary state, applies $V^{(2)}_{a_2}$ followed by $V^{(4)}_{a_4}$, and makes a forward challenge call under the requested nonce label. To simulate verification, it makes an inverse challenge call, applies $(V^{(4)}_{a_4})^{\dagger}$ followed by $(V^{(2)}_{a_2})^{\dagger}$, measures the tag, and returns the prescribed output. Thus it implements $GR$ in the forward direction and $R^{\dagger}G^{\dagger}$ in the inverse direction. Finally, it verifies the submitted ciphertext at $\nu^\star$ and outputs one precisely upon acceptance. All design operations are efficient and local. The simulation uses at most $q+1$ challenge calls, with real acceptance probability $p_{\mathsf{fresh}}$ and ideal acceptance probability $2^{-\ell}$, so
\begin{equation}
\bigl|p_{\mathsf{fresh}}-2^{-\ell}\bigr|
=\Adv_G^{\mathsf{cl},\pm}(\mathcal B;\lambda).
\label{eq:app-fresh-nonce-reduction}
\end{equation}
In particular, Eq.~\eqref{eq:fresh-nonce-rejection} follows.
\end{proof}

The same argument bounds $t$ adaptively chosen attempts at mutually distinct nonces, each unused before its attempt. In the Haar experiment every attempt has conditional acceptance probability $2^{-\ell}$; a union bound and a single PRFU reduction give
\begin{equation}
\Pr[\text{at least one attempt accepts}]
\leq t2^{-\ell}+\Adv_G^{\mathsf{cl},\pm}(\mathcal B_t;\lambda).
\label{eq:app-many-fresh-nonces}
\end{equation}
Thus polynomially many fresh-nonce attempts have negligible success probability when $\ell=\omega(\log\lambda)$, including replay of a ciphertext under a new nonce.

For the inherited authentication guarantee, choose the design accuracy and register sizes as in Ref.~\cite{HaugBansalMokKohBharti2025}: take $\ell,b=\omega(\log\lambda)$ and a relative-error $\varepsilon_4$-approximate $4$-design with, for example, $\varepsilon_4=2^{-2b-(\log_2\lambda)^2}$. These moment conditions are used for authentication; the fresh-nonce proof above only uses the independence and unitarity of the design layers. The authentication guarantee cited in the main text follows by fixing a nonce $\nu$. The family $\{G_{k_G}(\operatorname{enc}(\mathbf{auth},\nu))\}_{k_G}$ is an adaptive PRU: any single-unitary distinguisher can be simulated by querying the strong PRFU at this one label. The two design seeds are independent of this PRU key and of each other. Consequently, Eq.~\eqref{eq:nonce-auth-unitary}, together with the tag and mixed auxiliary state, satisfies the construction requirements of PQAS in Ref.~\cite{HaugBansalMokKohBharti2025}. Its authentication theorem therefore applies directly at a fixed nonce, including key reuse within its stated attack model. The fresh-nonce bound above separately covers adaptive encryption and verification queries across nonces. Rejection of unchanged ciphertexts replayed at their original nonce requires stateful nonce handling.

\subsection{Content hiding for coherently masked QRAM}
\label{app:qram-proof}

The interface in Proposition~\ref{prop:qram-content-hiding} exposes only forward applications of the masked lookup in Eq.~\eqref{eq:masked-qram}. On an address register $L$, an $m$-qubit bus $B$, and an arbitrary reference $R$, it acts as
\begin{equation}
\mathcal O_{\widetilde Q_{k,D}}(\rho_{LBR})
:=(\widetilde Q_{k,D}\otimes I_R)\rho_{LBR}
(\widetilde Q_{k,D}^{\dagger}\otimes I_R).
\label{eq:masked-qram-oracle}
\end{equation}
The database and key remain fixed throughout the interaction. The two challenge databases have the same address set and word length, are independent of the key, and admit efficient lookup circuits as in Eq.~\eqref{eq:QRAM}.

\begin{proof}[Proof of Proposition~\ref{prop:qram-content-hiding}]
For $j\in\{0,1\}$, construct a coherent-label PRFU distinguisher $\mathcal B_j$ by running $\mathcal A$ and answering each query with the known lookup $Q_{D_j}$ followed by one challenge call on the address and bus registers. The address itself serves as the PRFU label, exactly as in Eq.~\eqref{eq:masked-qram}. Thus the simulation preserves superpositions of addresses and uses one challenge call per masked lookup.

In the real branch the simulated unitary is $\widetilde Q_{k,D_j}$. In the ideal branch it is
\begin{equation}
\sum_{x\in\mathcal X_\lambda}\ketbra{x}{x}_L
\otimes H_x X^{D_j(x)}_B,
\label{eq:app-ideal-masked-qram}
\end{equation}
where $X^z\ket{b}=\ket{b\oplus z}$. Right invariance of Haar measure implies $H_xX^{D_j(x)}\overset{d}{=}H_x$ for every address. Applying these fixed right multiplications to the independently sampled $H_x$ preserves their joint distribution. This is an equality of distributions over persistent controlled unitaries, so it holds for the entire adaptive interaction, including interference between different addresses.

Let $p_j$ denote $\mathcal A$'s real acceptance probability for database $D_j$, and let $p_{\mathsf H}$ denote its acceptance probability under the common ideal distribution in Eq.~\eqref{eq:app-ideal-masked-qram}. Exact simulation gives $|p_j-p_{\mathsf H}|=\Adv_G^{\mathsf{coh},\rightarrow}(\mathcal B_j;\lambda)$. Therefore
\begin{align}
|p_0-p_1|
&\leq |p_0-p_{\mathsf H}|+|p_1-p_{\mathsf H}|\nonumber\\
&=\Adv_G^{\mathsf{coh},\rightarrow}(\mathcal B_0;\lambda)
+\Adv_G^{\mathsf{coh},\rightarrow}(\mathcal B_1;\lambda).
\label{eq:app-qram-triangle}
\end{align}
Both terms are negligible by coherent-label PRFU security, proving Eq.~\eqref{eq:qram-content-hiding}.
\end{proof}

A key holder can undo the mask coherently using its inverse circuit. This decoding operation is private: the adversary receives the combined masked interface, with no separate access to the mask or its inverse. The proposition concerns one fixed database and does not assert security for updates made while reusing the same masks.

\subsection{Emulation of random multi-time dynamics}
\label{app:random-dynamics-proof}

\begin{proof}[Proof of Corollary~\ref{cor:random-multitime-emulation}]
Fix a QPT concurrent distinguisher $\mathcal T$ making at most $Q$ step calls in total. Apply the simulator of Theorem~\ref{thm:prfu-function-like-prcomb-lifting} to the process family in Eq.~\eqref{eq:random-dynamics-emulator}, using $\ell_{\mathsf{dyn}}(x,i)$ for step $i$ of instance $x$. The public shape $\Sigma$ specifies the input, output, and memory registers; the simulator initializes separate private memory for every new session and updates it with one PRFU call at each step.

In the real branch this implements $\mathcal C_{k,x}^{\Sigma}$ exactly. In the ideal branch, injectivity of $\ell_{\mathsf{dyn}}$ gives an independent Haar unitary for every pair $(x,i)$, reused across sessions of the same instance. The simulated family is therefore exactly the sequential-Haar reference family. The acceptance-probability difference equals $\Adv_G^{\mathsf{cl},\rightarrow}(\mathcal B;\lambda)$ for a reduction using at most $Q$ calls and is negligible. Adaptive stopping requires no additional reduction: an $\mathsf{Abort}$ command discards the session memory, and the labels of preceding steps do not depend on when this command is issued.
\end{proof}

\end{document}